%% file: arxiv.tex
\documentclass{article}
\usepackage{arxiv}
\usepackage{tikz}
\usetikzlibrary{shapes,backgrounds,patterns,calc}
\usepackage{booktabs} % for professional tables
\usepackage{hyperref}

\input{pkgs}
\input{defs}

\usepackage{natbib}

\title{Multi-Agent Contract Design beyond Binary Actions}

\author{
	Federico Cacciamani\\
	Politecnico di Milano\\
	\texttt{federico.cacciamani@polimi.it}
	\And
	Martino Bernasconi\\
	Bocconi University\\
	\texttt{martino.bernasconi@unibocconi.it}
	\And
	Matteo Castiglioni\\
	Politecnico di Milano\\
	\texttt{matteo.castiglioni@polimi.it}
	 \And
	Nicola Gatti\\
	Politecnico di Milano\\
	\texttt{nicola.gatti@polimi.it}
}

\begin{document}
	
\maketitle

% Abstract. Note that this must come before \maketitle.

	% this must go after the closing bracket ] following \twocolumn[ ...
	
	% This command actually creates the footnote in the first column
	% listing the affiliations and the copyright notice.
	% The command takes one argument, which is text to display at the start of the footnote.
	% The \icmlEqualContribution command is standard text for equal contribution.
	% Remove it (just {}) if you do not need this facility.
	
	%\printAffiliationsAndNotice{}  % leave blank if no need to mention equal contribution

	\input{src/abstract}

\input{src/intro}
\input{src/prelim}
%\input{src/rcvsdc}
\input{src/optcrc}

\input{src/combinatorial}
\input{src/bayesian}
\input{src/optbcrc}

\input{src/approx}

\section*{Acknowledgements}
This paper is supported by the FAIR (Future Artificial Intelligence Research) project, funded by the NextGenerationEU program within the PNRR-PE-AI scheme (M4C2, Investment 1.3, Line on Artificial Intelligence), and by the EU Horizon project ELIAS (European Lighthouse of AI for Sustainability, No. 101120237).

% Bibliography
\bibliographystyle{ACM-Reference-Format}
\bibliography{sample-bibliography}

% Appendix
\newpage
\appendix
\input{src/appendix}

\end{document}

%% file: pkgs.tex
\usepackage{pict2e,picture,graphicx}
\usepackage[overload]{empheq}
\usepackage{microtype}
\usepackage{graphicx}
\usepackage{subfigure}
\usepackage{booktabs}
\usepackage{amsmath}
\usepackage{cases}
\usepackage{mathtools}
\usepackage{amsthm}
\usepackage{thm-restate}
\usepackage{enumitem}
\usepackage{xcolor}
\usepackage{nicefrac, xfrac}
\usepackage[capitalize,noabbrev]{cleveref}
\usepackage{xparse}
\usepackage{leftindex}
\usepackage{hyperref}
\usepackage{bbm}
\usepackage{yfonts}
\usepackage{stmaryrd}
\usepackage{mathtools}
\usepackage{xspace}
\usepackage{amsfonts}
\usepackage{amssymb}
\usepackage{float}

%% file: defs.tex
\newcounter{programCounter}
\crefname{programCounter}{Program}{Programs}
\newenvironment{program}{\begin{subequations}\refstepcounter{programCounter}}{\end{subequations}}

\makeatletter
\DeclareRobustCommand{\Arrow}[1][]{%
	\check@mathfonts
	\if\relax\detokenize{#1}\relax
	\settowidth{\dimen@}{$\m@th\rightarrow$}%
	\else
	\setlength{\dimen@}{#1}%
	\fi
	\sbox\z@{\usefont{U}{lasy}{m}{n}\symbol{41}}%
	\begin{picture}(\dimen@,\ht\z@)
		\roundcap
		\put(\dimexpr\dimen@-.7\wd\z@,0){\usebox\z@}
		\put(0,\fontdimen22\textfont2){\line(1,0){\dimen@}}
	\end{picture}%
}
\makeatother
\newcommand{\shortto}{\hspace{.6mm}\scalebox{1}{\Arrow[.2cm]}\hspace{.6mm}}

\newcommand{\mc}{\mathcal}
\newcommand{\Reals}{\mathbb{R}}

\newcommand{\A}{\mathcal{A}}
\newcommand{\eps}{\varepsilon}

\newcommand{\baylp}{\ensuremath{\textsc{B-Lp}}}

\newcommand{\BAindR}{\A^\dagger_\mathsf{R}}

\newcommand{\None}{N_{\uparrow}}
\newcommand{\Ntwo}{N_{\downarrow}}
\newcommand{\poly}{\textsf{poly}}

\newcommand{\opt}{\textsc{Opt}}
\newcommand{\optcrc}{\ensuremath{\textsc{Opt}_\textsc{R}}}
\newcommand{\optdc}{\ensuremath{\textsc{Opt}_\textsc{D}}}

\newcommand{\boptcrc}{\ensuremath{\textsc{B-Opt}_\textsc{R}}}
\newcommand{\boptdc}{\ensuremath{\textsc{B-Opt}_\textsc{D}}}
\newcommand{\nollboptdc}{\ensuremath{\textsc{noLL-B-Opt}_\textsc{D}}}

\newcommand{\optsing}{\ensuremath{\textsc{Opt}_\textsc{S}}}
\newcommand{\lprc}{\ensuremath{\textsc{Lp}}}
\newcommand{\Aind}{\A^\circ}
\newcommand{\AindR}{\Aind_\mathsf{R}}
\newcommand{\lpaux}{\ensuremath{\overline{\textsc{Lp}}}}

\newcommand{\aone}{a_{\uparrow}}
\newcommand{\atwo}{a_{\downarrow}}
\newcommand{\lambdaone}{\lambda_{\uparrow}}
\newcommand{\lambdatwo}{\lambda_{\downarrow}}

\newcommand{\detcontr}{\ensuremath{\textsc{AffineContract}_\delta}\xspace}

\newcommand{\sw}{\mathsf{Sw}}
\newcommand{\vsw}{\mathsf{V}_\sw}

\makeatletter
\newcommand{\pushright}[1]{\ifmeasuring@#1\else\omit\hfill$\displaystyle#1$\fi\ignorespaces}
\newcommand{\pushleft}[1]{\ifmeasuring@#1\else\omit$\displaystyle#1$\hfill\fi\ignorespaces}
\newcommand{\specialcell}[1]{\ifmeasuring@#1\else\omit$\displaystyle#1$\ignorespaces\fi}
\makeatother

\newtheorem{corollary}{Corollary}[section]
\newtheorem{assumption}{Assumption}[section]
\newtheorem{remark}{Remark}[section]
\newtheorem{theorem}{Theorem}[section]
\newtheorem{lemma}{Lemma}[section]

\usepackage{pifont}
\usepackage[normalem]{ulem} % for \sout
\definecolor{mygreen}{rgb}{0.0, 0.5, 0.0}
\definecolor{myorange}{rgb}{1, 0.7, 0.1}

\newcommand{\ie}{\emph{i.e.}}

\newcommand{\range}[1]{\llbracket{#1}\rrbracket}

\newcommand{\B}{\mathsf{B}}

%% file: src/abstract.tex
\begin{abstract}
	We study hidden-action principal-agent problems with multiple agents. Unlike previous work, we consider a general setting in which each agent has an arbitrary number of actions, and the joint action induces outcomes according to an arbitrary distribution.
	We study two classes of mechanisms: a class of deterministic mechanisms that is the natural extension of single-agent contracts, in which the agents play a Nash equilibrium of the game induced by the contract, and a class of randomized mechanisms that is inspired by single-agent randomized contracts and correlated equilibria.
	We first analyze the effectiveness of the two classes, showing that randomized mechanisms provide an arbitrarily larger principal’s utility than deterministic ones. Then, we design polynomial-time algorithms to compute almost optimal contracts. In particular, we show that an optimal randomized mechanism does not exist, but we can find a contract arbitrarily close to the supremum in polynomial time, while an optimal deterministic contract can be computed efficiently.
	While multi-agent contracts can be computed efficiently, it remains unclear how externalities negatively affect contract design problems and the principal's utility.
	We answer this question ‘‘reducing'' multi-agent to single-agent contract design through virtual costs. To do so, we relate the utility of the multi-agent problem to a single-agent problem with combinatorial actions, essentially removing externalities. We show that the principal's utility in the multi-agent instance is at least the one of a virtual single-agent instance in which the costs are increased by an $n$ multiplicative factor, where $n$ is the number of agents.
	Leveraging this result, we obtain that, for any $\delta\ge 0$, a simple linear contract extracts a constant fraction $\delta/(1+\delta)$ of the  social welfare of a virtual  instance in which costs are increased by a $(1+\delta) n$ factor. 
	%Leveraging this result, we obtain that, for any $\ell>n$, the optimal deterministic contract extracts a constant fraction of the  social welfare of a virtual  instance in which costs are increased by a $\ell$ factor. 
	Our result is constructive and shows how to recover a multi-agent contract from a contract for the virtual single-agent instance.
	Finally, we investigate if our results carry over the more complex Bayesian setting with typed agents. We show that an almost optimal randomized contract can be computed in polynomial time.
	However, we prove that our reduction between single- and multi-agent contracts does not extend to Bayesian settings.
	Indeed, we provide an instance in which an optimal deterministic contract cannot extract any fraction of the virtual social welfare for any choice of the virtual costs.
	Nonetheless, we show that this impossibility result can be circumvented by removing limited liability.
	In particular, we design a mechanism based on affine contracts that does not guarantee limited liability but extracts exactly the same approximation of the virtual social welfare as in non-Bayesian settings.
	We conclude by showing that in Bayesian settings with single-dimensional types, this result extends to the (non-virtual) social welfare under reasonable assumptions. This generalizes the results from \citet{alon2023bayesian} to multi-agent settings.
\end{abstract}

%% file: src/intro.tex
\section{Introduction}\label{sec:intro}

Over the last years, contract design has received a growing attention as a tool to model the interaction between a \emph{principal} and one or multiple \emph{agents}. The principal’s objective is to incentivize the agents to undertake costly actions that lead to desirable outcomes.
In this work, we focus on hidden-action problems, where the principal can only observe an outcome stochastically determined by the agents’ actions, but cannot observe the actions themselves.
To steer the agents' behavior, the principal designs a \emph{contract}, that specifies outcome-dependent payments to the agents.

Principal-agent interactions are ubiquitous in real-world scenarios such as crowdsourcing platforms~\citep{ho2016adaptive}, blockchain-based smart contracts~\citep{cong2019blockchain}, and healthcare~\citep{bastani2016analysis}.
The computational study of principal-agent problems involving a single agent has been widely investigated in the literature. The vanilla version of this problem can be solved in polynomial time via multiple linear programs~ \citep{grossman1992analysis}. 
%\citet{dutting2019simple} examine the properties of linear contracts. 
This model has been extended in many directions.  For instance, some works study the Bayesian settings in which agent have types~\cite{castiglioni2022bayesian,castiglioni2023designing, alon2021contracts, alon2023bayesian,GuruganeshPower23,guruganesh2021contracts}, while other works extend the model to include combinatorial actions~\citep{dutting2021complexity,dutting2022combinatorial,dutting2024combinatorial}.

A natural extension of the single-agent model involves a scenario where a principal has to incentivize a \emph{team} of agents~\citep{holmstrom1982moral}. This problem introduces additional challenges due to the its inherently combinatorial structure.
Thus, to circumvent such difficulties, most works focus on succinct representations, binary actions, and special structures~\citep{babaioff2012combinatorial, emek2012computing,dutting2023multi,deo2024supermodular}. 
%In particular, each agent chooses only between exerting or not exerting an effort.
An exception is \citep{castiglioni23multi}, which study a setting with arbitrary action sets but with limited externalities among the agents.

Surprisingly, the general problem of incentivizing multiple agents has been largely neglected. Here, we focus on the natural extension of the single-agent principal-agent problem, where each agent has an arbitrary number of actions, and the outcome is sampled from a general probability distribution determined by the joint action of the agents.
Considering this non-succinct multi-agent problem trivializes the computational challenge of computing an optimal contract. Indeed, it can be computed in polynomial time by solving a linear problem for each action profile, similarly to single-agent contracts~\citep{grossman1992analysis}, as the number of profiles is polynomial in the instance description. 

However, outside of these preliminary computational considerations, there are many other interesting questions. In this works, we focus on the following:
\begin{enumerate}[label=(\roman*)]
	\item \emph{Are there better (efficiently computable) mechanisms than standard deterministic contracts?}
	\item \emph{Which are the guarantees of  simple (e.g., linear) contracts?}
		\item \emph{How externalizes among the agents affect the principal's utility?}
\end{enumerate}

\subsection{Our Contributions}

In this paper, we delve into contract design in multi-agent principal-agent problems. Different from previous work, we consider the general setting in which each agent has an arbitrary number of actions, and the joint action induces outcomes according to an arbitrary distribution.
We examine two classes of mechanisms: deterministic contracts and randomized contracts,
Deterministic contracts are the natural extension of single-agent principal-agent problems where agents play a Nash equilibrium of the game induced by the contract.
On the other hand, randomized contracts draw inspiration from single-agent randomized contracts and  correlated equilibria. In particular, agents receive correlated payments and action recommendations.
Even though deterministic contracts proved optimal in single-agent (non-Bayesian) problems, we show that randomization is crucial for designing optimal multi-agent contracts. In particular, we show that randomized contracts provide an arbitrarily larger principal’s utility than deterministic ones. 

Next, we analyze the computational complexity of computing (almost) optimal contracts.
While the computation of optimal deterministic contracts does not pose any additional challenge compared to single-agent ones, this is not the case for randomized contracts.
We prove that an optimal randomized mechanism does not exist, but it is possible to compute in polynomial time a contract which gives utility arbitrarily close to the supremum. Our algorithm exploits a linear relaxation of the quadratic problem that defines an optimal contract. Then, it converts a solution to the relaxed problem to an arbitrarily close-to-optimal solution of the quadratic one, \emph{i.e.}, an almost optimal randomized contract.

Despite showing that optimal contracts can be computed efficiently, it remains unclear whether it is possible to achieve good approximations via simple contracts, as it is the case in single-agent principal-agent problems. We investigate whether there are similar results in the multi-agent setting. More in general, we explore how agents' externalities affect the principal's utility. To this end, we relate the principal's utility in a multi-agent instance to their utility in a single-agent instance where action set consists of the joint (combinatorial) action.\footnote{When not specified otherwise, when we refer to the the single-agent instance, we refer to the one with combinatorial actions.}
%
%We formalize the intuition that incentivize multiple agents is harder than a single one by proving that the gap between the principal's utility in the multi-agent instance is arbitrary smaller than in the single-agent one.
We formalize the intuition that incentivizing multiple agents is more challenging than incentivizing  a single agent. We achieve this objective by proving that the principal's utility in the multi-agent instance can be arbitrarily smaller than the one in the combinatorial single-agent instance.
We quantify this intuition through virtual costs. In particular, we define a virtual instance in which costs are incremented by a multiplicative factor. We refer to this modified instance as the \emph{virtual single-agent} instance. We show that transitioning from single-agent contract design to multi-agent contract design is ``no worse'' than increasing the costs by an $n$ factor, where $n$ is the number of agents. Furthermore, we show that this factor is tight.
Leveraging this result, we obtain that, for any $\delta\ge 0$, an optimal deterministic contract extracts a constant fraction $\delta/(1+\delta)$ of the social welfare of the virtual instance in which costs are increased by a $(1+\delta) n$ factor. Interestingly, this approximation is achieved by a linear contract.
Our result is constructive and shows how to recover a multi-agent contract from a single-agent contract computed on the virtual instance.
%
%At the core of our analysis there is the relationship between the equilibria in the multi-agent problem induced by a contract and the best responses to the same contract of a single agent in the related virtual instance.
At the core of our analysis lies the relationship between the multi-agent equilibria induced by a contract and the best responses to the same contract in the related single-agent instance.
 
In the second part of the paper, we investigate whether our results carry over to the more complex Bayesian setting with typed agents. We show that our computational result on randomized contracts extends to this setting, and an almost optimal randomized contract can be computed in polynomial time.
This is in contrast with deterministic contracts, for which the problem is NP-Hard even in single-agent problems~\cite{castiglioni2022bayesian,guruganesh2021contracts}.
This further motivates the study of simple contracts, which can be computed efficiently.
We prove that our reduction between single- and multi-agent contracts does not extend to Bayesian settings.
Specifically, we provide an instance in which an optimal deterministic contract fails to extract any fraction of the virtual social welfare, regardless of the choice of virtual costs.
Nonetheless, we show that this impossibility result can be overcome by removing limited liability.
We remark that in single-agent non-Bayesian settings removing limited liability trivializes the contract design problem and simple contracts guarantee to extract all the social welfare. This is not the case in our setting. In particular, we show that, even without limited liability, the gap between the optimal principal's utility and the social welfare can be arbitrarily large.
Motivated by this result, we design a mechanism that does not guarantee limited liability but extracts an approximation of the virtual social welfare identical to the one obtained in non-Bayesian settings.
While this task is unattainable with linear contracts (which satisfy limited liability by construction), we prove that it is sufficient to employ affine contracts, \emph{i.e.}, contracts composed of a linear contract and an outcome-independent constant negative payment.
We conclude by showing that in Bayesian settings with single-dimensional types our result extends to the (non-virtual) social welfare under reasonable assumptions. This generalizes the results from \citet{alon2023bayesian} to multi-agents settings.

\subsection{Related Works}

We survey the most-related computational works on hidden-action principal-agent problems.

\paragraph{Works on Single-Agent Contract Design}

\citet{babaioff2014contract} study the complexity of contracts in terms of the number of different payments that they specify, while \cite{dutting2019simple} use the computational lens to analyze the efficiency (in terms of principal’s utility) of linear contracts.
Subsequent works focus on the more challenging Bayesian setting. \citet{guruganesh2021contracts, GuruganeshPower23} study the gap in power between different classes of contracts, while \citet{castiglioni2022bayesian} study the efficiency of linear contracts in Bayesian settings. \citet{alon2021contracts,alon2023bayesian} study one-dimensional typed
principal-agent problem where the cost scales linearly with the hidden type. 
 \citet{castiglioni2023designing} introduce the class of menus of randomized contracts and show that they can be computed efficiently.
Finally, \citet{gan2022optimal} study a generalization of the principal-agent problem, proving that optimal randomized mechanisms can be computed efficiently.
Other works study principal-agent problems in which the underlying structure is combinatorial. In particular, \citet{dutting2021complexity}  address settings in which the outcome space is combinatorial. \citet{dutting2022combinatorial} and~\cite{deo2024supermodular} focus on settings in which the agent can take  any subset of a given set of unobservable actions. \citet{dutting2022combinatorial} focus on submodular functions, while \cite{,dutting2024combinatorial} on supermodular functions.

\paragraph{Works on Multi-Agent Contract Design}

All the previous works on multi-agent contract design focus on non-Bayesian settings and succinct representations. \cite{babaioff2012combinatorial,emek2012computing} are the first to study a setting in which agents have binary actions, and each agent succeeds in his task with a probability that depends on whether they decide to undertake effort. Then, a function maps the success or failure of individual tasks to the success or failure of the project. \citet{dutting2023multi} consider a binary-action setting and provide a constant factor approximation  for submodular rewards, while \citet{deo2024supermodular} focus on supermodular setting providing both positive and negative results. \citet{castiglioni23multi} study a setting with arbitrary action sets but limited externalities among the agents. They focus on both scenarios with increasing and diminishing returns.

%% file: src/prelim.tex
\section{Preliminaries}\label{sec:prelim}
A \emph{principal-multi-agent} problem is characterized by a tuple $(N, \Omega, (A_i)_{i\in N}, (c_i)_{i\in N}, F, r)$, where $N\coloneqq\range{n}$ is a set of $n$ agents, $\Omega$ is a finite set of $m$ outcomes, $A_i$ is a finite set of actions available to each agent.\footnote{In this work, for any $n\in\mathbb{N}$, we denote as $\range{n}$ the set $\{1,\ldots,n \}$.} Moreover, we define $\ell\coloneqq\max_{i\in N}|A_i|$ as the maximum number of actions available to an agent.
Additionally, we denote as $\A \coloneqq \times_{i\in N} A_i$, the set of joint action profiles, while we let $\A_{-i}\coloneqq \times_{j\in N\setminus\{i\}}A_j$ be the set of action profiles of all agents except the $i$-th agent. For each action profile $a\in\A$, we let $a_i \in\A_{-i}$ be the action played by agent $i$, while we let $a_{-i}$ be the action profile obtained from $a$ by removing action $a_i$. Similarly, for $a_i\in A_i$ and $a_{-i}\in\A_{-i}$, we denote as $(a_i, a_{-i})\in\A$ the action profile obtained by combining $a_i$ with $a_{-i}$.
For each agent $i$, the cost function $c_i: A_i\to[0,1]$ specifies the cost of playing action $a_i\in A_i$.
Moreover, the action profile of the agents induces an outcome. Given an action profile $a\in\A$, we denote with $F_a\in\Delta(\Omega)$ the outcome distribution induced by the action profile $a$.\footnote{For any finite set $S$ we denote with $\Delta(S)$ the set of distributions over $S$.} Then, for any $a\in\A$ and $\omega\in\Omega$, we denote with $F_{a}(\omega)$ the probability of inducing $\omega$ when the action profile $a$ is played.
Finally, each outcome $\omega\in\Omega$ is associated with a reward $r_\omega\in [0,1]$ for the principal.

\subsection{Contract Design}
We investigate principal-multi-agent problems taking the perspective of the principal who is interested in maximizing their expected utility by incentivizing the agents to play action profiles that induce profitable outcomes. To this extent, the principal commits to a \emph{contract}, which specifies a \emph{payment scheme} for each agent (specifying the payments for each outcome $\omega$), and private signals (in forms of action recommendations to each agent).
%\mat{Direi perchè o non specificherei che facciamo qualcosa di diverso dagli altri. Da un certo punto di vista anche gli altri fanno action raccomandation per cordinare gli agenti in un NE. Proposta: To this extent, the principal commits to a \emph{contract}, which specifies a \emph{payment scheme} for each agents (specifying the payments for each outcome $\omega$), and private signals (in forms of action recommendations to each agent).}

%\ma{We assume \emph{limited liability} (LL), \emph{i.e.,} the principal can only transfer a non-negative amount to the agents. Trovare dove spostare} 

%Let us outline the high-level sketch of the interaction between the principal and the agents.
Now, we outline a high-level sketch of the interaction between the principal and the agents.
\begin{enumerate}[label=(\roman*)]
	\item The principal commits to a \emph{contract}, which includes action recommendations and payment schemes;
	\item The principal privately recommends to each agent $i \in N$ an action $a_i\in A_i$;
	\item Each agent $i \in N$ observes the action recommendation and plays an action $a_i^\prime\in A_i$, not necessarily following the recommendation, incurring in cost $c_i(a_i^\prime)$;	
	\item An outcome $\omega\in\Omega$ is sampled according to $F_{a^\prime}$, where $a':=(a_i')_{i\in N}$;
	\item The principal receives reward $r_\omega$ and pays the agents as specified by the payment scheme.
\end{enumerate}
While the general interaction is similar to previous works on (binary-actions) multi-agent contract design, this general model can accommodate complex randomized mechanisms.
This is fundamental, as we show that randomized action recommendations and payments schemes are more effective than deterministic ones.
On the other hand, deterministic  mechanisms are simpler to understand for the agents and might be more applicable. For this reason, in the following we will analyze the two extreme cases: fully deterministic and fully randomized mechanisms.

%Attaching action recommendations (we will consider both deterministic and randomized action recommendations) to the contract will circumvent known computational intractability issues in principal-multi-agent problems with externalities.\ma{capire cosa citare qui, e se è il posto giusto}. \mat{eliminerei, c'è tutta una letteratura sui nash puri. Si può mettere qualcosa di simile solo mooolto più avanti secondo me, tipo quando diciamo che i nash puri fanno schifo?}

Before introducing the two classes of mechanisms, we make the following common assumptions on the interaction process described above. First, as common in the contract theory \citep{carroll2015robustness}, we assume \emph{limited liability} (LL), \ie, that the principal can only transfer a non-negative amount to the agents.
%
%\mat{Io qui non parlerei di IC. i) non è un assunzione, 2) non si capisce cosa si intende, 3) ppoi tanto devi parlare dei pareggi. Scriverei solo:} 
To ease the presentation, we assume w.l.o.g.~that there exists an action $a^\circ\in A_i$ such that $c_i(a^\circ) = 0$ for all $i\in N$. This guarantees that there always exists an action with non-negative expected utility and thus the agents are always incentivised to take part in the interaction.
Formally, the existence of action $a^\circ$, together with LL, ensures that the \emph{individually rationality} (IR) of the agents is satisfied.
%ence we can ignore the \emph{individually rational} (IR) constraint that is always satisfied.
%Then we assume that the agent play actions $a_i'$ which are \emph{incentive compatible} (IC) and \emph{individually rational} (IR). Specifically we assume that actions $a_i'$ are a best response \emph{i.e.,} it maximizes the agent's expected utility given the contract and assuming that the other agents follow the recommendation of the principal. On the other hand IR means that the expected utility of such an action is positive.
%To ease the presentation, we assume w.l.o.g.~that there exists an action $a^\circ\in A_i$ such that $c_i(a^\circ) = 0$ for all $i\in N$, thus trivially including the IR constraints into the IC ones.
%\mat{io rimuoverei da precedente commento a qui. Ma vado avanti per convincermi.}

%The action selected by each agent is both \emph{incentive compatible} (IC)\ma{IC è una proprietà del meccanismo, non degli agenti}, \emph{i.e.,} it maximizes the agent's expected utility given the contract and assuming that the other agents follow the recommendation of the principal, and \emph{individually rational} (IR), \emph{i.e.,} it guarantees a non-negative utility to the agent (otherwise the agent would prefer to not participate to the game at all). To ease the presentation, we assume w.l.o.g. that there exists an action $a^\circ\in A_i$ such that $c_i(a^\circ) = 0$ for all $i\in N$, thus trivially including the IR constraints into the IC ones.

In this work, we consider two main classes of contracts: \emph{randomized contracts} and \emph{deterministic contracts}. In the remaining part of this section, we describe the two classes of contracts and, for each of them, we will describe the interaction among the principal and the agents.

\paragraph{Randomized contracts.}
A randomized-contract is specified by a tuple $(\mu, \pi)$, where $\mu\in\Delta(\A)$ is a probability distribution over action profiles (\ie,~recommendations) and $\pi=(\pi_a^i)_{(i,a)\in N\times \A}$ is a tuple of payment functions $\pi_a^i:\Omega\to\Reals^+$. The value $\pi_a^i(\omega)$  specifies the payment received by agent $i$ when the principal recommends action profile $a\in\mc A$ and the outcome achieved is $\omega$.
Given a contract $(\mu, \pi)$, we are interested in $U^{\mu, \pi}_{i}(a_i\shortto a_i^\prime)$, which represents the $i$-th agent unnormalized expected utility when they are recommended to play action $a_i\in A_i$, and they play action $a_i^\prime\in A_i$. Formally, we have:
\[
	U_{i}^{\mu,\pi}(a_i\shortto a_i')\coloneqq \sum_{a_{-i}\in\A_{-i}} {\mu(a_i, a_{-i})}\left( \sum_{\omega\in\Omega} \pi^i_{(a_i, a_{-i})}(\omega) F_{(a_i^\prime, a_{-i})}(\omega) - c_i(a_i^\prime)\right).
\]
For ease of notation we will write $U_{i}^{\mu,\pi}(a_i)$ if $a_i=a_i'$, \emph{i.e.,}, when the agent follows the recommendation.

Then, an optimal randomized contract can be found as a solution to the following optimization problem: 
\begin{program}
	\label{prog:optrc}
	\begin{align}[left =\optcrc\coloneqq\empheqlbrace]
%\begin{subnumcases}{\optcrc\coloneqq\label{prog:optrc}} 
		\sup_{\mu, \pi}&\sum_{a\in\A}\sum_{\omega\in\Omega} \mu(a) F_a(\omega) \left( r_\omega - \sum_{i\in N} \pi^i_a(\omega)\right)&\textnormal{s.t.}  \label{eq:optrc_obj}\\
%		&\sum_{a_{-i}\in\A_{-i}} \mu(a_i, a_{-i}) \left[U_{i}^{\mu,\pi}(a_i) - U_{i}^{\mu,\pi}(a_i\shortto a_i')\right] \geq 0 &\forall i\in N,\,\forall a_i,a_i^\prime\in A_i \label{eq:optrc_ic}\\
		 &U_{i}^{\mu,\pi}(a_i) \ge U_{i}^{\mu,\pi}(a_i\shortto a_i')  &\forall i\in N,\,\forall a_i,a_i^\prime\in A_i \label{eq:optrc_ic}\\
		&\pi^i_a(\omega)\geq 0&\forall i\in N,\,\forall a\in\A,\,\forall \omega\in\Omega \label{eq:optrc_1} \\
		&\mu\in\Delta(\A) &\label{eq:optrc_2}
	\end{align}
\end{program}
%
%\mat{Io non normalizzarei $U$. equazione 1b scritta cosi mi sembra totalmente controintuitiva. Quindi non normalizzarei U e scriverei 1b senza somma e $\mu$.}
%\ma{Io togliere la normalizzazione, ma lascierei la somma. Senno dobbiamo cambaire tutto. Tra l'altro la normalizzazione svanga la linearizzazione mi sa}

where the objective of \cref{prog:optrc} is the difference between the expected reward for the principal and the total expected payment given to the agents. Constraint \eqref{eq:optrc_ic} implements the incentive compatibility (IC) constraints, \emph{i.e.}, it guarantees that no agent is incentivized to deviate from the action recommendation received, assuming that all the other agents follow the recommendation. 
%We denote the optimal value of optimization problem \eqref{prog:optrc} as $\optcrc$.
%
%\fe{Aggiungere remark sul fatto che i contrattiì randomizzati sono complicati da implementare in pratica} \mat{potrei averlo già fatto prima.}
%Note that not only the recommended profile action $a\in\A$ is drawn by the principal, but also the payments $\pi_a^i$ are stochastic, as they depend on the action profile $a$. This means that each agent, one received the private recommendation $a_i$, only has a distribution over payments. More specifically the conditional distribution over payments, after having observed recommendation $a_{i}$ is $\pi^i_{a_i, a_{-i}}$ with probability $\mu(a_i, a_{-i})/\sum_{a'_{-i}\in \A_{-i}}\mu(a_i, a'_{-i})$. This makes the mechanism similar to \emph{randomized contracts} (also called lotteries) in single agent principal-agent problem \citep{castiglioni2022designing}.  \mat{non sono d'accordo con questa frase. sono totalmente diversi. Scriverei:}

Finally, we define the set of inducible actions $\AindR$ by randomized contracts. Formally:

\[
\AindR:=\left\{a: \exists (\mu, \pi)\text{ s.t. }\mu(a)>0, (\mu,\pi)\text{ is feasible for \Cref{prog:lprc}}\right\}.
\]
Intuitively, $\AindR$ includes all the action profiles that the principal can incentivize the agent to play.

\begin{remark} Randomized payments $\pi_a^i$ are different from \emph{randomized contracts} in single-agent principal-agent problems (see, \emph{e.g.,} \citep{castiglioni2023designing,GuruganeshPower23}). Indeed, in single-agent problems, after the agent receive the recommendation, deterministic payment schemes are equivalent to a randomized ones as the agent makes decisions based only on their expected utilities. This is different from what happens in multi-agent contracts, in which the randomized payments $\pi_a^i$ are not equivalent to a the single deterministic contract obtained as expectation over all possible payments. Intuitively, this is due to the correlation between agents' actions and payment schemes.
%
%	is not the case for multi-agent contracts in which the randomized payment $\pi_a^i$ (and the expected payment) depends to the other agents actions, and hence they can be more effective that deterministic ones (see \mat{abbiamo qualche esempio in cui si vede questa cosa?} for an example). \mat{tbc}
\end{remark}

%Randomized contracts are not always easy to implement, and in practice simpler mechanism might be preferred. This is why we also study the subclass of deterministic contracts.
While randomized contracts are the most general class of mechanisms, not all agents might be willing to take part to such randomized mechanism, which require an high level of rationality (on the  agent side) and are complex to implement. This motivates our inquiry on simpler contracts, namely, \emph{deterministic contracts}.
 
\paragraph{Deterministic Contracts.} A deterministic contract is defined by a tuple $(a, p)$, where $a\in\A$ is an action profile incentivized by the principal, and $p:=(p^i)_{i\in N}$ is a payment function such that $p^i:\Omega\to\Reals^+$, where $p^i(\omega)$ specifies the payment received by agent $i$ when the outcome is $\omega$. Differently than randomized contracts, in a deterministic contracts the principal deterministically incentivizes a single action profile, and thus we can drop the dependency of $p$ from the recommended action $a$. We remark that action recommendations are not strictly required in deterministic contracts and they are only used to solve the equilibrium selection problem. Equivalently, some previous works on deterministic multi-agent contracts~\citep{dutting2023multi,deo2024supermodular} solve the equilibrium selection problem by assuming that the agent selects the best Nash equilibrium for the principal.

An optimal deterministic contract can be found as a solution to the following program:
%
%\begin{subnumcases}{\optdc\coloneqq\label{prog:optdc}}
%		\sup_{a, p} \sum_{\omega\in\Omega} F_a(\omega)\left[r_\omega - \sum_{i\in N} p^i(\omega)\right]&\textnormal{s.t.} \label{eq:optdc_obj}\\
%		\phantom{\sup_{a, p}}\sum_{\omega\in\Omega} p^i(\omega)F_{a}(\omega)  -c_i(a_i)\ge \sum_{\omega\in\Omega} p^i(\omega) F_{(a_i^\prime, a_{-i})}(\omega) - c_i(a_i^\prime)&$\forall i\in N,\,\forall a_i^\prime\in A$\label{eq:optdc_ic} \\
%		\phantom{\sup_{a, p}}p^i(\omega) \geq 0&$\forall i\in N,\forall\omega\in\Omega$\\
%		\phantom{\sup_{a, p}}a\in\A,
%\end{subnumcases}

\begin{program}
	\label{prog:optdc}
	\begin{align}[left =\optdc\coloneqq\empheqlbrace]
	\max_{a, p} &\sum_{\omega\in\Omega} F_a(\omega)\left(r_\omega - \sum_{i\in N} p^i(\omega)\right)&\textnormal{s.t.} \nonumber\\
	&\sum_{\omega\in\Omega} p^i(\omega)F_{a}(\omega)  -c_i(a_i)\ge \sum_{\omega\in\Omega} p^i(\omega) F_{(a_i^\prime, a_{-i})}(\omega) - c_i(a_i^\prime)&\forall i\in N,\,\forall a_i^\prime\in A\label{eq:optdc_ic} \\
	&p^i(\omega) \geq 0&\forall i\in N,\forall\omega\in\Omega\\
	&a\in\A,
	\end{align}
\end{program}

where Constraint \eqref{eq:optdc_ic} guarantees that the contract is IC, \emph{i.e.,} that no agent has incentive to deviate from the recommended action profile assuming that the others do not deviate. 
%The optimal value attained by optimization problem \eqref{prog:optdc} is denoted as $\optdc$.
In this case, differently from randomized contracts, the payment received by each agent $i$, is stochastic only in the outcome $\omega$ and it is the natural extension of a deterministic contract for the single-agent setting. This is the class of contract commonly used in previous works on (binary-action) multi-agent contracts~\cite{dutting2023multi,deo2024supermodular}.

Given a payment scheme $p=(p^i)_{i \in N}$, we denote with $E(p)$ the set of equilibria induced by $p$. Formally,
\[ \mathsf{E}(p)= \left\{a \in \A: (a,p)\text{ is feasible for \Cref{prog:optdc}}\right\} .\]

Finally, we define the set $\A^\circ_\mathsf{D}$ the set of inducible action profiles:
%
%\[
%\A^\circ_\mathsf{D}:=\left\{a=(a_i)_{i\in N}\in\A:\exists p\ge 0,\sum_{\omega\in\Omega} p^i(\omega)F_{a}(\omega)  -c_i(a_i)\ge \sum_{\omega\in\Omega} p^i(\omega) F_{(\hat a_i, a_{-i})}(\omega) - c_i(\hat a_i) ,\forall i\in N, \hat a_i\in A_i\right\}.
%\]
\[
\A^\circ_\mathsf{D}:=\left\{a\in\A: \exists p\text{ s.t. } (a,p)\text{ is feasible for \Cref{prog:optdc}}\right\}.
\]

Intuitively, $\A^\circ_\mathsf{D} \subseteq \AindR$ includes all the action profiles that the principal can incentivize the agents to play with deterministic contracts.

\subsection{Combinatorial Contracts and Virtual costs}

In this work, we will show a non-trivial connection between multi-agent contract design and a single-agent contract design problem with combinatorial actions. In particular, we show that the two problems are connected through \emph{virtual} cost function. 

\paragraph{Combinatorial Contracts} Every multi-agent principal-agent problem $(N, \Omega, (A_i)_{i\in N}, (c_i)_{i\in N}, F, r)$ is related to a \emph{single-agent} principal-agent problem with \emph{combinatorial actions}. In this problem, there is a single agent whose set of actions is $\A$, and for each combinatorial action $a=(a_i)_{i\in N}\in \A$ the distribution over outcome is specified by $F_a$ while the cost is $c_a=\sum_{i \in N} c_i(a_i)$.
Here, a deterministic contract is a tuple $(a, p)$, where $a\in\A$ is an action recommendation and $p: \Omega\to\Reals^+$ is a payment function that determines the payment received by the agent for each possible outcome.\footnote{As for deterministic multi-agent contracts, action recommendations are used solely to model the assumption that ties are broken in favor of the principal.}
Notice that in single-agent problems (such as ours with combinatorial actions) deterministic contracts are optimal, and randomization is useful only in Bayesian setting with agent types.

The problem of computing optimal deterministic contracts for combinatorial actions can be written as follows:
\begin{program}\label{prog:single}
	\begin{align}[left =\optsing\coloneqq\empheqlbrace]
		\max_{a, p}&\sum_{\omega\in\Omega} F_a(\omega)\left[r_\omega - p(\omega)\right] &\textnormal{s.t.}\nonumber \\
		&\sum_{\omega\in\Omega} F_a(\omega)p(\omega) - c(a) \geq \sum_{\omega\in\Omega}F_{a^\prime}(\omega) p(\omega) - c(a^\prime)&\forall a^\prime\in\A  \\
		&p(\omega)\ge 0&\forall\omega\in\Omega\\
		&a\in\A.
	\end{align}
\end{program}
Notice that \Cref{prog:single} is incomparable with \Cref{prog:optdc}. Intuitively, this is due to the fact that the set of deviations is larger for \Cref{prog:single} (as the agent can deviate over multiple actions in the action profile), while the IC constraint in \Cref{prog:optdc} is stronger as it must hold for all agents $i\in N$.

Finally, given an payment scheme $p:\omega \rightarrow \mathbb{R}^+$ we denote with $\mathsf{B}(p)$ the set of action profiles that are best response to $p$. Formally, 
\[B(p)\coloneqq \left\{a \in \A: (a,p) \text{ is feasible for \Cref{prog:single}} \right\}.\]

\paragraph{Virtual Cost Function}

Given a function $\phi:\Reals^{+}\to \Reals^{+}$, for every action $a_i\in\A_i$, $i \in N$, we define the virtual costs of action $a_i$ as $\phi(c_i(a_i))$.
This will let us build virtual instances in which we can tune the extent of the costs.
In this paper, we will focus on linear virtual costs function, \ie, $\phi(x)=\alpha x$, and thus, with slight abuse of notation, we define $\phi(c(a))\coloneqq \sum_{i\in N}\phi(c_i(a_i))$ for the action profile $a=(a_i)_{i\in N}\in \A$.
Then, we define as $\optsing^\phi$ the value of \Cref{prog:single} in which each real cost $c(a)$ is replaced with the virtual cost $\phi(c(a))$. 
%
%Similarly, we define $\optsing^\alpha:=\optsing^\phi$, where $\phi(x)=\alpha x$.

We conclude the section defining other quantities related to virtual costs.
In particular, we define $\mathsf{B}^\phi$ as the set of best responses  in which each real cost $c(a)$ is replaced with the virtual cost $\phi(c(a))$. 
Moreover, we define the virtual social welfare as 
\[\vsw^\phi:=\max_{a\in\A}\left[\sum_{\omega\in\Omega} F_a(\omega)r_\omega-\phi(c(a))\right].\]
Since we will focus on linear cost functions, with abuse of notation, given an $\alpha\in \mathbb{R}^+$ we define as $\optsing^\alpha \coloneqq\optsing^\phi$, $\B^\alpha\coloneqq \B^\phi $, and $\vsw^\alpha\coloneqq \vsw^\phi$, where $\phi(x)=\alpha x$.
 Finally, we define the classical (non-virtual) social welfare as $\sw\coloneqq\vsw^1$, \ie, the sum of the utilities of the principal and the agents. 
 We remark that the (virtual and non-virtual) social welfare for the single-agent and multi-agent instance are equivalent.

%\subsection{Approximations}
%
%
%\mat{Throughout the paper, we will use to types of approximation. We say that a contract is $\eps$-optimal if it achieves principal's expected utility at least $\opt-\eps$, where $\opt$ is the utility of the optimal contract. 
%We say that that a contract is $\alpha$-approximation if it achieves principal's expected utility at least $\alpha \opt$, where $\opt$ is the utility of the optimal contract. 
%Similar conventions hold when we compare different quantities.}
%\mat{per ora un po a cazzo. Capire cosa serve nel paper.}

%% file: src/optcrc.tex
\section{The Power of Randomized Contracts}\label{sec:optcrc}
This section is dedicated to the study of randomized contracts. In the first part of the section, we compare the performances of deterministic and randomized contracts, showing the sub-optimality of the former with respect to the latter. Then, we derive a polynomial-time algorithm for computing an approximately optimal randomized contract.

\subsection{Randomized Versus Deterministic Contracts}

As a first result, we show that deterministic contracts constitute a more restrictive contract model for the principal, resulting in a potentially infinite gap between the utility achieved by an optimal deterministic contract and the one achieved by an optimal randomized contract.

\begin{restatable}{proposition}{propgap}\label{prop:propgap}
	There exists an instance of the principal-multi-agent problem in which $\frac{\optcrc}{\optdc} = \infty$.
\end{restatable}

The result stated in \cref{prop:propgap} highlights the relevance of randomized contracts and establishes the importance of randomization in designing optimal contracts. This motivates the introduction of this more complex class of mechanisms. However, while deterministic contracts can be trivially computed in polynomial time by enumerating over all possible action profiles and finding, for each of those, the optimal payment function incentivizing that action profile, the same does not hold for randomized contracts. In particular, given the quadratic nature of \cref{prog:optrc}, it is not obvious whether it is possible to compute an optimal randomized contract efficiently. In the remaining part of this section we answer this question for the affirmative. 

\subsection{Finding Approximately Optimal Randomized Contracts}
As a first step for the derivation of an efficient algorithm for computing optimal randomized contracts, we show that, in general, no randomized contract can guarantee an expected utility for the principal equal to $\optcrc$, and thus approximate algorithms are necessary to solve the problem.

\begin{restatable}{proposition}{lemmasup}\label{le:lemmasup}
	There exists an instance of a principal-multi-agent problem in which no randomized contract achieves the supremum $\optcrc$.
\end{restatable}
Notice that this fundamentally differentiates multi-agent problems from single-agent ones, where an optimal randomized contract always exists in non-Bayesian instances. 
%\mat{pensiamoci insieme ma bisogna dire perchè è diverso dalla non esistenza dell'ottimo nel single agent. Ad esempio qui è vero senza tipi. Nel frattempo diamo un intuizione dell'istanza}\fe{Ho aggiunto la frase sull'ottimo single-agent. Rimango un po' dubbioso sull'intuizione dell'istanza che non mi sembra troppo semplice da spiegare.}

As a consequence of \cref{le:lemmasup}, we focus on designing an algorithm that finds randomized contracts guaranteeing an expected utility additively close to \optcrc.
Note that \Cref{prog:optrc} is a quadratic program and it is not clear a priori how to solve it efficiently. Here, we present a two-step process to solve it.

First, we introduce a linear program related to \Cref{prog:optrc}, in which, for all $a\in\A, i\in N, \omega\in\Omega$, we linearize the quadratic variables $\mu(a)\pi_a^i(\omega)$ with a single variable $x_a^i(\omega)$. 
Moreover, we add an extra linear constraint $x_a^i(\omega)\le M\mu(a)$ (here $M$ has to be thought of as ``large''). As we will see in the following, this constraints ensures that the second step can always carried out.
As a second step, we use the solution of this linear program to ``recover'' a approximately optimal solution to the \Cref{prog:optrc}.
%The new linear program is designed so that the recovering step is always possible and can be carried out effectively.
%
%\ma{spiegare: $M$ e linearizzazione}.
%\ma{Formally, $\lprc(M,\A')$ is designed from \Cref{prog:optrc} by substituting the porduct variables $\mu(a)\pi_a^i(\omega)$ $x_a^i(\omega)$ with the single variable $x_a^i(\omega)$, thus linearizing \Cref{prog:optrc}.}

For every $\A'\subseteq \A$ and $M>0$, let us formally define the linear program $\lprc(M, \A^\prime)$ related to  \Cref{prog:optrc} described above.  $\lprc(M, \A^\prime)$  is defined as:
\begin{program}\label{prog:lprc}
	\begin{align}
		&\max_{\mu, x} \sum_{a\in\A}\sum_{\omega\in\Omega}F_a(\omega) \left(\mu(a) r_\omega - \sum_{i\in N}x^i_a(\omega)\right) \hspace{1cm}\textnormal{s.t.}\nonumber \\
		&\sum_{a_{-i}\in\A_{-i}}\left[\sum_{\omega\in\Omega}\left(x^i_{(a_i, a_{-i})}(\omega)F_{(a_i, a_{-i})}(\omega)\right) - \mu(a_i, a_{-i})c_i(a_i)\right] \nonumber \\
		& \hspace{1.3cm}\geq \hspace{-0.1cm}\sum_{a_{-i}\in\A_{-i}}\left[\sum_{\omega\in\Omega}\left(x^i_{(a_i, a_{-i})}(\omega)F_{(a_i^\prime, a_{-i})}(\omega)\right) - \mu(a_i, a_{-i})c_i(a_i^\prime)\right] \hspace{0.5cm} \forall i\in N,\forall a_i,a_i^\prime\in A_i \label{eq:lprc_ic}\\
		& 0\le x_a^i(\omega) \leq M\mu(a) \hspace{6.65cm}\forall i\in N,\forall a\in\A,\forall\omega\in\Omega \label{eq:lprc_bound}\\
		& \mu(a) = 0  \hspace{9.85cm}\forall a\in \A\setminus\A^\prime\label{eq:lprc_restr1} \\
		& x_a^i(\omega) = 0  \hspace{7.33cm} \forall i\in N,\forall a\in\A\setminus\A^\prime,\forall\omega\in\Omega\label{eq:lprc_restr2}\\
		& \mu \in \Delta(\A)\label{eq:lprc_defmu}.
	\end{align}
\end{program}
Moreover we define $\opt_{\lprc(M, \A^\prime)}$ as the optimal value of $\lprc(M, \A^\prime)$.

\begin{remark}
We defined \Cref{prog:lprc} as $\lprc(M, \A^\prime)$ making explicit its dependency from the parameter $M$, which appears in Constraints~\eqref{eq:lprc_bound}, and from the set of action profiles $\A^\prime\subseteq\A$, which appear in Constraints~\eqref{eq:lprc_restr1}~and~\eqref{eq:lprc_restr2}.
This notational gimmick helps us compactly define  multiple linear programs. Specifically, we can choose $\A'=\AindR$ if we want to refer to the program  defined only over inducible action profiles $\AindR$ or , $\A'= \A$ if we want consider all action profiles. Similarly, we can choose an $M<+\infty$ to bound the maximum payment or
set $M=+\infty$, removing the constraints $x_a^i(\omega)\le M\mu(a)$. 
\end{remark}

Our algorithm to find approximately optimal randomized contracts simply solves $\lprc(M, \A)$ for a suitably defined $M$, and then converts it to a feasible solution of \Cref{prog:optrc}. While the algorithm is quite simple, its analysis is more involved and uses the linear program $\lprc(M, \AindR)$.
We emphasize that the solution to such a linear program is utilized solely in the analysis and not within the algorithm itself. Indeed, it remains unclear how, or whether, it is possible to efficiently compute the set $\AindR$.

In the following, with a slight abuse of notation we will write $\lprc(\A^\prime)$ instead of $\lprc(+\infty, \A^\prime)$. 
Then, at a high level, the proof of the approximate optimality of the solution to $\lprc(M, \A)$ boils down to demonstrating the following inequalities:
\begin{align}
\opt_{\lprc(M, \A)}&\ge \opt_{\lprc(M, \AindR)}\tag{$\AindR\subseteq \A$}\\ 
&\gtrapprox \opt_{\lprc(\AindR)}\tag{\Cref{le:lemmatrans}}\\
&\ge\optcrc,\tag{\Cref{le:lemmarelax}}
\end{align}
where the ``$\gtrapprox$'' symbol hides a small additive approximation that decreases with $M$.

We start showing the last inequality, \ie~that $\lprc(\AindR)$ is indeed a relaxation of \Cref{prog:optrc}. Formally we need the following result:
\begin{restatable}{lemma}{lemmarelaxation}\label{le:lemmarelax}
	It holds that $\opt_{\lprc(\AindR)} \geq \optcrc$.
\end{restatable}
 Intuitively, this easily follows observing that $\lprc(\AindR)$ is a relaxation of \Cref{prog:optrc}, and that the support of $\mu$ in an optimal solution to \Cref{prog:optrc} is a subset of $\AindR$ by the definition of $\AindR$.

Next, we show that, starting from any feasible solution to $\lprc(\AindR)$, it is always possible to construct a feasible solution for $\lprc(M, \AindR)$ for a suitably defined $M$, with a small loss in the value of the objective function, essentially proving that $\opt_{\lprc(M, \AindR)}\gtrapprox \opt_{\lprc(\AindR)}$.

\begin{restatable}{lemma}{lemmatransformation}\label{le:lemmatrans}
%	For any instance $I$ of the principal-multi-agent problem instance exists a function $\tau(I)=\mc O(2^{\poly(|I|)})$. 
	%
	For each $\varepsilon > 0$ and for each $(\mu, x)$ feasible for $\lprc(\AindR)$, there exists $(\bar\mu, \bar x)$ feasible for $\lprc(M, \AindR)$ such that
	%, where $M=\frac{2 |\AindR|^2 N \tau^3}{\eps}$ such that:
	\[
		\sum_{a\in\A}\sum_{\omega\in\Omega}F_a(\omega)\left(\bar\mu(a)r_{\omega} - \sum_{i\in N} \bar x_a^i(\omega)\right) \geq \sum_{a\in\A}\sum_{\omega\in\Omega}F_a(\omega)\left(\mu(a)r_{\omega} - \sum_{i\in N} x_a^i(\omega)\right) - \eps,
	\]
	where $M\coloneqq M(|I|,\varepsilon,K)\in\poly(2^{|I|}, \nicefrac1\varepsilon,K)$ and $K=\max_{i\in N, a \in \A, \omega \in \Omega} x^i_a(\omega)$.
	%Moreover, the function $M:\mathbb{N}\times\mathbb{R}_{\ge0}\to\mathbb{R}_{\ge 0}$ is computable. \mat{io non mi lancerei in queste discussioni. A questo punto dovresti chiederti se è computatibile in tempo polinomiale. Consiglio fortemente da rimuovere da Morover in poi. Max mettiamo qualcosa subito dopo lo statement}
\end{restatable}
%
%Our algorithm for computing an approximately optimal randomized contract leverages the aforementioned results. 
Then we need to show that, thanks to the clever design of $\lprc(M, \A)$, a solution $(\mu, x)$ feasible for $\lprc(M, \A)$ can always be converted into a randomized construct $(\mu, \pi)$ feasible for \Cref{prog:optrc}.
In particular, we can define $(\mu, \pi)$ as:
\begin{equation}\label{eq:backtracking}
	\pi^i_a(\omega) \coloneq \begin{cases}
		\frac{x^i_a(\omega)}{\mu(a)} & \textnormal{if } \mu(a) > 0 \\
		0 & \textnormal{otherwise}
	\end{cases}\quad\forall i\in N,\,\forall a\in\A,\,\forall\omega\in\Omega.
\end{equation}
Notice that is is never the case that $x^i_a(\omega)>0$ and $\mu(a)=0$  thanks to constraint \eqref{eq:lprc_bound}. Hence, we never set $\pi^i_a(\omega)=0$ when $x^i_a(\omega)>0$. Formally, we can guarantee that $x^i_a(\omega)=\mu(a)\pi^i_a(\omega) $, and hence $(\mu, \pi)$ is essentially equivalent to $(\mu, x)$. 
Putting all together we can prove the prove the guarantees our algorithm.
%Then, we can leverage Lemmas \ref{le:lemmarelax} and \ref{le:lemmatrans} to show that, by suitably setting the value of $M$, the solution found in this way provides indeed a good approximation of $\optcrc$. 
%
%\ma{DROP:
%\begin{theorem}
%	For any $\eps > 0$, let $(\mu, x)$ be an optimal solution to $\lprc(\A, M)$, where $M$ is defined as per \cref{le:lemmatrans}, and let $\pi$ be defined as per \cref{eq:backtracking}. Then, the following holds: 
%	\[
%	\sum_{a\in\A}\sum_{\omega\in\Omega} \mu(a) F_a(\omega) \left[ r_\omega - \sum_{i\in N} \pi^i_a(\omega)\right] \geq \optcrc - \eps.
%	\]
%\end{theorem}
%}

\begin{restatable}{theorem}{theoremRandom}
	For any $\varepsilon>0$, there exists an algorithm which finds a $\varepsilon$-optimal randomized contract in time polynomial in the instance size and $\log(1/\varepsilon)$.
\end{restatable}

%% file: src/combinatorial.tex
\section{Multi- to Single-Agent Contract Design through Virtual Costs }\label{sec:combinatorial}

In this section, we investigate the relationship between multi-agent and single-agent problems. 
Our main result is to show that any instance of the principal-multi-agent problem is related to a single-agent instance with combinatorial actions in which costs are increased through a virtual cost function. 
%
%Intuitively, this formalizes the fact that incentivizing multiple agent is harder than incentivizing a single one. 
Intuitively, this quantifies the increased difficulty in incentivizing multiple agents compared to incentivizing a single one, which performs the joint action profile.
In particular, we show that moving from the single-agent problem to a multi-agent one is ‘‘no worse'' than increasing the costs by a multiplicative factor of $n$. 
Formally, we show that in any instance it holds that $\optdc\ge \optsing^n$, and that this result is tight. We show that this result also implies that we can extract a constant fraction $\delta/(1+\delta)$ of the $n(1+\delta)$-virtual social welfare $\vsw^{n(1+\delta)}$.
We emphasize that this result is constructive, as it allows for the efficient construction of a symmetric linear contract that guarantees to the principal utility at least $\frac{\delta}{1+\delta}\vsw^{n(1+\delta)}$.
Moreover, at the end of the section, we provide some corollaries to this result, that concern the approximation of the (non-virtual) social welfare.

We start relating the set of best responses induced by a contract in the single-agent problem with virtual costs to the Nash equilibria induced by a deterministic contract in the multi-agent one.
In particular, we show that for any payment scheme $p$ and any $a \in  \B^n(p)$ for the single-agent problem, the contract that offers $p_i=p/n$ to every agent $i\in N$ induces the joint action $a$ as an equilibrium, \emph{i.e.}, $a \in \mathsf{E}(\bar p)$. Formally:
%
%the symmetric payment scheme $\bar p=(\bar p_i)_{i \in N}$ such that $\bar p_i = p/n$ induces $a$ as the equilibrium in the multi-agent problem, \emph{i.e.}, $a \in \mathsf{E}(\bar p)$.
%
\begin{restatable}{lemma}{lemmavirtual}\label{le:lemmavirtual}
		For any instance, and any $p: \Omega\to\Reals^+$, it holds that $\B^n(p)\subseteq E(\bar p)$, where $\bar p = (\bar p^i)_{i\in N}$ is such that $\bar p^i(\omega) = p^i(\omega)/n$ for all $i\in N$ and for all $\omega\in\Omega$.
\end{restatable}
The result of \cref{le:lemmavirtual} can be used to relate the utility of an optimal deterministic contract in a multi-agent problem to the one of an optimal single-agent contract. Intuitively, the result follows observing that thanks to \Cref{le:lemmavirtual} the cumulative payment doesn't change, while the virtual cost increases by an $n$ factor.

\begin{restatable}{theorem}{thvirtual}\label{th:thvirtual}
	For any instance, it holds that $\optdc\geq\optsing^n$.
\end{restatable}

This theorem shows that, given any principal-multi-agent problem, we can consider the $n$ agents as a single agent with combinatorial actions. This has the negative effect of increasing the virtual costs by a factor $n$. As hinted in \Cref{le:lemmavirtual}, the proof of \cref{th:thvirtual} is constructive, meaning that any contract for the single-agent problem with virtual costs can be modified into a contract for the multi-agent problem dividing equally the payment among the agents.

We proceed to show that the increase in the virtual costs stated by \cref{th:thvirtual} is tight. In particular, we show that for any sub-linear (in the number of agents $n$) virtual costs, the utility in the multi-agent instance is arbitrary smaller than the one in the virtual single-agent one.

\begin{restatable}{proposition}{propvirtual}\label{prop:propvirtual}
	For any $\alpha > 0$ and $\eps > 0$, there exists an instance such that 
	\[
	\frac{\optdc}{\optsing^{n^{1-\alpha}}} \leq \eps.
	\]
\end{restatable}
In the first part of the section, we related the optimal principal's utility in the multi- and single-agent problem.
In the remaining part of this section, we relate the principal's utility with the the stronger benchmark of the ‘‘first-best'' contract, namely the one that maximizes the social welfare.
In particular, we show how to extend \cref{th:thvirtual} to provide some approximation results for the $n$ virtual social welfare $\vsw^n$. Additionally, we recover supplementary results on the approximation of the (non-virtual) social welfare $\sw$.

 Our first result follows directly from \cref{th:thvirtual}, and shows that we can convert any (instance-dependent) approximation result between $\optsing^n$ and  $\vsw^n$ into an (instance-dependent)
 approximation between $\optdc$ and  $\vsw^n$.
 %
%  for any instance-dependent approximation factor that $\optsing^n$ provides for $\vsw^n$, then an optimal deterministic contract guarantees to the principal a utility that approximates the $n$ virtual social welfare.
 %
 In particular, we assume the existence of a function $x: I\rightarrow [0,1]$ that lowerbounds the ratio between the optimal contract and the social welfare of any instance $I$. 
 %We will see some examples of such functions in the following. 
%
\begin{corollary}\label{cor:insdepvsw}
	Assume there exists a function $x$ that for all instances $I = (N, \Omega, (A_i)_{i\in N}, (c_i)_{i\in N}, F, r)$, guarantees that $\optsing \geq x(I) \vsw$. Then, it holds that
	\[
		\optdc \geq x(I) \vsw^n.
	\]
\end{corollary}
\begin{proof}
	The result follows trivially from \cref{th:thvirtual}, since $\optdc\geq\optsing^n\ge  x(I) \vsw^n $.
\end{proof}

We stated the above general result in order to accommodate different instantiations from the literature on single agent contracts.
For instance, a classical example of such an approximation function $x(\cdot)$ is $x(I) =\nicefrac 1{\Pi_{i\in N}|A_i|}$, but other approximation factors can be obtained, for instance, depending on the range of rewards or costs (see \citep{dutting2019simple} for further details). 
We remark that most of these approximations were proven via simple linear contracts.

Additionally, we complement the result of \cref{cor:insdepvsw} with an instance-independent lower bound.
In particular, we show that there always exists a linear contract that provides a $\delta/(1+\delta)$ approximation of the $n(1-\delta)$-virtual social welfare $\vsw^{n(1+\delta)}$.

\begin{restatable}{theorem}{corollaryApprinsind}\label{cor:apprinsind}
	For any instance, and any $\delta > 0$, it holds that
	\[
		\optdc \geq \frac{\delta}{1 + \delta} \vsw^{n(1+\delta)}.
	\]
	Moreover, there exists a linear contract that achieves such an approximation. 
%	provides utility at least $\frac{\delta}{1 + \delta} \vsw^{n(1+\delta)}$.
\end{restatable}
The previous results establish a connection between the principal's utility and the \emph{virtual} social welfare. 
In the following, we shift our focus to comparing the principal's utility with the (non-virtual) social welfare. 
As indirectly proved in \cref{prop:propvirtual}, it is generally impossible to recover any constant approximation of the social welfare. 
Nevertheless, we show that, under reasonable conditions, there exists a linear contract that guarantees a good approximation of the social welfare. 

First, we show that, when the all joint action have a small cost with respect to their welfare, then the virtual social welfare is close to the non-virtual one.

\begin{restatable}{lemma}{leapprsw}\label{le:leapprsw}
	For any instance $(N, \Omega, (A_i)_{i\in N}, (c_i)_{i\in N}, F, r)$ and any $\alpha \ge 1$, it holds that
	\[
		\vsw^\alpha \geq (1 - \beta(\alpha - 1))\sw,\quad\quad\textnormal{where}\quad \quad\beta\coloneqq \max_{a\in\A}\left\{\frac{c(a)}{\sum_{\omega\in\Omega}F_a(\omega)r_\omega - c(a)}\right\}.
	\]
\end{restatable}

This result is intuitive: when each action gives an expected welfare far above its costs ($\beta\approx 0$), then we can increase the costs in the virtual instance, and still have good guarantees on the virtual social welfare.
We can now combine the results of \cref{cor:apprinsind} and \cref{le:leapprsw} to show that when each action has small cost with respect to the social welfare that it yields, there exists a linear contract providing a good approximation of $\sw$.

\begin{corollary}
	For any instance $(N, \Omega, (A_i)_{i\in N}, (c_i)_{i\in N}, F, r)$, it holds that 
	\[
		\optdc \geq \frac{\delta}{1+\delta}\left(1 - \beta n(1+\delta)\right)\sw,
	\]
	where $\beta$ is defined as per \cref{le:leapprsw}. Moreover, there exists a linear contract that guarantees such an approximation.
\end{corollary}

\begin{proof}
		The result follows trivially from \cref{cor:apprinsind} and~\cref{le:leapprsw} since:
		\[\optdc\geq\frac{\delta}{1 + \delta} \vsw^{n(1+\delta)}\ge  \frac{\delta}{1+\delta}\left(1 - \beta\left(n(1+\delta)- 1\right)\right)\sw\ge \frac{\delta}{1+\delta}\left(1 - \beta n(1+\delta)\right)\sw.\]
\end{proof}

An extreme case arises when the costs of actions are close to $0$. In this scenario, there are de facto no externalizes  (\ie,~$\beta\approx0$) and we can extract the entire social welfare by setting $\delta\rightarrow \infty$. On the other hand, if actions are very costly (\ie,~$\beta\gg 1$) then we can extract only a zero fraction of $\sw$. The results above interpolates between such extremes.
%
%While these results give significant only in specific instances in the general setting 
%In the following, we will analyze Bayesian settings with types and we will provide more general conditions under which a constant fraction of the social welfare can be extracted.

%% file: src/bayesian.tex
\section{The Bayesian Principal-Multi-Agent Problem}\label{sec:bayesian}
A \emph{Bayesian principal-multi-agent} problem is defined by a tuple $(N, \Omega, (\Lambda_i)_{i\in N}, G, (A_i)_{i\in N}, c, F, r)$, where  $N$, $\Omega$, $(A_i)_{i\in N}$, and $r$ are defined in the non-Bayesian setting. 
 %$N=\range{n}$ is a set of $n$ agents, $\Omega$ is a set of $m$ outcomes, $A_i$ is a discrete set of actions available to each agent $i$, and $r_\omega\in [0,1]$ is the reward for the principal associated to outcome $\omega\in\Omega$. 
 Additionally, the problem instance is enhanced with a set $\Lambda_i \subset \Reals$ of types for each agent $i$, which determine both the cost functions of the agents and the outcome probabilities. In particular, when the agents have type profile $\lambda\in\Lambda\coloneqq \times_{i\in N}\Lambda_i$, the outcome distribution induced by each action profile $a\in\A$ is denoted as $F_a^\lambda\in\Delta(\Omega)$, while the cost function of each agent is denoted as $c_i^{\lambda_i}: A_i\to [0,1]$ (notice that the cost of each agent depends exclusively on their type). The agents' types are drawn from some publicly-known distribution $G$.

\subsection{Contract Design}
In the Bayesian setting, the interaction protocol between the principal and the agents is augmented with a preliminary \emph{type-reporting} stage, in which each agent privately observes their type and report it (possibly untruthfully) to the principal. More in detail, the interaction goes as follows: 
\begin{enumerate}[label=(\roman*)]
	\item The principal commits to a contract, which determines action recommendations and payments for each possible type profile communicated by the agents;
	\item A type profile $\lambda\in\Lambda$ is drawn from $G$, each agent privately observes their type $\lambda_i$ and communicates a type $\lambda_i^\prime\in\Lambda_i$ to the principal (possibly lying, \ie,  $\lambda'_i\neq\lambda_i$);
	\item The principal observes the tuple of types $\lambda^\prime$ reported by the agents and privately recommends to each agent $i\in N$ an action $a_i\in A_i$; 
	\item Each agent observes the action recommendation and plays $a_i^\prime\in A_i$, incurring cost $c_i^{\lambda_i}(a_i^\prime)$. Also in this case, the agent can deviate from the action recommendation;
	\item An outcome $\omega\in\Omega$ is sampled according to $F_{a^\prime}^{\lambda}$;
	\item The principal receives reward $r_\omega$ and pays the agent according to the payment scheme for the reported type profile $\lambda^\prime$. 
\end{enumerate}

%\fe{Non so se menzionare ancora le cose di LL e di IR come sopra, ma forse eviterei.}
%
Similarly to the previous setting, we will focus on two main classes of contracts, \emph{i.e.,} \emph{randomized contracts} and \emph{deterministic contracts}. 
\paragraph{Randomized contract.} A randomized contract is a tuple $(\mu, \pi)$, where $\mu=\{\mu^{\lambda}\in\Delta(\A):\lambda\in\Lambda\}$ is a set of of probability distributions over action recommendations for each tuple of types $\lambda$, and $\pi = \left\{\pi^{\lambda} :\lambda\in\Lambda\right\}$ is a set of tuples $\pi^\lambda = (\pi^{\lambda, i}_a)_{(i,a)\in N\times\A}$, with $\pi^{\lambda, i}_a: \Omega\to\Reals^+$. In particular, the value $\pi^{\lambda, i}_a(\omega)$, specifies the payment for agent $i$ in outcome $\omega\in\Omega$, when the types reported are $\lambda$ and the action profile recommended is $a\in\A$. Moreover, given any randomized contract $(\mu, \pi)$, we define the quantity $U_i^{\mu, \pi}(\lambda_i\to\lambda_i^\prime, a_i\to a_i^\prime\vert \lambda_{-i})$, which models the (unnormalized) expected utility for agent $i$ when the true type profile is $(\lambda_i, \lambda_{-i})\in\Lambda$, they report type $\lambda_i^\prime\in\Lambda_i^\prime$ to the principal, receive recommendation $a_i\in A_i$ and play action $a_i^\prime\in A_i$. Formally, 
\[
	U_i^{\mu, \pi}(\lambda_i\to\lambda_i^\prime, a_i\to a_i^\prime\vert \lambda_{-i})\coloneqq \sum_{a_{-i}\in\A_{-i}} \mu^{(\lambda_i^\prime, \lambda_{-i})}(a_i, a_{-i}) \left(\sum_{\omega\in\Omega}\pi^{(\lambda_i^\prime, \lambda_{-i}), i}_{(a_i, a_{-i})}(\omega)F_{(a_i^\prime, a_{-i})}^{(\lambda_i, \lambda_{-i})}(\omega) - c_i^{\lambda_i}(a_i^\prime)\right).
\]
For notational convenience, we will write $U_i^{\mu, \pi}(\lambda_i, a_i \vert \lambda_{-i})$ instead of $U_i^{\mu, \pi}(\lambda_i\to\lambda_i, a_i\to a_i \vert \lambda_{-i})$. 

An optimal randomized contract for the Bayesian setting can be found as a solution to the following optimization problem: 
\begin{subequations}\label{progr:bay-optrc}
	\begin{align}[left =\boptcrc\coloneqq\empheqlbrace]
		&\sup_{\mu, \pi}\, \mathbb{E}_{\lambda\sim G} \left[\sum_{a\in\A}\sum_{\omega\in\Omega} \mu^\lambda(a)F_a^\lambda(\omega)\left(r_\omega - \sum_{i\in N} \pi^{\lambda, i}_a(\omega)\right)\right]\hspace{1cm} \textnormal{s.t.}\hspace{2.25cm}\, \nonumber\\
		&\hspace{.2cm}\sum_{a_i\in A_i}U_i^{\mu, \pi}\left(\lambda_i, a_i\vert\lambda_{-i}\right)	\nonumber \\ 
		&\specialcell{\hspace{.6cm}\geq\hspace{-2mm} \sum_{a_i\in A_i} \max_{a_i^\prime\in A_i} U_i^{\mu, \pi}\hspace{-1mm}\left(\lambda_i\to\lambda_i^\prime, a_i\to a_i^\prime\vert\lambda_{-i}\right) \hfill\forall i\in N,\forall a_i\in A_i,\forall \lambda \in\Lambda,\forall\lambda_i^\prime\in\Lambda_i}\label{eq:bay-optrc_ic}\\
		&\specialcell{\hspace{2mm}\pi^{\lambda, i}_a(\omega)\geq 0\hfill\forall i\in N,\forall a\in\A,\forall\lambda\in\Lambda,\forall\omega\in\Omega} \\
		&\specialcell{\hspace{2mm}\mu^{\lambda}\in\Delta(\A)\hfill\forall\lambda\in\Lambda},
	\end{align}
\end{subequations}
where the Constraint \eqref{eq:bay-optrc_ic} restricts the feasibility set to randomized contracts that are \emph{dominant-strategy incentive compatible} (DSIC).

\paragraph{Deterministic contract.} A deterministic contract is a tuple $(\mathfrak{a}, p)$, where $ \mathfrak{a}=\{a^\lambda\in\A:\lambda\in\Lambda\}$ is a set of one action recommendation for each type profile, and $p=\{p^{\lambda}: \lambda\in\Lambda\}$ is a set of tuples $p^\lambda =(p^{\lambda, i})_{i\in N}$, with $p^{\lambda, i}:\Omega\to\Reals^+$. Here, the value $p^{\lambda, i}(\omega)$ denotes the payment received by agent $i$ when the outcome is $\omega\in\Omega$ and the type profile reported by the agents is $\lambda\in\Lambda$. An optimal deterministic contract can be found as a solution to the following optimization problem:
\begin{subequations}\label{progr:bay-optdc}
	\begin{align}[left =\boptdc\coloneqq\empheqlbrace]
		&\sup_{\mathfrak{a}, p}\,\mathbb{E}_{\lambda\sim G} \left[F_{a^\lambda}^\lambda(\omega)\left(r_\omega - \sum_{i\in N} p^{i, \lambda}(\omega)\right)\right]\hspace{1cm}\textnormal{s.t.} \hspace{4cm}\,\nonumber\\
		&\hspace{.2cm} \sum_{\omega\in\Omega} F^{\lambda}_{a^\lambda}(\omega) p^{\lambda, i}(\omega) - c_i^{\lambda_i}(a_i^{\lambda}) \nonumber \\
		&\specialcell{\hspace{.4cm}\geq \sum_{\omega\in\Omega}F^{\lambda}_{(a_i^\prime, a^{(\lambda_i^\prime,\lambda_{-i})}_{-i})}(\omega) p^{(\lambda_i^\prime, \lambda_{-i}), i}(\omega) - c_i^{\lambda_i}(a_i^\prime) \hfill\forall i\in N,\lambda\in\Lambda,\lambda_i^\prime\in\Lambda_i, a_i^\prime\in A_i} \label{eq:bayoptdc_ic}\\
		&\specialcell{\hspace{.2cm}p^{\lambda, i}(\omega)\geq 0\hfill\forall i\in N, \lambda\in\Lambda,\omega\in\Omega}\label{eq:bayoptdc_ll}\\
		&\specialcell{\hspace{.2cm}a^\lambda \in \A\hfill\forall\lambda\in\Lambda}.
	\end{align}
\end{subequations}

\subsection{Single-dimensional Types}
As a particular problem instance, we will consider the case in which the agents' types are \emph{single-dimensional}, \emph{i.e.,} they only affect the costs. More in detail, with single-dimensional types, the cost functions are defined such that $c_i^{\lambda_i}(a_i) = \lambda_i c_i(a_i)$, while the outcome distribution is independent from the types of the agents, \emph{i.e.} it is such that $F_a^\lambda =F_a$ for all $a\in\A$ and for all $\lambda\in\Lambda$.

\subsection{(Virtual) Social Welfare}
Throughout our analysis of the Bayesian setting, we will be interested in providing approximation results for the social welfare and virtual social welfare. For any type profile $\lambda\in\Lambda$ and virtual cost function $\phi:\Reals^+\to\Reals^+$, we let $\vsw^\phi(\lambda)$ be the virtual social welfare of the instance $(N, \Omega, (A_i)_{i\in N}, (c_i^{\lambda_i})_{i\in N}, F^{\lambda}, r)$ defined by the type profile $\lambda$. The virtual social welfare of the Bayesian instance is defined as $\vsw^{\phi} \coloneqq \mathbb{E}_{\lambda\sim G}\vsw^{\phi}(\lambda)$. Similarly, we denote as $\sw(\lambda) = \vsw^1(\lambda)$ the social welfare of the instance defined by the types $\lambda$, and as $\sw \coloneqq \mathbb{E}_{\lambda\sim G}\sw(\lambda)$ the social welfare of the Bayesian instance.    

%% file: src/optbcrc.tex
\section{Finding Approximately Optimal Randomized Contracts in the Bayesian Setting}\label{sec:boptcrc}

In this section we investigate the problem of finding an approximately optimal randomized contract for the Bayesian principal-multi-agent problem. For the sake of this objective, we focus our attention to cases in which the sets of types $\Lambda_i$ are discrete. Moreover, we will make no further assumption on the problem instances and consider the general case in which the types are \emph{multi-dimensional}.

First, notice that the result of \cref{le:lemmasup} for the non-Bayesian case trivially carries over to the Bayesian setting. Thus, also in this case, our objective is to find a randomized contract that guarantees an expected utility close to $\boptcrc$ up to a small additive approximation error.

Hence, similarly to what we did in \cref{sec:optcrc}, we introduce a linear program which relates to Program \ref{progr:bay-optrc}. In this case, together with the linearization of the quadratic variables $\mu^\lambda(a)\pi_a^{\lambda, i}(\omega)$ by means of $x_a^{\lambda, i}(\omega)$, and the introduction of an additional constraint $x_a^{\lambda, i}(\omega)\leq M\mu^\lambda(a)$ for each $i\in N$, $a\in \A$, $\omega\in\Omega$ and $\lambda\in\Lambda$, we introduce an additional auxiliary variable $z$ to replace the $\max$ in constraint \eqref{eq:bay-optrc_ic}. All these modifications will prove crucial in order to recover an approximately optimal randomized contract from a suitably parameterized version of such a linear program. 
Formally, for any $M>0$, and any tuple $\A^\prime = (\A^\prime(\lambda))_{\lambda\in\Lambda}$, where $\A^\prime(\lambda)\subseteq\A$ for each $\lambda\in\Lambda$, we denote as $\baylp(M, \A^\prime)$, the following linear program: 
\begin{subequations}\label{prog:baylprc}
	\begin{align}
		\max_{\mu, x, z}& \hspace{.2cm}\sum_{\lambda\in\Lambda} G(\lambda)\sum_{a\in\A}\sum_{\omega\in\Omega}F_a(\omega) \left(\mu^\lambda(a) r_\omega - \sum_{i\in N}x^{\lambda,i}_a(\omega)\right) \hspace{1cm}\textnormal{s.t.}\hspace{3.7cm}\,\nonumber \\
		& \specialcell{\sum_{a\in\A} x_a^{\lambda, i}(\omega)F_a^\lambda(\omega) - c_i^{\lambda_i}(a_i) \geq \sum_{a_i\in A_i} z^i(\lambda, \lambda_i^\prime, a_i)\hfill\forall i\in N,\,\forall \lambda\in\Lambda,\,\forall\lambda_i^\prime\in\Lambda_i}\label{eq:baylprc_ic}\\
		& z^i(\lambda, \lambda_i^\prime, a_i) \nonumber\\
		&\specialcell{\hspace{.4cm} \geq\hspace{-.3cm}\sum_{a_{-i}\in\A_{-i}} \hspace{-.3mm}\sum_{\omega\in\Omega} \hspace{-.4mm} x^{(\lambda_i^\prime, \lambda_{-i}), i}_{(a_i, a_{-i})}\hspace{-.5mm}(\omega) F^{\lambda}_{(a_i^\prime, a_i)}\hspace{-.5mm}(\omega) - c_i^{\lambda_i}\hspace{-.5mm}(a_i^\prime)\hfill\forall i\in N,\forall\lambda\in\Lambda,\forall\lambda_i^\prime\in\Lambda_i,\forall a_i, a_i^\prime\in A_i} \label{eq:baylprc_max}\\
		%&\specialcell{z^i\hspace{-.5mm}(\lambda, \lambda_i^\prime, a_i)\hspace{-1mm} \geq\hspace{-4mm}\sum_{a_{-i}\in\A_{-i}} \hspace{-.6mm}\sum_{\omega\in\Omega} \hspace{-.4mm} x^{(\lambda_i^\prime, \lambda_{-i}), i}_{(a_i, a_{-i})}\hspace{-.5mm}(\omega) F^{\lambda}_{(a_i^\prime, a_i)}\hspace{-.5mm}(\omega) - c_i^{\lambda_i}\hspace{-.5mm}(a_i^\prime)\hfill\forall i\in N,\forall\lambda\in\Lambda,\forall\lambda_i^\prime\in\Lambda_i,\forall a_i, a_i^\prime\in A_i}\label{eq:baylprc_max}\\
		& \specialcell{0\le x_a^{\lambda, i}(\omega) \leq M\mu^\lambda(a) \hfill\forall i\in N,\forall\lambda\in\Lambda,\forall a\in\A,\forall\omega\in\Omega} \label{eq:baylprc_bound}\\
		& \specialcell{\mu^\lambda(a) = 0  \hfill\forall\lambda\in\Lambda,\forall a\in \A\setminus\A^\prime(\lambda)}\label{eq:baylprc_restr1} \\
		& \specialcell{x_a^{i,\lambda}(\omega) = 0  \hfill \forall i\in N,\forall\lambda\in\Lambda,\forall a\in\A\setminus\A^\prime(\lambda),\forall\omega\in\Omega}\label{eq:baylprc_restr2}\\
		& \specialcell{\mu^{\lambda} \in \Delta(\A)\hfill\forall\lambda\in\Lambda}\label{eq:baylprc_defmu}.
	\end{align}
\end{subequations}
We denote the optimal value of the above linear program as $\opt_{\baylp(M, \A^\prime)}$. Additionally, we will use $\baylp(\A^\prime)$ as a shorthand for $\baylp(+\infty, \A^\prime)$, and we will write $\baylp(M, \A) = \baylp(M, \A^\prime)$, when $\A^\prime(\lambda) = \A$ for all $\lambda\in\Lambda$. 

Our algorithm to find approximately optimal randomized contracts solves $\lprc(M, \A)$ for a suitably defined $M$, and then converts it to a feasible solution of Program~\ref{progr:bay-optrc}.
To do so, starting from an optimal $(\mu, x, z)$ for $\baylp(M, \A)$, we can recover a randomized contract $(\mu, \pi)$ such that 
\begin{equation}\label{eq:bayconversion}
	\pi^{\lambda, i}_a(\omega) = \begin{cases}
		\frac{x^{\lambda, i}_a(\omega)}{\mu^\lambda(a)} &\textnormal{if } \mu^\lambda(a) > 0 \\
		0 &\textnormal{otherwise} 
	\end{cases}\quad\quad\forall i\in N,\forall \lambda\in\Lambda,\forall a\in\A,\forall\omega\in\Omega.
\end{equation}
We remark that, thanks to the definition of the linear program, and in particular thanks to constraints \eqref{eq:baylprc_ic} and \eqref{eq:baylprc_max}, such a contract is guaranteed to be DSIC. Similarly, constraint \eqref{eq:baylprc_bound}, guarantees that the conversion procedure described in \cref{eq:bayconversion} can be carried out safely.

Following similar steps to the one described in \cref{sec:optcrc}, it is possible to show that such a $(\mu, \pi)$ guarantees to the principal an expected utility approximately equal to $\boptcrc$, thus yielding the following result. 

\begin{restatable}{theorem}{thbayoptcrc}\label{th:thbayoptcrc}
	For any $\eps > 0$, there exists an algorithm which finds an $\eps$-optimal randomized contract in any instance of the Bayesian principal-multi-agent problem in time polynomial in the instance size and $\log(1/\eps)$.
\end{restatable}

Due to space constraints, we defer to \cref{app:app_bayoptcrc} a more thorough description of the technical details behind our algorithm for finding an approximately optimal randomized contract in Bayesian principal-multi-agent problems.

%% file: src/approx.tex
\section{Approximation results for the (Virtual) Social Welfare 
	%Approximation Results for Deterministic Contracts
	}\label{sec:approx}
In this section, we investigate whether the approximation result in Section~\ref{sec:combinatorial} extends to Bayesian settings. 
Our first result is negative, as we show that no deterministic contract can obtain any approximation of the virtual social welfare for any choice of the virtual cost function.

%the power of deterministic contracts for Bayesian principal-multi-agent problems. We will focus on instances with single-dimensional types, without making any additional assumption on the set of types (\emph{i.e.,} we no longer assume that the sets of types are discrete). \fe{Va giustificato il fatto che ci interessiamo al vsw (citare previous works?)} 

\begin{restatable}{proposition}{propbaysw}\label{prop:propbaysw}
	For any $\alpha \in (0,1)$, there exists an instance of the Bayesian principal-multi-agent problem in which $\frac{\vsw^{1/\alpha}}{\boptdc}= \infty$. 
\end{restatable}

The result of \cref{prop:propbaysw} highlights the additional challenges in designing contracts within Bayesian settings. As a direct implication of this result, we obtain that, unlike in non-Bayesian settings, here it is not possible to achieve good approximations of the (virtual) social welfare by means of linear contracts.
 
In the following, we circumvent this negative result, showing that, if we drop the limited liability assumption, then it is possible to obtain deterministic contracts that achieve a good approximation factor of the $n$ virtual social welfare. To this extent, we denote as $\nollboptdc$ the optimal value of Program \ref{progr:bay-optdc}, in which we drop the LL constraint \eqref{eq:bayoptdc_ll} and add the following IR constraint (notice that, without LL, IR is no longer implied by the IC constraint \eqref{eq:bayoptdc_ic} and thus we need to add an explicit IR constraint):
\begin{equation}\label{eq:bayoptdc_ir}
	\sum_{\omega\in\Omega}F_{a^\lambda}(\omega)p^{\lambda, i}(\omega) \ge c_i^{\lambda_i}(a_i^\lambda)\quad\quad\forall \lambda\in\Lambda,\,\forall i\in N.
\end{equation}

We will show that, for any $\delta>0$, $\nollboptdc$ guarantees an instance-independent constant approximation factor of the $n(1+\delta)$-virtual social welfare $\vsw^{n(1+\delta)}$. In order to obtain such a result, for any $\delta>0$ we consider the following deterministic and affine contract $(\mathfrak{a}_\delta, p_\delta)$, which we denote as $\detcontr$:
\begin{equation}\label{eq:deterministicnoll}
	\begin{split}
		&a^\lambda_\delta \in arg\max_{a\in\A}\left\{\frac{1}{n(1+\delta)} \sum_{\omega\in\Omega} F_{a}(\omega)r_\omega - \sum_{i\in N} c_i^{\lambda_i}(a_i)\right\} \hspace{4.5cm}\forall \lambda\in\Lambda \\
		&\specialcell{p_\delta^{\lambda, i}(\omega) = \frac{1}{n(1+\delta)}r_\omega - \sum_{j\in N\setminus\{i\}} c_j^{\lambda_j}(a^\lambda_j)\hfill\forall \lambda\in\Lambda,\forall i\in N,\forall\omega\in\Omega}.
	\end{split}
\end{equation}
Intuitively, the above contract incentivizes, for each profile of types $\lambda\in\Lambda$, the action profile which maximizes the $n(1+\delta)$-virtual social welfare $\vsw^{n(1+\delta)}(\lambda)$. Moreover, we can highlight two components defining the affine payment functions: a linear contract, which gives to each agent a suitably defined fraction of the principal reward, and a second component which, for each set of types $\lambda$, penalizes each agent with an amount that is equal to the sum of the other players' costs for the incentivized actions. We remark that a payment function defined in this way may not satisfy the limited liability property. As a first step, we show how the peculiar definition of $\detcontr$ guarantees that such a deterministic contract satisfies the DSIC and IR constraints. Intuitively, this can be shown considering that the action recommendations are designed so to maximize the $n(1+\delta)$-virtual social welfare $\vsw^{n(1+\delta)}(\lambda)$. In a similar way, the payment function is designed so that each agent $i$ ‘‘pays'' the cost of the other players' actions, thus guaranteeing to $i$ an expected utility that corresponds to a fraction of the virtual social welfare. This guarantees that each agent is incentivized to act truthfully, since this would maximize the virtual social welfare (and hence also their utility). The following lemma formalizes this argument.

\begin{restatable}{lemma}{lemmadsicbay}\label{le:lemmadsicbay}
	For any $\delta>0$, $\detcontr$ is DSIC and IR, while it may not satisfy LL.
\end{restatable}

Then, it is possible to show that, for any profile of types $\lambda\in\Lambda$, the deterministic contract $\detcontr$ guarantees an $\delta/(1+\delta)$-approximation of the $n(1+\delta)$-virtual social welfare. 

\begin{restatable}{lemma}{lemmaapproxvswbay}\label{le:lemmaapproxvswbay}
	For any $\delta>0$ and $\lambda\in\Lambda$, $\detcontr$ guarantees that
	\[
		\sum_{\omega\in\Omega} F_{a_\delta^\lambda}(\omega)\left(r_\omega - \sum_{i\in N}p_\delta^{\lambda, i}(\omega)\right) \ge \frac{\delta}{1+\delta} \vsw^{n(1+\delta)}(\lambda). 
	\]
\end{restatable}
Thus, we can use \cref{le:lemmaapproxvswbay} to prove that $\detcontr$ does indeed provide a good approximation of the $n(1 + \delta)$ virtual social welfare.

\begin{theorem}\label{th:thapproxvswbay}
	For any $\delta>0$, $\detcontr$ guarantees that
	\[
	\mathbb{E}_{\lambda\sim G}\left[\sum_{\omega\in\Omega} F_{a_\delta^\lambda}(\omega)\left(r_\omega - \sum_{i\in N}p_\delta^{\lambda, i}(\omega)\right)\right] \ge \frac{\delta}{1+\delta}\vsw^{n(1+\delta)}.
	\]
\end{theorem}
\begin{proof}
	The result follows trivially from \cref{le:lemmaapproxvswbay} and the definition of virtual social welfare:
	\[
		\mathbb{E}_{\lambda\sim G}\left[\sum_{\omega\in\Omega} F_{a_\delta^\lambda}(\omega)\left(r_\omega - \sum_{i\in N}p_\delta^{\lambda, i}(\omega)\right)\right] \ge \frac{\delta}{1+\delta}\mathbb{E}_{\lambda\sim G}\left[\vsw^{n(1+\delta)}(\lambda)\right] = \frac{\delta}{1+\delta}\vsw^{n(1+\delta)}.
	\]
\end{proof}

Finally, given the result of \cref{th:thapproxvswbay} and the fact that the deterministic contract $\detcontr$ satisfies both DSIC and IR constraints (\cref{le:lemmadsicbay}), we can obtain the following corollary giving an approximation result for \nollboptdc. 

\begin{corollary}\label{cor:corapproxvswbay}
	For any $\delta>0$, it holds that
	\[
	\nollboptdc \ge \frac{\delta}{1+\delta}\vsw^{n(1+\delta)}.
	\]
\end{corollary}	

%Let us point out that, differently than the non-Bayesian case, it is non-trivial to efficiently compute an optimal deterministic mechanism. In this context, the result of \cref{th:thapproxvswbay} shows that it is enough to use a simple and intuitive deterministic contract like $\detcontr$ to retain the approximation guarantees offered by an optimal contract. \mat{optimal contract? forse stai dicendo worst-case optimal? non chiaro.} 
To strengthen this statement, we complement the result of \cref{cor:corapproxvswbay}, showing that the lower bound for $\nollboptdc$ is tight in its dependency from $n$.
	
\begin{restatable}{proposition}{proplowerboundbay}\label{prop:proplowerboundbay}
	For any $\alpha>0$ and $\eps \in (0,1)$, there exists an instance of the Bayesian principal-multi-agent problem where 
	\[
		\frac{\nollboptdc}{\vsw^{n^{1-\alpha}}} \le \eps.
	\]
\end{restatable}
In other words, \cref{prop:proplowerboundbay} shows that for any choice of sublinear (in $n$) virtual costs, the utility that can be extracted by an optimal deterministic contract can be arbitrarily smaller than the virtual social welfare, even without enforcing limited liability. Again, this result highlights the advantages of $\detcontr$, which, despite its low computational burden, offers approximation guarantees of the virtual social welfare that are worst-case optimal. 
\subsection{Approximation Guarantees for Social Welfare}
We conclude this section deriving some approximation results for $\detcontr$ in terms of the social welfare $\sw$.To do so, we will start from the result of \cref{th:thapproxvswbay} and investigate the relation between the virtual social welfare $\vsw^{n(1+\delta)}$ and the social welfare $\sw$. 

Deriving approximation results for the social welfare is arguably more challenges. Indeed, in general the gap between the virtual social welfare and the social welfare can be arbitrarily large.
Following the work of \citet{alon2023bayesian} on single-agent contract design, we will restrict our attention to instances of the Bayesian principal-multi-agent problem exhibiting some regularity. 
In particular, we focus on \emph{single-dimensional} types. Moreover, we assume the following two conditions.

%\begin{assumption}
%	The set $\Lambda$ is both an open set and a Borel set.
%\end{assumption}
%
\begin{assumption}[Small-tail assumption]\label{ass:smalltail}
	The problem instance is $(\delta, \eta)$-small-tail, \emph{i.e.,} it is such that
	\[
		\int_{n(1+\delta)\Lambda\cap\Lambda} \sw(\lambda)g(\lambda) d\lambda \ge \eta\sw,
	\]
	where $n(1+\delta)\Lambda \coloneqq \left\{(n(1+\delta)\lambda_i)_{i\in N}: \lambda\in\Lambda\right\}$.
\end{assumption}
\begin{assumption}[Non-increasing density]\label{ass:noninc}
	The probability density function $g$ associated to distribution $G$ is non-increasing, \emph{i.e.,}
	\(
		g(\lambda) \ge g(\lambda^\prime),
	\)
	for each $\lambda,\lambda^\prime\in\Lambda$ such that $\lambda_i \le \lambda_i^\prime$ for all $i\in N$.
\end{assumption}

Notice that both \cref{ass:smalltail} and \cref{ass:noninc} generalize to multi-agent settings the small-tail and non-increasing density assumptions that were already studied for single-agent problems by \citet{alon2023bayesian}. At an high level, the small-tail assumption serves the purpose of determining how much of the social welfare is distributed around large types. Intuitively, the closer $\eta$ is to $1$, the more social welfare can be extracted in correspondence of types with large values. Moreover, let us point out that if the sets of types $\Lambda_i$ take the form $[0, d]$ for some value of $d>0$, then the small-tail assumption is trivially satisfied for each value of $\eta\in [0,1]$. The non-increasing density assumption, instead, enforces that sampling lower types is more probable than sampling higher ones, thus ruling out all those cases in which the agents' costs are more likely to be high. Then, we can show that, under such assumptions on the problem instance, the deterministic contract $\detcontr$ guarantees a good approximation of the social welfare. 

\begin{restatable}{theorem}{thsocialwelfarebayesian}\label{th:thsocialwelfarebayesian}
	Assume that Assumptions \ref{ass:smalltail} and \ref{ass:noninc} hold. Then, $\detcontr$ guarantees that
	\[
			\mathbb{E}_{\lambda\sim G}\left[\sum_{\omega\in\Omega} F_{a_\delta^\lambda}(\omega)\left(r_\omega - \sum_{i\in N}p_\delta^{\lambda, i}(\omega)\right)\right] \ge \eta \frac{\delta}{n^n(1+\delta)^{n+1}}\sw.
		\]
\end{restatable}

We remark that the most interesting applications of our results are with a small set of agents. This is due both to our negative results such as \cref{prop:proplowerboundbay} and computational reasons. Indeed, the representation of the instance grows exponentially in $n$. This makes the exponential dependency in $n$ of the previous result tolerable. 

%\mat{bisogna dire perche c'è n all'esponente. altrimenti è una follia se uno lo vede senza motivazione,}
%
We can use \cref{th:thsocialwelfarebayesian} and \cref{le:lemmadsicbay} to derive the following result for $\nollboptdc$.
\begin{corollary}
	Assume that Assumptions \ref{ass:smalltail} and \ref{ass:noninc} hold. Then,
	\[
		\nollboptdc \ge \eta \frac{\delta}{n^n(1+\delta)^{n+1}} \sw.
	\]
\end{corollary}

We conclude by showing a simple example in which it is possible to leverage the approximation result of \cref{th:thsocialwelfarebayesian}. 

\begin{corollary}\label{cor:exampleapp}
	If the types have support in $[0, d]^n$ for some $d>0$ and the distribution has non-increasing density, then \detcontr guarantees that
	\[
	\mathbb{E}_{\lambda\sim G}\left[\sum_{\omega\in\Omega} F_{a_\delta^\lambda}(\omega)\left(r_\omega - \sum_{i\in N}p_\delta^{\lambda, i}(\omega)\right)\right] \ge  \frac{\delta}{n^n(1+\delta)^{n+1}} \sw.
	\]
\end{corollary}
\cref{cor:exampleapp} captures only a small family of instances in which our simple affine contract guarantees to extract a constant approximation factor of the social welfare when the number of agents is constant. A special case is the one in which the types are uniformly distributed on $[0, d]^n$ for some $d>0$. Moreover, similar approximation results can be obtained for other problem instances  by relaxing the non-increasing density assumption and enlarging the scope to \emph{slowly-increasing} density functions (we refer the interested reader to \cite{alon2023bayesian} for more details).

%% file: src/appendix.tex
\input{src/app_optcrc}
\input{src/app_combinatorial}
\input{src/app_bayoptcrc}
\input{src/app_approx}

%% file: src/app_optcrc.tex
\section{Proofs Omitted from Section \ref{sec:optcrc}}
\propgap*
\begin{proof}
	Consider an instance of the principal-multi-agent problem where $N=\{1,2\}$,  $\Omega=\{\omega_1,\omega_2\}$, and $A_1=A_2=\{\aone,\atwo\}$. The rewards are defined such that $r_{\omega_1} = 1$ and $r_{\omega_2} = 0$. The outcome probabilities are defined as in the following table: 
	\begin{table}[H]
		\centering
		\begin{tabular}{|c | c c |}
			\hline 
			$F_a(\omega)$ & $\omega_1$ & $\omega_2$ \\
			\hline
			$\aone\aone$ & $1$ & $0$ \\
			$\aone\atwo$ & $2/5$ & $3/5$ \\
			$\atwo\aone$ & $2/5$ & $3/5$ \\
			$\atwo\atwo$ & $0$ & $1$ \\
			\hline 
		\end{tabular}
	\end{table}
	Additionally, the action costs are $c_i(\aone) = 2/5$ and $c_i(\atwo) = 0$ for all $i\in N$. To prove the proposition, we first show that in the consider instance the utility of the optimal deterministic contract is $0$, and then prove that $\optcrc>0$. 
	
	\paragraph{Upper bound on $\optdc$.} 
	Let $(a, p)$ be an optimal deterministic contract, \emph{i.e.,} such that 
	\[
	\sum_{\omega\in\Omega} F_a(\omega) \left(r_\omega - \sum_{i\in N} p^i(\omega)\right) = \optdc.
	\]
	By  IC  with respect to action $\atwo$, the expected payment received by an agent to which is recommended action $\aone$ must be at least equal to the action cost (otherwise the agent would play $\atwo$). Formally, for each $i\in N$ such that $a_i=\aone$, it must hold that $\sum_{\omega\in\Omega} F_a(\omega) p^i(\omega) \geq c_i(\aone)$, which gives the following:
	\begin{equation}\label{eq:bounddc_propgap}
		\optdc = \sum_{\omega\in\Omega} F_a(\omega) \left(r_\omega - \sum_{i\in N} p^i(\omega)\right) \leq \sum_{\omega\in\Omega} F_a(\omega)r_\omega  - \sum_{i\in N} c_i(a_i).
	\end{equation}
	From \cref{eq:bounddc_propgap}, it follows that the only action profile that can achieve a strictly positive upper bound on $\optdc$ is $\aone\aone$, while all the others action profiles guarantee non-positive principal's expected utility. Hence, let us assume that the incentivized action is $a=\aone\aone$. Then, we have that the payment scheme $p$ must satisfy the following IC constraint for each $i\in N$: 
	\[
	p^i(\omega_1) - \frac{2}{5} \geq \frac{2}{5}p^i(\omega_1) + \frac{3}{5}p^i(\omega_2).
	\]
	Rearranging, we get
	\[
	\frac{3}{5}p^i(\omega_1) \ge \frac{2}{5} + \frac{3}{5}p^i(\omega_2) \ge \frac{2}{5}.
	\]
	This implies $p^i(\omega_1) \geq 2/3$ for each $i\in N$. As a consequence, we have that if a contract incentivizes action profile $\aone\aone$, then the utility for the principal must be
	\[
	\sum_{\omega\in\Omega} F_{(\aone,\aone)}(\omega) r_\omega - \sum_{i\in N} p^i(\omega) \leq 1 - \frac{4}{3} \leq 0,
	\]
	which proves that $\optdc\leq0$.
	
	\paragraph{Lower bound on $\optcrc$.} Consider the randomized contract $(\mu, \pi)$ where
	\[
	\mu(a) = \begin{cases}
		1/3 & \textnormal{if } a \neq \atwo\atwo \\
		0 &  \textnormal{if } a \neq \atwo\atwo,
	\end{cases} \quad\quad\textnormal{and}\quad\quad \pi^i_a(\omega) = \begin{cases}
		2	& \textnormal{if } a_i = \aone\,\wedge\, a_{-i}=\atwo\,\wedge\, \omega=\omega_1 \\
		0 & \textnormal{otherwise}
	\end{cases}\quad\forall i\in N.
	\]
	First, notice that the randomized contract $(\mu, \pi)$ is IC. In particular, the IC constraint when action $\aone$ is recommended to the first agent is:
	\begin{align*}
		\sum_{a_2\in A_2}& \mu(\aone, a_2) \left(U_1^{\mu, \pi}(\aone) - U_1^{\mu, \pi}(\aone\to\atwo) \right) \\
		&= \sum_{a_2\in A_2}\mu(\aone,a_2)\sum_{\omega\in\Omega} \pi^i_{(\aone, a_2)}(\omega)\left(F_{(\aone, a_2)}(\omega) - F_{(\atwo, a_2)}(\omega)\right) - \sum_{a_2\in A_2} \mu(\aone, a_2)\left(c_1(\aone) - c_1(\atwo)\right) \\
		&= \mu(\aone,\atwo)\sum_{\omega\in\Omega} \pi^1_{(\aone,\atwo)}(\omega)\left(F_{(\aone,\atwo)}(\omega) - F_{(\atwo,\atwo)}(\omega)\right) - \sum_{a_2\in A_2} \mu(\aone, a_2)\left(c_1(\aone) - c_1(\atwo)\right) \\
		&= \frac{4}{15} - \frac{4}{15} \geq 0.
	\end{align*}
	Similarly, the IC constraint is trivially satisfied when agent $i$ is recommended to play action $\atwo$, since in such a case her expected payment is $0$. This proves that $(\mu, \pi)$ is IC for agent $1$ (and hence also for agent $2$, since the agents are symmetric). To conclude the proof notice that since $(\mu, \pi)$ is a feasible contract, \emph{i.e.,} a feasible solution to \cref{prog:optrc}, the following holds: 
	\begin{align*}
		\optcrc &\ge \sum_{a\in\A}\sum_{\omega\in\Omega} \mu(a) F_a(\omega)\left(r_\omega - \sum_{i\in N} \pi^i_a(\omega)\right) = \frac{1}{3} + \frac{1}{3}\left[\frac{2}{5} - \frac{4}{5}\right] + \frac{1}{3}\left[\frac{2}{5} - \frac{4}{5}\right] = \frac{1}{15} > 0.
	\end{align*}
	This concludes the proof.
\end{proof}

\lemmasup*
\begin{proof}
	Consider an instance of a principal-multi-agent problem where $N=\{1,2\}$, $A_1=A_2=\{\aone,\atwo\}$, and $\Omega=\{\omega_1,\omega_2,\omega_3,\omega_4\}$. The rewards are defined such that  $r_\omega = 0$ for all $\omega\in\{\omega_1,\omega_2,\omega_3\}$ and $r_{\omega_4} = 1$. The outcome probabilities are defined as in the following table:
	\begin{table}[H]
		\centering
		\begin{tabular}{|c | c c c c|}
			\hline
			$F_a(\omega)$ & $\omega_1$ & $\omega_2$ & $\omega_3$ & $\omega_4$ \\
			\hline
			$\aone\aone$ & 1 & 0 & 0 & 0 \\
			$\aone\atwo$ & 0 & 1 & 0 & 0 \\
			$\atwo\aone$ & 0 & 0 & 0 & 1 \\
			$\atwo\atwo$ & 0 & 0 & 1/2 & 1/2 \\
			\hline
		\end{tabular}
	\end{table}
	Finally, the action costs are such that $c_1(\aone) = c_1(\atwo) = c_2(\atwo) = 0$ and $c_2(\aone) = 1/4$.
	
	First, fix any $\eps \in (0,1]$ and let us consider the randomized contract $(\mu,\pi)$, where 
	\begin{align*}
					\mu(a)=
		\begin{cases}
			\epsilon &\text{if}\quad a=(\aone,\aone)\\
			1-\epsilon &\text{if}\quad a=(\atwo,\aone)\\
			0&\text{otherwise}
		\end{cases}
	\end{align*}
	%
	%$\mu(\aone,\aone) = \eps$, $\mu(\atwo,\aone) = 1-\eps$ and $\mu(\aone,\atwo) = \mu(\atwo,\atwo)=0$. 
	while the payment schemes $\pi$ are defined as
	\[
	\pi^2_a(\omega) = \begin{cases}
		\frac{1}{4\eps} & \textnormal{if } a = \aone\aone\,\wedge\, \omega=\omega_1 \\
		0 & \textnormal{otherwise}
	\end{cases}\quad\quad\textnormal{and}\quad\quad \pi^1_a(\omega) = 0 \quad\forall a\in\A,\,\forall\omega\in\Omega.
	\]
	Notice that the randomized contract $(\mu, \pi)$ is trivially IC for agent $1$, since she never gets payed and her action costs are always $0$. Now, consider the IC constraint for agent $2$ when she is recommended to play action $\aone$ and she deviates by playing action $\atwo$. Notice that there is no need to consider the opposite case since the principal never recommends to agent $2$ to play action $\atwo$. It holds that: 
	\begin{align*}
		U^{\mu, \pi}_2&( \aone) - U^{\mu, \pi}_{2} (\aone\shortto\atwo) \\ 
		&= \sum_{a_1\in A_1}\mu(a_1,\aone) \sum_{\omega\in\Omega}\pi^2_{a_1,\aone}(\omega)\left(F_{(a_1,\aone)}(\omega) - F_{(a_1,\atwo)}(\omega)\right) + \sum_{a_1\in A_1}\mu(a_1,\aone)\left(c_2(\atwo) - c_2(\aone)\right) \\
		&=  \frac{1}{4} - \frac{1}{4} = 0,
	\end{align*} 
	which implies that $(\mu, \pi)$ is IC. The utility achieved by the principal when committing to $(\mu, \pi)$ is 
	\begin{align*}
		\sum_{a\in\A}\sum_{\omega\in\Omega}\mu(a)F_a(\omega)\left(r_\omega - \sum_{i\in N} \pi^i_a(\omega)\right) &= \mu(\atwo,\aone)F_{\atwo.\aone}(\omega_4) r_{\omega_4} - \mu(\aone,\aone) F_{\aone,\aone}(\omega_1)\pi^2_{\aone,\aone}(\omega_1) \\
		&= 1 - \eps - \eps \frac{1}{4\eps} = \frac{3}{4} - \eps.
	\end{align*}
	This proves that $\optcrc\ge\frac 34$.

	To conclude the proof, we need to show that no feasible randomized contract that can guarantee a principal's utility equal to $3/4$. To this extent, notice that by  IC  with respect to action $\atwo$ the expected payment received by an agent must be at least equal to the expected cost (otherwise the agent would play $\atwo$).
	Formally, 
	\[
	\sum_{a\in\A}\sum_{\omega\in\Omega} \mu(a)F_a(\omega)\pi^i_a(\omega) \geq \sum_{a\in\A}\mu(a) c_i(a_i) \quad\forall i\in N.
	\]
	Thus, we can upper bound the principal's expected utility as follows: 
	\begin{align*}
		\sum_{a\in\A}\sum_{\omega\in\Omega}\mu(a)F_a(\omega)\left(r_\omega - \sum_{i\in N} \pi^i_a(\omega)\right) 
		&\leq \sum_{a\in\A}\mu(a)\left(\sum_{\omega\in\Omega}F_a(\omega)r_\omega - \sum_{i\in N} c_i(a_i)\right)\\
		&= -\frac{1}{4}\mu(\aone,\aone) + \frac{3}{4}\mu(\atwo,\aone) + \frac{1}{2}\mu(\atwo,\atwo).
	\end{align*}
	As a consequence, the only way to guarantee a principal's utility equal to $3/4$ would be to select a contract $(\bar\mu, \bar\pi)$ such that $\bar\mu(\atwo,\aone) = 1$. However, notice that, in such a case, the IC constraint for agent $2$ would require the following
	\begin{align*}
		0 & \leq U_2^{\bar\mu, \bar\pi}(\aone) - U_2^{\bar\mu, \bar\pi}(\aone\to\atwo) \tag{IC}\\
		&= \sum_{\omega\in\Omega}\bar\pi^2_{(\atwo,\aone)}(\omega)\left(F_{(\atwo,\aone)}(\omega) - F_{(\atwo,\atwo)}(\omega)\right) + c_2(\atwo) - c_2(\aone)\tag{$\bar\mu(\atwo,\aone)=1$} \\
		&= \frac{1}{2}\bar\pi^2_{(\atwo,\aone)}(\omega_4)-\frac{1}{2}\bar\pi^2_{(\atwo,\aone)}(\omega_3) - \frac{1}{4}.  
	\end{align*}
	Rearranging, we get 
	\[
	\bar\pi^2_{(\atwo,\aone)}(\omega_4) \geq \frac{1}{2} + \bar\pi^2_{(\atwo,\aone)}(\omega_3) \geq \frac{1}{2},
	\]
	which implies that the expected utility for the principal under such a randomized contract is
	\begin{align*}
	\sum_{a\in\A}\sum_{\omega\in\Omega}\bar\mu(a)F_a(\omega)\left(r_\omega - \sum_{i\in N}\bar\pi^i_a(\omega)\right) &\leq \bar\mu(\atwo,\aone)\left(1- \bar\pi_{(\atwo,\aone)}^i(\omega_4)\right)\\
	&\leq \frac{1}{2},
	\end{align*}
	where in the last inequality we used that $\bar\mu(\atwo,\aone)=1$ and $\bar\pi^2_{(\atwo,\aone)}\ge 1/2$.
	This shows that no randomized contract that is IC can attain an expected utility for the principal equal to $3/4$, thus giving the desired result.	
\end{proof}

\lemmarelaxation*
\begin{proof}
	In order to prove the result, it is enough to show that for each $(\mu, \pi)$ feasible for \Cref{prog:lprc}, there exists a $(\mu, x)$ that is feasible for $\lprc(\AindR)$ achieving the same  value. To this extent, let $(\mu, \pi)$ be a feasible solution for  \Cref{prog:lprc} and define $x$ as:
	\[
	x_a^i(\omega)\coloneqq\mu(a)x_a^i(\omega)\,\,\forall i\in N, \forall a\in\A, \omega\in \Omega.
	\]
%	
%	$x$ be such that for each $i\in N$, for each $a\in\A$ and for each $\omega\in\Omega$, $x_a^i(\omega) = \mu(a)\pi^i_a(\omega)$. 
	First, let us focus on the objective function and notice that by definition of $x$ it holds that:
	\begin{align*}
		\sum_{a\in\A}\sum_{\omega\in\Omega} \left(\mu(a) F_a(\omega)r_\omega- \sum_{i\in N} \mu(a)F_a(\omega)\pi_a^i(\omega)\right) &= \sum_{a\in\A}\sum_{\omega\in\Omega} \left(\mu(a) F_a(\omega)r_\omega- \sum_{i\in N} F_a(\omega)x_a^i(\omega)\right),
	\end{align*}
	which proves the equivalence in the objective function of \Cref{prog:optrc} in $(\mu,\pi)$ and \Cref{prog:lprc} in $(\mu, x)$.
	
	Now, let us prove the feasibility of $(\mu, x)$ for \Cref{prog:lprc}.  For all $i\in N$ and for all $a_i, a_i^\prime\in A_i$, it holds that
	\begin{subequations}
	\begin{align}
		\sum_{a_{-i}\in\A_{-i}}&\left[\sum_{\omega\in\Omega}\left(x^i_{a_i, a_{-i}}(\omega)F_{a_i, a_{-i}}(\omega)\right) - \mu(a_i, a_{-i})c_i(a_i)\right]\notag \\
		&= \sum_{a_{-i}\in\A_{-i}} \mu(a_i, a_{-i}) \left(\sum_{\omega\in\Omega}F_{(a_i, a_{-i})}(\omega)\pi^i_{(a_i, a_{-i})}(\omega) - c_i(a_i)\right)\label{eq:relax_defx1} \\
		&\geq \sum_{a_{-i}\in\A_{-i}} \mu(a_i, a_{-i}) \left(\sum_{\omega\in\Omega}F_{(a_i^\prime, a_{-i})}(\omega)\pi^i_{(a_i, a_{-i})}(\omega) - c_i(a_i^\prime)\right)\label{eq:relax_feasib} \\
		&= \sum_{a_{-i}\in\A_{-i}}\left[\sum_{\omega\in\Omega}\left(x^i_{a_i, a_{-i}}(\omega)F_{a_i^\prime, a_{-i}}(\omega)\right) - \mu(a_i, a_{-i})c_i(a_i^\prime)\right]\label{eq:relax_defx2},
	\end{align}
	\end{subequations}
	where Equations \eqref{eq:relax_defx1} and \eqref{eq:relax_defx2} follow from the definition of $x$, while \cref{eq:relax_feasib} follows from feasibility of $(\mu, \pi)$. Hence, we can conclude that Constraint \eqref{eq:lprc_ic} is satisfied by $(\mu, x)$.
	Additionally, notice that Constraints~\eqref{eq:lprc_bound} and~\eqref{eq:lprc_defmu} are trivially satisfied by $(\mu, x)$. 
	Finally, by definition of the set of inducible actions $\AindR$ we have that for every non inducible action profile $a\in \A\setminus\AindR$ we have that $\mu(a)=0$ and thus by definition of $x$:
	\[
	x_a^i(\omega)=0,\quad\forall a\in\A\setminus\AindR,\forall i\in N, \forall\omega\in\Omega.
	\]
%	
%	and $x$, for each non inducible action $a\in\A\setminus\Aind$, for each $i\in N$ and for each $\omega\in\Omega$, it holds that $x^i_a(\omega) = 0$ and $\mu(a) = 0$. 
%	Noticing that since $\A^\prime\supseteq\Aind$, then $\A\setminus\A^\prime\subseteq \A\setminus\Aind$, it follows that $(\mu, x)$ satisfy also constraints \eqref{eq:lprc_restr1} and \eqref{eq:lprc_restr2}.
%	
	which proves that Constraints~\eqref{eq:lprc_restr1} and~\eqref{eq:lprc_restr2} are satisfied.
	Thus, $(\mu,x)$ is feasible for \Cref{prog:lprc} and this concludes the proof.
\end{proof}

%In order to prove \cref{le:lemmatrans}, we will leverage the following known result about the boundedness of vertices of polytopes.
%
%\begin{lemma}[Lemma 8.2 \citep{bertsimas1997introduction}]\label{le:lemmacomplexity}
%	Let $A$ be an $m\times n$ integer matrix and let $b$ be a vector in $\Reals^m$. Let $U$ be the largest absolute value of the entries in $A$ and $b$. Then, every extreme point of the polyhedron $P=\left\{x\in\Reals^n\mid Ax\geq b\right\}$ satisfies 
%	\[
%		|x_j| \leq (nU)^n\quad \forall j\in [n].
%	\]
%\end{lemma}

\lemmatransformation*
\begin{proof}
%	\ma{$=\frac{2 |\AindR|^2 N \tau^3}{\eps}$}
%	Fix $\AindR\subseteq \Aind$ and $(\mu, x)$ feasible for $\lprc(\AindR)$. 
	To prove this lemma, we first construct $(\bar\mu, \bar x)$ starting from $(\mu, x)$. Then, we show that $(\bar\mu, \bar x)$ is feasible for $\lprc(M, \AindR)$. Finally, we conclude the proof bounding the gap of the objective values achieved by $(\mu, x)$ and $(\bar\mu, \bar x)$.
	
	\paragraph{Definition of $(\bar\mu, \bar x)$.}
	For all $a\in\AindR$, let us introduce the following linear program, which we denote as $\lpaux(a)$:
	\begin{program}\label{eq:lpaux_trans}
		\begin{align} 
			&\max_{\mu, x} \,\, \mu(a) \hspace{10cm}\textnormal{s.t.} \nonumber\\
			&\sum_{a_{-i}\in\A_{-i}}\left[\sum_{\omega\in\Omega}\left(x^i_{(a_i, a_{-i})}(\omega)F_{(a_i, a_{-i})}(\omega)\right) \hspace{-0.05cm}-\hspace{-0.05cm} \mu(a_i, a_{-i})c_i(a_i)\right] \nonumber \\
			&\hspace{1.05cm}\geq \hspace{-0.2cm}\sum_{a_{-i}\hspace{-0.05cm}\in\A_{-i}}\left[\sum_{\omega\in\Omega}\left(x^i_{(a_i, a_{-i})}(\omega)F_{(a_i^\prime, a_{-i})}(\omega)\right) \hspace{-0.05cm}-\hspace{-0.05cm} \mu(a_i, a_{-i})c_i(a_i^\prime)\right]\hspace{0.3cm} \forall i\in N,\forall a_i,a_i^\prime\in A_i \label{eq:lpaux_trans_ic}\\
			& \mu(a) = 0\hspace{9.2cm} \forall a\in \A\setminus\AindR\label{eq:lpaux_trans_restr1} \\
			& x_a^i(\omega) = 0\hspace{7.cm} \forall i\in N, a\in\A\setminus\AindR,\omega\in\Omega\label{eq:lpaux_trans_restr2}\\
			& x_a^i(\omega)\geq 0\hspace{7.4cm} \forall i\in N,\forall a\in\A,\forall\omega\in\Omega\label{eq:lpaux_trans_defx} \\
			& \mu \in \Delta(\A)\label{eq:lpaux_trans_defmu}.
		\end{align}
	\end{program}
	We make the following considerations about the bit-complexity of the solutions of the LPs $\lpaux(a)$.
	Denote with $|I|$ the instance size of the principal-multi-agent problem, and for all $\bar a \in\AindR$ define as $(\bar \mu^{\bar a}, \bar x^{\bar a})$ the solution of $\lpaux(\bar a)$. 
	Note that we are solving $\lpaux(\bar a)$ only for $\bar a\in\AindR$. Thus, there are feasible solutions to \Cref{prog:optrc} with $\mu(\bar a)>0$  and hence to $\lpaux(\bar a)$ with $\mu(\bar a)>0$.
	Then, it is well known that for all $\bar a \in \AindR$ there exists a solution $(\bar \mu^{\bar a}, \bar x^{\bar a})$ such that $\lVert \bar x^{\bar a}\rVert_\infty\le \tau(|I|)$ and $\bar \mu^{\bar a}(\bar a)\ge \tau(|I|)^{-1}$, where $\tau:\mathbb{N}\to\mathbb{R}^+$ is a function such that $\tau (x)\in 2^{\poly(x)}$ \citep{bertsimas1997introduction}.
	 %Also note that we are solving $\lpaux(\bar a)$ only for $\bar a\in\AindR$. Thus, there are feasible solutions to \Cref{prog:optrc} with $\mu(\bar a)>0$  and hence to $\lpaux(\bar a)$. As before we can show that $\mu^{\bar a}(\bar a)\ge 1/\tau(|I|)$. \mat{sbagliato!}
	 
	 We define $(\bar \mu,\bar x)$ as follows:
	\begin{align}
		&\bar \mu(a) := \alpha \mu(a) + (1-\alpha) \frac{1}{|\AindR|}\sum_{\bar a\in \AindR}\mu^{\bar a}(a)&\forall a\in\A\label{eq:defmubar_trans}\\
		&\bar x_a^i(\omega) := \alpha x_a^i(\omega) + (1-\alpha) \frac{1}{|\AindR|}\sum_{\bar a\in \AindR}x^{\bar a, i}_a(\omega)&\forall i\in N,\,\forall a\in\A,\,\forall\omega\in\Omega\label{eq:defxbar_trans},
	\end{align}
	where $(\mu,x)$ is a feasible solution to $\lprc(\AindR)$ and $\alpha:=\frac{2|\AindR|n\tau(|I|)-\eps}{2|\AindR|n\tau(|I|)}\in[0,1]$.

	\paragraph{Feasibility of $(\bar\mu, \bar x)$.} We now proceed to prove that $(\bar\mu, \bar x)$ is feasible for $\lprc(M, \AindR)$, where $M:={2 |\A|^2 n \tau(|I|)^3K\eps^{-1}}{}$.
	First, note that for each $\bar a\in\AindR$, constraints \eqref{eq:lpaux_trans_ic}-\eqref{eq:lpaux_trans_restr2} and \eqref{eq:lpaux_trans_defmu} in $\lpaux(\bar a)$ are equivalent to the constraints \eqref{eq:lprc_ic}, and \eqref{eq:lprc_restr1}-\eqref{eq:lprc_defmu} of $\lprc(M, \AindR)$. This shows that Constraints \eqref{eq:lprc_ic} and \eqref{eq:lprc_restr1}-\eqref{eq:lprc_defmu} of $\lprc(M, \AindR)$ are satisfied by $(\mu^{\bar a}, x^{\bar a})$ for every $\bar a\in\AindR$. $(\mu, x)$ satisfies the same constraints by construction (it is the solution of $\lprc(\AindR)$ that includes the same constraints). 
	Then, the same constraints of $\lprc(M, \AindR)$ are satisfied by $(\bar \mu, \bar x)$, as it is a convex combination of $(\mu, x)$ and $(\mu^{\bar a}, x^{\bar a})_{\bar a\in\AindR}$.
	Now, we need to verify that Constraint~\eqref{eq:lprc_bound} of $\lprc(M, \AindR)$ is satisfied by $(\bar \mu,\bar x)$. This can be shown by the following inequalities:
	\begin{align*}
		\frac{\bar x_a^i(\omega)}{\bar\mu(a)} &:= \frac{\alpha x_a^i(\omega) + (1-\alpha) \sum_{\bar a\in \AindR}x^{\bar a,i}_a(\omega)/|\AindR|}{\mu(a) + (1-\alpha) \sum_{\bar a\in \AindR}\mu^{\bar a}(a)/|\AindR|}\\
		&\leq \frac{\alpha K + (1-\alpha) \sum_{\bar a\in \AindR}x^{\bar a,i}_a(\omega)/|\AindR|}{(1-\alpha) \sum_{\bar a\in \AindR} \mu^{\bar a}(a)/|\AindR|} \tag{$x_a^i(\omega)\le K$ and $\mu(a)\ge 0$}\\
		&\leq \frac{\tau(|I|)^2 K |\AindR|}{(1-\alpha)} \tag{$\sum_{\bar a\in\AindR}\mu^{\bar a}(a)\ge 1/\tau(|I|)$ and $x_a^{\bar a, i}(\omega)\le \tau(|I|)$}\\
		&= \frac{2 |\AindR|^2 N \tau(|I|)^3K}{\eps} \tag{Definition of $\alpha$}\\
		&= M.\tag{Definition of $M$ and $\AindR\subseteq \A$}
	\end{align*}
	This proves that $(\bar\mu,\bar x)$ is feasible for $\lprc(M, \AindR)$.
	
	\paragraph{Gap in the Objective Functions.} We now study the relationship in the objective function of the just defined $(\bar\mu,\bar x)$ and $(\mu, x)$. In particular, we will show that $(\bar\mu,\bar x)$ is only slightly worst (order of $O(\eps)$) than $(\mu, x)$.
	
	By definition of $(\bar\mu,\bar x)$ we have:
	\begin{align*}
		\sum_{a\in\A}\sum_{\omega\in\Omega}F_a(\omega)&\left(\bar\mu(a)r_{\omega} - \sum_{i\in N} \bar x_a^i(\omega)\right)\notag \\
		&:= \underbrace{\alpha\left[\sum_{a\in\A}\sum_{\omega\in\Omega}F_a(\omega)\left(\mu(a)r_{\omega} - \sum_{i\in N} x_a^i(\omega)\right)\right] }_{(A)}\\
		&+ \underbrace{(1-\alpha)\left[\sum_{a\in\A}\sum_{\omega\in\Omega}F_a(\omega)\frac{1}{|\AindR|}\sum\limits_{\bar a\in\AindR}\left(\mu^{\bar a}(a)r_{\omega} - \sum_{i\in N}  x_a^{\bar a,i}(\omega)\right)\right]}_{(B)}. \nonumber
	\end{align*}
	
	We analyze the two terms separately. For $(A)$, define $V:=\sum_{a\in\A}\sum_{\omega\in\Omega}F_a(\omega)\left(\mu(a)r_{\omega} - \sum_{i\in N} x_a^i(\omega)\right)$ and consider the following inequalities:
	\begin{align*}
		(A)&=\alpha V\\
%		&= V -(1-\alpha)V\\
		&\ge V-(1-\alpha)\tag{$V\le 1$}\\
		&= V - \frac{\eps}{2|\AindR|N\tau(|I|)}\tag{Definition of $\alpha$}\\
		&\ge V-\eps/2.
	\end{align*}
	
	For $(B)$, consider the following inequalities:
	\begin{align*}
		(B)&=(1-\alpha) \sum\limits_{a\in\AindR}\sum\limits_{\omega\in\Omega}\frac{F_a(\omega)}{|\AindR|}\sum\limits_{\bar a\in\AindR}\left(\mu^{\bar a}(a)r_\omega-\sum\limits_{i\in N}x_a^{\bar a,i}(\omega)\right)\tag{$\mu^{\bar a}(a)=x_a^{\bar a,i}(\omega)=0,\forall a\in\A\setminus \AindR$}\\
		&\ge-(1-\alpha) \sum\limits_{a\in\AindR}\sum\limits_{\omega\in\Omega}F_a(\omega)\frac{1}{|\AindR|}\sum\limits_{\bar a\in\AindR}\sum\limits_{i\in N}x_a^{\bar a,i}(\omega)\tag{$r_\omega\ge 0$}\\
		&\ge -(1-\alpha)\sum\limits_{a\in\AindR}\sum\limits_{\omega\in\Omega}F_a(\omega)\sum\limits_{i\in N}\tau(|I|)\tag{$\tau(|I|)\ge x_a^{\bar a,i}$}\\
		&= -(1-\alpha)|\AindR|n\tau(|I|)\tag{$F_a\in\Delta(\Omega)$}\\
		&=-\eps/2\tag{Definition of $\alpha$},
	\end{align*}
	
	Which proves that 
	\[
	(A)+(B)=\sum_{a\in\A}\sum_{\omega\in\Omega}F_a(\omega)\left(\bar\mu(a)r_{\omega} - \sum_{i\in N} \bar x_a^i(\omega)\right)\ge \sum_{a\in\A}\sum_{\omega\in\Omega}F_a(\omega)\left(\mu(a)r_{\omega} - \sum_{i\in N}  x_a^i(\omega)\right)-\eps,
	\]
	as desired.
\end{proof}

\theoremRandom*

\begin{proof}
	Let $(\mu, x)$ be an optimal solution to $\lprc(M, \A)$ where $M=M(|I|,\varepsilon, K)$ is defined as in \Cref{le:lemmatrans} and $K=\tau({|I|})$. Note that $(\mu, x)$ can be find in time polynomial in   $\poly(|I|, \log(M))=\poly(|I|, \varepsilon^{-1})$ solving the linear program $\lprc(M, \A)$. Indeed, it is well known that there exists a solution $( \mu,  x)$ such that $\lVert  x\rVert_\infty\le \tau(|I|)$, where $\tau:\mathbb{N}\to\mathbb{R}^+$ is the same function of \Cref{le:lemmatrans} such that $\tau (x)\in 2^{\poly(x)}$, and this solution can be found in polynomial time \citep{bertsimas1997introduction}. 
	
	%Let $(\mu, x)$ be a solution to $\lprc(\A, M)$, \ie~it achieves $\opt_{\lprc(\A, M)}$.
	By using \Cref{le:lemmarelax} and \Cref{le:lemmatrans}, the following inequalities hold:
	\begin{align}
		\opt_{\lprc(M, \A)}&\ge \opt_{\lprc(M, \AindR)}\tag{$\AindR\subseteq\A$}\\
		&\ge \opt_{\lprc(\AindR)}-\varepsilon\tag{\Cref{le:lemmatrans}}\\
		&\ge \opt_{\mathsf{R}}-\varepsilon\tag{\Cref{le:lemmarelax}}.
	\end{align}
	
	We used \Cref{le:lemmatrans} with $K=\tau(|I|)$ as we apply it to the optimal solution of $\lprc(\AindR)$, which is guaranteed to be bounded by $K=\tau(|I|)$.
	Then,  define $(\mu, \pi)$ from $(\mu, x)$, where:
	\[
	\pi^i_a(\omega) \coloneq \begin{cases}
		\frac{x^i_a(\omega)}{\mu(a)} & \textnormal{if } \mu(a) > 0 \\
		0 & \textnormal{otherwise}
	\end{cases}\quad\forall i\in N,\,\forall a\in\A,\,\forall\omega\in\Omega.
	\]
	It follows that:
	\[
	\sum_{a\in\A}\sum_{\omega\in\Omega} \mu(a) F_a(\omega) \left( r_\omega - \sum_{i\in N} \pi^i_a(\omega)\right) = \sum_{a\in\A}\sum_{\omega\in\Omega}F_a(\omega)\left(\mu(a) r_{\omega} - \sum_{i\in N} x_a^i(\omega)\right)\eqqcolon\opt_{\lprc(M, \A)}.
	\]
	This, together with the previous inequalities, shows that:
	\[
	\sum_{a\in\A}\sum_{\omega\in\Omega} \mu(a) F_a(\omega) \left( r_\omega - \sum_{i\in N} \pi^i_a(\omega)\right)\ge \opt_\mathsf{R}-\varepsilon.
	\]
	
	Finally, note that $(\mu,\pi)$ is feasible for \Cref{prog:optrc} since when $\mu(a)=0$ then $x_a^i(\omega)=0$ for all $i\in N$ and $\omega\in\Omega$. This guarantees that $x^i_a(\omega)=\mu(a)\pi^i_a(\omega) $ for each $i \in N$, $a \in \A$, and $\omega \in \Omega$. 
%	
	%	
	%	Let $(\tilde\mu, \tilde x)$ be an optimal solution to $\lprc(\Aind)$, \emph{i.e.,} such that
	%	\[
	%		\sum_{a\in\A}\sum_{\omega\in\Omega}F_a(\omega)\left[\tilde\mu(a)r_{\omega} - \sum_{i\in N} \tilde x_a^i(\omega)\right] = \opt_{\lprc(\Aind)}. 
	%	\]
	%	By \cref{le:lemmatrans}, there exists $(\bar\mu, \bar x)$ feasible for $\lprc(\Aind, M)$ such that 
	%	\[
	%		\sum_{a\in\A}\sum_{\omega\in\Omega}F_a(\omega)\left[\bar\mu(a) r_{\omega} - \sum_{i\in N} \bar x_a^i(\omega)\right] \geq \opt_{\lprc(\Aind)} - \eps. 
	%	\]
	%	Furthermore, notice that, trivially, it holds that $\textsc{Opt}_{\lprc(\A, M)} \geq \textsc{Opt}_{\lprc(\Aind, M)}$, since $\A\supseteq\Aind$. Since, by definition, $(\mu, x)$ is optimal for $\lprc(\A, M)$, it holds that
	%	\begin{align*}
		%		\sum_{a\in\A}\sum_{\omega\in\Omega}F_a(\omega)\left[\mu(a) r_{\omega} - \sum_{i\in N} x_a^i(\omega)\right] &= \textsc{Opt}_{\lprc(\A, M)} \\
		%		&\geq \textsc{Opt}_{\lprc(\Aind, M)} \\
		%		&\geq \sum_{a\in\A}\sum_{\omega\in\Omega}F_a(\omega)\left[\bar\mu(a) r_{\omega} - \sum_{i\in N} \bar x_a^i(\omega)\right] \\
		%		&\geq \opt_{\lprc(\Aind)} - \eps \\
		%		&\geq \optcrc - \eps, 		
		%	\end{align*}
	%	where the last equation follows from \cref{le:lemmarelax}. The result follows noticing that, by definition of $\pi$ (\cref{eq:backtracking}), it holds that
	%		\[
	%		\sum_{a\in\A}\sum_{\omega\in\Omega} \mu(a) F_a(\omega) \left[ r_\omega - \sum_{i\in N} \pi^i_a(\omega)\right] = \sum_{a\in\A}\sum_{\omega\in\Omega}F_a(\omega)\left[\mu(a) r_{\omega} - \sum_{i\in N} x_a^i(\omega)\right] \geq \optcrc - \eps.
	%	\]
\end{proof}

%% file: src/app_combinatorial.tex
\section{Proofs Omitted from Section \ref{sec:combinatorial}}

\lemmavirtual*
\begin{proof}
	Let us fix a payment function $p: \Omega\to\Reals^+$ for the single-agent problem and take $a\in \B^n(p)$. To prove the lemma we will show that $a\in E(\bar p)$.
	
	By definition of $\B^n(p)$, the following inequality holds for each $i\in N$ and for each $a_i^\prime\in A_i$:
	\[
	\sum_{\omega\in\Omega} F_a(\omega) p(\omega) - n\sum_{j\in N} c_j(a_j) \geq \sum_{\omega\in\Omega} F_{(a_i^\prime, a_{-i})}(\omega) p(\omega) - n\sum_{j\in N\setminus\{i\}} c_j(a_j) - nc_i(a_i^\prime),
	\]
	which in turn implies that: 
	\[
	\sum_{\omega\in\Omega} F_a(\omega) p(\omega) - n c_i(a_i) \geq \sum_{\omega\in\Omega} F_{(a_i^\prime, a_{-i})}(\omega) p(\omega) - nc_i(a_i^\prime).
	\]
	Dividing the above inequality by $n$, by definition of $\bar p$ we get:
	\[
	\sum_{\omega\in\Omega} F_a(\omega) \bar p^i(\omega) - c_i(a_i) \geq \sum_{\omega\in\Omega} F_{(a_i^\prime, a_{-i})}(\omega) \bar p^i(\omega) - c_i(a_i^\prime),
	\]
	which proves that $a\in E(\bar p)$ as desired. 
\end{proof}

\thvirtual*
\begin{proof}
	Let $(a, p)$ be an optimal contract for the single-agent problem with virtual costs increased by an $n$ factor, \emph{i.e.,} such that
	\[
	\sum_{\omega\in\Omega}F_{a}(\omega)[r_\omega - p(\omega)] = \optsing^n.
	\]
	Furthermore, let $\bar p=(\bar p^i)_{i\in N}$, where $\bar p^i(\omega) = p(\omega)/n$ for each $i\in N$ and for each $\omega\in\Omega$. By \cref{le:lemmavirtual}, it holds that $a\in E(\bar p)$, which is equivalent to saying that $(a, \bar p)$ is a feasible solution to \cref{prog:optdc}. As a consequence, we have that:
	\begin{align*}
		\optdc &\geq \sum_{\omega\in\Omega} F_a(\omega)\left(r_\omega - \sum_{i\in N} \bar p^i(\omega)\right) \\
		&=\sum_{\omega\in\Omega} F_a(\omega) \left(r_\omega - p(\omega)\right) \\
		&= \optsing^n,
	\end{align*}
	which yields the result.
\end{proof}

\propvirtual*
\begin{proof}
	Fix $\alpha > 0$, $\eps > 0$, and take $n= \lceil \eps^{-1/\alpha}\rceil$. Consider the principal-multi-agent problem with $N=\range{n}$ agents where, for each $i\in N$, the actions available are $A_i=\{\aone, \atwo\}$ with costs
	\[
	c_i(\aone) = \frac{1}{2n^{2-\alpha}}\quad\quad\textnormal{and}\quad\quad c_i(\atwo) = 0.
	\]
	The set of outcomes is $\Omega = \{\omega_1, \omega_2\}$, with $r_{\omega_1} = 1$ and $r_{\omega_2}$ = 0. The outcomes probabilities are defined so that 
	\[
	F_a(\omega_1) = 1 - \frac{\sum_{i\in N} \mathbbm{1}\left\{a_i = \atwo\right\}}{n}\quad\quad\forall a\in \A,
	\]
	where $\mathbbm{1}\{\cdot\}$ is the indicator function. Intuitively, the probability of inducing outcome $\omega_1$ decreases linearly with the number of agents playing the action with null cost $\atwo$.
	
	We prove the proposition in two steps. First, we show that $\optsing^{n^{1-\alpha}} \geq 1/2$, and then we show that $\optdc \leq \eps/2$, thus implying the desired result. 
	
	\paragraph{Lower bound on $\optsing^{n^{1-\alpha}}$.} Consider the payment function $p$ such that $p(\omega_1)=1/2$ and $p(\omega_2)=0$. Notice that the action profile $\hat a\in\A$, where $\hat a_i=\aone$ for all $i\in N$, is a best response to $p$ when the virtual cost function is $n^{1-\alpha}$, \emph{i.e.,} $\hat a\in B^{n^{1-\alpha}}(p)$. Indeed, we have that the utility for the agent when playing $\hat a$ is
	\[
	\sum_{\omega\in\Omega}F_{\hat a}(\omega)p(\omega) - n^{1-\alpha}\sum_{i\in N}c_i(\hat a_i) = \frac{1}{2} - n^{1-\alpha} \sum_{i\in N} \frac{1}{2n^{2-\alpha}} = 0.
	\] 
	On there other hand, for each action profile $a^\prime\in \A$, letting $\kappa = \sum_{i\in N} \mathbbm{1}\left\{a_i^\prime = \atwo\right\}$, we have that
	\[
	\sum_{\omega\in\Omega}F_{a^\prime}(\omega) p(\omega) - n^{1-\alpha}\sum_{i\in N} c_i(a_i) = \frac{1}{2}\left(1 - \frac{\kappa}{n}\right) - n^{1-\alpha} \frac{n-\kappa}{2n^{2-\alpha}} = 0.	
	\]
	Hence, we can conclude that
	\[
	\optsing^{n^{1-\alpha}} \geq \sum_{\omega\in\Omega}F_a(\omega)\left(r_{\omega} - p(\omega)\right) = \frac{1}{2}.
	\]
	
	\paragraph{Upper bound on $\optdc$.} Let $(\bar a, \bar p)$ be an optimal deterministic contract, \emph{i.e.,} a solution to \cref{prog:optdc}. By definition, it holds that
	\[
	\optdc = \sum_{\omega\in\Omega} F_{\bar a}(\omega)\left(r_{\omega} - \sum_{i\in N} \bar p^i(\omega)\right).
	\]
	Let $\None = \left\{i\in N: \bar a_i = \aone\right\}$, $\Ntwo = \left\{i\in N : \bar a_i = \atwo\right\}$ and $\kappa = |\Ntwo|$. If $\kappa = n$ (\emph{i.e.,} all agents play action $\atwo$), then the principal's utility is trivially at most $0$. Assume now $k < n$. The IC constraint imposes that for each $i\in \None$ the following must hold:
	\begin{align*}
		0 &\leq \sum_{\omega\in\Omega} \bar p^i(\omega)\left( F_{\bar a}(\omega) - F_{\atwo, \bar a_{-i}}(\omega) \right)- c_i(\aone) \\
		&= \left(1 - \frac{\kappa}{n}\right)\bar p^i(\omega_1) + \frac{\kappa}{n} \bar p^i(\omega_2) - \left(1 - \frac{\kappa+1}{n}\right)\bar p^i(\omega_1) - \frac{\kappa + 1}{n}\bar p^i(\omega_2) - \frac{1}{2n^{2-\alpha}}.
	\end{align*}
	Rearranging and using the fact that $\bar p^i(\omega_2) \geq 0$, we get
	\[
	\bar p^i(\omega_1) \geq \frac{1}{2 n^{1-\alpha}} \quad\forall i\in \None.
	\]
	Thus, the principal's utility satisfies the following upper bound:
	\begin{align*}	
		\optdc &= F_{\bar a}(\omega_1) \left(1 - \sum_{i\in\None}\bar p^i(\omega_1) - \sum_{i\in\Ntwo}\bar p^i(\omega_1) \right) - F_{\bar a}(\omega_2)\left(\sum_{i\in\None}\bar p^i(\omega_2) + \sum_{i\in\Ntwo}\bar p^i(\omega_2)\right) \\
		&\leq F_{\bar a}(\omega_1) \left(1 - \sum_{i\in\None}\bar p^i(\omega_1)\right) \\
		&\leq \left(1 - \frac{\kappa}{n}\right)\left(1 - \frac{n - \kappa}{2n^{1-\alpha}}\right).
	\end{align*}
	Then, we can conclude that the principal's utility satisfies: 
	\[
	\optdc \leq \max\left\{0, \max_{\kappa\in\range{n-1}}\left\{\left(1 - \frac{\kappa}{n}\right)\left(1 - \frac{n - \kappa}{2n^{1-\alpha}}\right)\right\}\right\}.
	\]
	To conclude the proof, it suffices to bound $\max_{\kappa\in\range{n-1}}\left\{\left(1 - \frac{\kappa}{n}\right)\left(1 - \frac{n - \kappa}{2n^{1-\alpha}}\right)\right\}$. Setting the derivative with respect to $\kappa$ equal to $0$, we get that the optimal unconstrained $\kappa^\star$ must satisfy 
	\[
	\frac{1}{n^{1-\alpha}} - \frac{1}{n} - \frac{\kappa^\star}{n^{2-\alpha}} = 0,
	\] 
	which gives $\kappa^\star = n - n^{1-\alpha}$, implying 
	\[
	\max_{\kappa\in\Reals}\left\{\left(1 - \frac{\kappa}{n}\right)\left(1 - \frac{n - \kappa}{2n^{1-\alpha}}\right)\right\} = \left(1 - \frac{n - n^{1-\alpha}}{n}\right)\frac{1}{2} = \frac{1}{2n^\alpha} \leq \frac{\eps}{2},
	\]
	where the last equation follows from the fact that, by definition of $n$, $n\geq \eps^{-1/\alpha}$. Then, we can conclude that
	\begin{align*}
		\optdc &\leq \max\left\{0, \max_{\kappa\in\range{n-1}}\left\{\left(1 - \frac{\kappa}{n}\right)\left(1 - \frac{n - \kappa}{2n^{1-\alpha}}\right)\right\}\right\} \\
		&\leq \max\left\{0, \max_{\kappa\in\Reals}\left\{\left(1 - \frac{\kappa}{n}\right)\left(1 - \frac{n - \kappa}{2n^{1-\alpha}}\right)\right\}\right\} \\
		&= \max\left\{0, \frac{\eps}{2}\right\}\\
		&\leq \frac{\eps}{2},
	\end{align*}
	as desired.
	This concludes the proof.
\end{proof}

\corollaryApprinsind*

\begin{proof}
	Consider the linear contract $p$ for the single-agent problem, where $p(\omega) = \frac{1}{1+\delta} r_\omega$ for each $\omega \in \Omega$. Let $a^\star$ be any action in $B^n(p)$, \emph{i.e.,} such that
	\[
	a^\star\in arg\max_{a\in\A} \left\{\sum_{\omega\in\Omega}F_a(\omega) \frac{r_\omega}{1+\delta}- n\sum_{i\in N}c_i(a_i)\right\}.
	\]
	Then, the following holds: 
	\begin{subequations}\label{corollaryAppinsindEq}
	\begin{align}
		\frac{1}{1+\delta}\sum_{\omega\in\Omega} F_{a^\star}(\omega)r_\omega &\geq \frac{1}{1+\delta}\sum_{\omega\in\Omega} F_{a^\star}(\omega)r_\omega - n\sum_{i\in N} c_i(a_i^\star)\\
		&= \max_{a\in\A} \left\{\sum_{\omega\in\Omega}F_a(\omega) \frac{r_\omega}{1+\delta}- n\sum_{i\in N}c_i(a_i)\right\} \\
		&= \frac{1}{1+\delta} \max_{a\in\A} \left\{\sum_{\omega\in\Omega}F_a(\omega) r_\omega- n(1+\delta)\sum_{i\in N}c_i(a_i)\right\} \\
		&=\frac{1}{1+\delta} \vsw^{n(1+\delta)}.
	\end{align}
	\end{subequations}
	Hence, the principal's utility in the single-agent problem is 
	\[
	\optsing^n \geq \sum_{\omega\in\Omega}F_{a^\star}(\omega)\left(r_\omega- p(\omega)\right) = \sum_{\omega\in\Omega}F_{a^\star}(\omega)\left(r_\omega - \frac{1}{1+\delta}r_\omega\right) = \frac{\delta}{1+\delta} \sum_{\omega\in\Omega}F_{a^\star}(\omega)r_\omega \geq \frac{\delta}{1+\delta}\vsw^{n(1+\delta)},
	\]
	where the first inequality follows from $ a^\star \in \B^n(p) $ and the last inequality from Equation~\eqref{corollaryAppinsindEq}.
	Then, the Corollary follows since by \cref{th:thvirtual} $\optdc\geq\optsing^n \ge  \frac{\delta}{1+\delta}\vsw^{n(1+\delta)}$.
	
	Moreover, note that the designed contract is linear.
\end{proof}

\leapprsw*

\begin{proof}
	Let $a^\star\in arg\max_{a\in\A}\left\{\sum_{\omega\in\Omega}F_a(\omega)r_\omega - c(a)\right\}$. Then, for all $\alpha > 0$, it holds that
	\begin{align*}
		\vsw^\alpha &= \max_{a\in\A} \left\{\sum_{\omega\in\Omega} F_a(\omega)r_\omega - \alpha c(a)\right\} \\
		&\geq \sum_{\omega\in\Omega} F_{a^\star}(\omega)r_\omega - \alpha c(a^\star) \\
		&= \left(\frac{\sum_{\omega\in\Omega} F_{a^\star}(\omega)r_\omega - \alpha c(a^\star)}{\sum_{\omega\in\Omega}F_{a^\star}(\omega)r_\omega - c(a^\star)}\right) \left(\sum_{\omega\in\Omega}F_{a^\star}(\omega)r_\omega - c(a^\star)\right)\\
		&= \frac{\sum_{\omega\in\Omega} F_{a^\star}(\omega)r_\omega - \alpha c(a^\star) + c(a^\star) - c(a^\star)}{\sum_{\omega\in\Omega}F_{a^\star}(\omega)r_\omega - c(a^\star)}\sw \\
		&= \left(1 - \frac{(\alpha - 1) c(a^\star)}{\sum_{\omega\in\Omega}F_{a^\star}(\omega)r_\omega - c(a^\star)}\right)\sw \\
		&\geq  (1 - (\alpha-1)\beta)\sw,
	\end{align*}
	which yields the result.
\end{proof}

%% file: src/app_bayoptcrc.tex
\section{Proof of Theorem \ref{th:thbayoptcrc}}\label{app:app_bayoptcrc}
In this section, we prove \cref{th:thbayoptcrc}. In particular, we prove that there exists a polynomial-time algorithm that finds an $\eps$-optimal randomized contract in Bayesian principal-multi-agent problems for any $\epsilon>0$. 

First, given a profile of types $\lambda\in\Lambda$, we define the set of action profiles that are inducible by a randomized contract when the types are $\lambda$, which is defined as
\[
	\BAindR(\lambda)\coloneqq \left\{a: \exists (\mu, \pi)\,\textnormal{s.t.}\, \mu^\lambda(a)>0,(\mu, \pi) \textnormal{ is feasible for Program \ref{progr:bay-optrc}}\right\}.
\]
Moreover, let $\BAindR \coloneqq (\BAindR(\lambda))_{\lambda\in\Lambda}$.
Similarly to our procedure for finding an approximately optimal randomized contract in non-Bayesian settings, we will rely on an optimal solution $(\mu, x, z)$ to $\baylp(M, \A)$ for some $M >0$, and then we will use this solution to obtain a randomize contract $(\mu, \pi)$ where $\pi$ satisfies the following:
\begin{equation}\label{eq:bayconversion_app}
	\pi^{\lambda, i}_a(\omega) = \begin{cases}
		\frac{x^{\lambda, i}_a(\omega)}{\mu^\lambda(a)} &\textnormal{if } \mu^\lambda(a) > 0 \\
		0 &\textnormal{otherwise} 
	\end{cases}\quad\quad\forall i\in N,\forall \lambda\in\Lambda,\forall a\in\A,\forall\omega\in\Omega.
\end{equation}

By definition, $(\mu, \pi)$ guarantees an expected utility to the principal equal to $\opt_{\baylp(M, \A)}$. 
In order to prove the theorem, we follow two steps. First, we relate  $\opt_{\baylp(M, \BAindR)}$ to $\opt_{\baylp(\BAindR)}$. 

\begin{lemma}\label{le:lemmatransbay}
	For each $\eps>0$ and for each $(\mu, x, z)$ feasible for $\baylp(\BAindR)$, there exists $(\bar\mu, \bar x, \bar z)$ feasible for $\baylp(M, \BAindR)$ such that
	\[
	\sum_{\lambda\in\Lambda}\sum_{a\in\A}\sum_{\omega\in\Omega}G(\lambda)F_a(\omega)\left(\bar\mu(a)r_{\omega} - \sum_{i\in N} \bar x_a^i(\omega)\right) \geq \sum_{\lambda\in\Lambda}\sum_{a\in\A}\sum_{\omega\in\Omega}G(\lambda)F_a(\omega)\left(\mu(a)r_{\omega} - \sum_{i\in N} x_a^i(\omega)\right) - \eps,
	\]
	where $M\coloneqq M(|I|,\varepsilon,K)\in\poly(2^{|I|}, \nicefrac1\varepsilon,K)$ and $K=\max_{i\in N, a \in \A, \omega \in \Omega,\lambda\in\Lambda} x^{i,\lambda}_a(\omega)$.
\end{lemma}
\begin{proof}
	To prove this lemma, we first construct $(\bar\mu, \bar x, \bar z)$ starting from $(\mu, x, z)$. Then, we show that $(\bar\mu, \bar x, \bar z)$ is feasible for $\baylp(M, \BAindR)$. Finally, we conclude the proof bounding the gap of the objective values achieved by $(\mu, x, z)$ and $(\bar\mu, \bar x, \bar z)$.
	
	\paragraph{Definition of $(\bar\mu, \bar x, \bar z)$.}
	For all $\lambda\in\Lambda$ and $a\in\BAindR(\lambda)$, let us introduce the following linear program, which we denote as $\lpaux(a, \lambda)$:
	\begin{program}\label{eq:lpaux_trans}
		\begin{align} 
			&\max_{\mu, x, z} \,\, \mu^\lambda(a) \hspace{1cm}\textnormal{s.t.}\hspace{9.7cm}\, \nonumber\\
			& \specialcell{\hspace{2mm}\sum_{a\in\A} x_a^{\lambda, i}(\omega)F_a^\lambda(\omega) - c_i^{\lambda_i}(a_i) \geq \sum_{a_i\in A_i} z^i(\lambda_i, \lambda_i^\prime, a_i)\hfill\forall i\in N,\,\forall \lambda\in\Lambda,\,\forall\lambda_i^\prime\in\Lambda_i}\label{eq:baylpaux_trans_ic}\\
			& \hspace{2mm}z^i(\lambda, \lambda_i^\prime, a_i) \nonumber\\
			&\specialcell{\hspace{.6cm} \geq\hspace{-.3cm}\sum_{a_{-i}\in\A_{-i}} \hspace{-.3mm}\sum_{\omega\in\Omega} \hspace{-.4mm} x^{(\lambda_i^\prime, \lambda_{-i}), i}_{(a_i, a_{-i})}\hspace{-.5mm}(\omega) F^{\lambda}_{(a_i^\prime, a_i)}\hspace{-.5mm}(\omega) - c_i^{\lambda_i}\hspace{-.5mm}(a_i^\prime)\hfill\forall i\in N,\forall\lambda\in\Lambda,\forall\lambda_i^\prime\in\Lambda_i,\forall a_i, a_i^\prime\in A_i} \label{eq:baylpaux_trans_max}\\
			&\specialcell{\hspace{2mm}\mu^\lambda(a) = 0  \hfill\forall\lambda\in\Lambda,\forall a\in \A\setminus\A^\prime(\lambda)}\label{eq:baylpaux_trans_restr1} \\
			&\specialcell{\hspace{2mm}x_a^{\lambda,i}(\omega) = 0  \hfill \forall i\in N,\forall\lambda\in\Lambda,\forall a\in\A\setminus\A^\prime(\lambda),\forall\omega\in\Omega}\label{eq:baylpaux_trans_restr2}\\
			&\specialcell{\hspace{2mm}x^{\lambda,i}_a(\omega) \geq 0 \hfill\forall i\in N,\forall \lambda\in\Lambda,\forall a\in\A,\forall\omega\in\Omega}\\
			&\specialcell{\hspace{2mm}\mu^{\lambda} \in \Delta(\A)\hfill\forall\lambda\in\Lambda}\label{eq:baylpaux_trans_defmu}.
		\end{align}
	\end{program}
	We make the following considerations about the bit-complexity of the solutions of  $\lpaux(a, \lambda)$.
	Denote with $|I|$ the instance size of the Bayesian principal-multi-agent problem, and for all $\bar\lambda\in\Lambda$ and $\bar a \in\BAindR(\bar\lambda)$ define as $(\bar \mu^{(\bar\lambda,\bar a)}, \bar x^{(\bar\lambda,\bar a)}, \bar z^{(\bar\lambda,\bar a)})$ the solution of $\lpaux(\bar a, \bar\lambda)$. 
	Note that we are solving $\lpaux(\bar a,\bar\lambda)$ only for $\bar a\in\AindR(\bar\lambda)$. Thus, there are feasible solutions to Program \ref{progr:bay-optrc} with $\mu^{\bar\lambda}(\bar a)>0$,  and hence solutions to $\lpaux(\bar a,\bar\lambda)$ with $\mu^{\bar\lambda}(\bar a)>0$.
	Then, it is well known that for all $\bar\lambda\in\Lambda$ and for all $\bar a \in \BAindR(\bar\lambda)$ there exists a solution $(\bar \mu^{(\bar\lambda,\bar a)}, \bar x^{(\bar\lambda,\bar a)}, \bar z^{(\bar\lambda,\bar a)})$ such that $\lVert \bar x^{(\bar\lambda,\bar a)}\rVert_\infty\le \tau(|I|)$ and $\bar \mu^{(\bar\lambda,\bar a), \bar\lambda}(\bar a)\ge \tau(|I|)^{-1}$, where $\tau:\mathbb{N}\to\mathbb{R}^+$ is a function such that $\tau (x)\in 2^{\poly(x)}$ \citep{bertsimas1997introduction}.

	We define $(\bar \mu,\bar x, \bar z)$ as follows:
	\begin{align}
		&\bar \mu^\lambda(a) := \alpha \mu^\lambda(a) + (1-\alpha) \sum_{\bar\lambda\in\Lambda}\frac{1}{|\Lambda||\BAindR(\lambda)|}\sum_{\bar a\in \BAindR(\lambda)}\mu^{(\bar\lambda,\bar a),\lambda}(a)\hspace{2.1cm}\forall a\in\A,\forall\lambda\in\Lambda,\label{eq:defmubar_transbay}\\
		&\specialcell{\bar x_a^i(\omega) := \alpha x_a^i(\omega) + (1-\alpha) \sum_{\bar\lambda\in\Lambda} \frac{1}{|\Lambda||\BAindR(\lambda)|}\sum_{\bar a\in \AindR}x^{(\bar\lambda,\bar a), i}_a(\omega)\hfill\forall i\in N,\,\forall a\in\A,\,\forall\omega\in\Omega}\label{eq:defxbar_transbay},\\
		&\bar z^i(\lambda, \lambda_i^\prime, a_i) := \alpha z^i(\lambda, \lambda_i^\prime, a_i) \nonumber\\
		&\specialcell{\hspace{2cm}+ (1-\alpha) \sum_{\bar\lambda\in\Lambda} \sum_{\bar a\in \AindR}\frac{\bar z^{(\bar\lambda,\bar a),i}(\lambda, \lambda_i^\prime, a_i)}{|\Lambda||\BAindR(\lambda)|}\hfill\forall i\in N,\forall a_i\in A_i,\forall\lambda\in\Lambda,\forall\lambda_i^\prime\in\Lambda_i,}\label{eq:defzbar_transbay}
	\end{align}
	where $(\mu,x, z)$ is a feasible solution to $\baylp(\BAindR)$ and $\alpha:=\frac{2|\A|n\tau(|I|)-\eps}{2|\A|n\tau(|I|)}\in[0,1]$.

	\paragraph{Feasibility of $(\bar\mu, \bar x, \bar z)$.} We now proceed to prove that $(\bar\mu, \bar x, \bar z)$ is feasible for $\baylp(M, \BAindR)$, where $M:={2 |\Lambda||\A|^2 n \tau(|I|)^3K\eps^{-1}}{}$.
	First, note that for each $\bar\lambda\in\Lambda$ and for each $\bar a\in\BAindR(\bar\lambda)$ Constraints \eqref{eq:baylpaux_trans_ic}-\eqref{eq:baylpaux_trans_restr2} and \eqref{eq:baylpaux_trans_defmu} in $\lpaux(\bar a)$ are equivalent to Constraints \eqref{eq:baylprc_ic}, \eqref{eq:baylprc_max}, and \eqref{eq:baylprc_restr1}-\eqref{eq:baylprc_defmu} of $\baylp(M, \BAindR)$. This shows that Constraints \eqref{eq:baylprc_ic} and \eqref{eq:baylprc_restr1}-\eqref{eq:baylprc_defmu} of $\baylp(M, \BAindR)$ are satisfied by $(\mu^{(\bar\lambda,\bar a)}, x^{(\bar\lambda,\bar a)}, z^{(\bar\lambda,\bar a)})$ for every $\lambda\in\Lambda$ and for every $\bar a\in\BAindR(\lambda)$. $(\mu, x, z)$ satisfies the same constraints by construction (it is the solution of $\baylp(\BAindR)$ that includes the same constraints). 
	Then, the same constraints of $\baylp(M, \BAindR)$ are satisfied by $(\bar \mu, \bar x, \bar z)$, as it is a convex combination of $(\mu, x, z)$ and $(\mu^{(\bar\lambda,\bar a)}, x^{(\bar\lambda,\bar a)}, z^{(\bar\lambda,\bar a)})_{\bar\lambda\in\Lambda, \bar a\in\BAindR(\lambda)}$. Now, we need to verify that also Constraint~\eqref{eq:baylprc_bound} of $\baylp(M, \BAindR)$ is satisfied by $(\bar \mu,\bar x, \bar z)$. This can be showed by the following inequalities:
	\begin{align*}
		&\frac{\bar x_a^{\lambda, i}(\omega)}{\bar\mu^\lambda(a)} := \frac{\alpha x_a^{\lambda, i}(\omega) + (1-\alpha) \sum_{\bar\lambda\in\Lambda}\sum_{\bar a\in \BAindR(\bar\lambda)}x^{(\bar\lambda,\bar a),\lambda,i}_a(\omega)/(|\Lambda||\BAindR(\bar\lambda)|)}{\mu^\lambda(a) + (1-\alpha)\sum_{\bar\lambda\in\Lambda} \sum_{\bar a\in \BAindR(\bar\lambda)}\mu^{(\bar\lambda,\bar a),\lambda}(a)/(|\Lambda||\BAindR(\bar\lambda)|)}\hspace{3.5cm}\,\\
		&\specialcell{\hspace{4mm}\leq \frac{\alpha K + (1-\alpha) \sum_{\bar\lambda\in\Lambda}\sum_{\bar a\in \BAindR(\bar\lambda)}x^{(\bar\lambda,\bar a),\lambda,i}_a(\omega)/(|\Lambda||\BAindR(\bar\lambda)|)}{(1-\alpha)\sum_{\bar\lambda\in\Lambda} \sum_{\bar a\in \BAindR(\bar\lambda)}\mu^{(\bar\lambda,\bar a),\lambda}(a)/(|\Lambda||\BAindR(\bar\lambda)|)} \hfill (x_a^{\lambda,i}(\omega)\le K \textnormal{ and } \mu^\lambda(a)\ge 0)}\\
		&\specialcell{\hspace{4mm}\leq \frac{|\Lambda||\A|(\alpha K + (1-\alpha) \tau(|I|))}{(1-\alpha)\sum_{\bar\lambda\in\Lambda} \sum_{\bar a\in \BAindR(\bar\lambda)}\mu^{(\bar\lambda,\bar a),\lambda}(a)} \hfill (x_a^{(\bar\lambda,\bar a), \lambda, i}(\omega) \leq \tau(|I|) \textnormal{ and } \BAindR(\bar\lambda)\subseteq\A)}\\
		&\specialcell{\hspace{4mm}\leq \frac{\tau(|I|)^2K|\Lambda||\A|}{(1-\alpha)}\hfill \makebox[0.5\linewidth][r]{\text{($\sum_{\bar\lambda\in\Lambda} \sum_{\bar a\in \BAindR(\bar\lambda)}\mu^{(\bar\lambda,\bar a),\lambda}(a)\geq 1/\tau(|I|)$)}}}\\
		&\specialcell{\hspace{4mm}=\frac{2 \tau(|I|)^3nK|\Lambda||\A|^2}{\eps}\hfill (\textnormal{Definition of }\alpha)}\\
		&\specialcell{\hspace{4mm}=M \hfill (\textnormal{Definition of } M)}
	\end{align*}
	This proves that $(\bar\mu,\bar x, \bar z)$ is feasible for $\baylp(M, \BAindR)$.
	
	\paragraph{Gap in the Objective Functions.} We now study the relationship in the objective function of the just defined $(\bar\mu,\bar x,\bar z)$ and $(\mu, x, z)$. In particular, we will show that $(\bar\mu,\bar x,\bar z)$ is only slightly worse (order of $O(\eps)$)  than $(\mu, x, z)$.
	
	By definition of $(\bar\mu,\bar x, \bar z)$ we have:
	\begin{align*}
		&\sum_{\lambda\in\Lambda}\sum_{a\in\A}\sum_{\omega\in\Omega}G(\lambda)F_a^\lambda(\omega)\left(\bar\mu^\lambda(a)r_{\omega} - \sum_{i\in N} \bar x_a^{\lambda,i}(\omega)\right)\notag \\
		&:= \underbrace{\alpha\left[\sum_{\lambda\in\Lambda}\sum_{a\in\A}\sum_{\omega\in\Omega}G(\lambda)F_a^\lambda(\omega)\left(\mu^\lambda(a)r_{\omega} - \sum_{i\in N} x_a^{\lambda,i}(\omega)\right)\right] }_{(A)} \\
		&\hspace{6mm}+ \underbrace{(1-\alpha)\left[\sum_{\lambda\in\Lambda}\sum_{a\in\A}\sum_{\omega\in\Omega}G(\lambda)F_a^\lambda(\omega)\sum_{\bar\lambda\in\Lambda}\sum\limits_{\bar a\in\BAindR(\bar\lambda)}\frac{1}{|\Lambda||\BAindR(\bar\lambda)|}\left(\mu^{(\bar\lambda,\bar a),\lambda}(a)r_{\omega} - \sum_{i\in N}  x_a^{(\bar\lambda,\bar a),\lambda,i}(\omega)\right)\right]}_{(B)}. \nonumber
	\end{align*}
	We analyze the two terms separately.
	For $(A)$, define \[V:=\sum_{\lambda\in\Lambda}\sum_{a\in\A}\sum_{\omega\in\Omega}G(\lambda)F_a^\lambda(\omega)\left(\mu^\lambda(a)r_{\omega} - \sum_{i\in N} x_a^{\lambda, i}(\omega)\right),\] and consider the following inequalities:
	\begin{align*}
		(A)&=\alpha V\\
		%		&= V -(1-\alpha)V\\
		&\ge V-(1-\alpha)\tag{$V\le 1$}\\
		&= V - \frac{\eps}{2|\A|n\tau(|I|)}\tag{Definition of $\alpha$}\\
		&\ge V-\eps/2.
	\end{align*}
	
	For $(B)$, consider the following inequalities:
	\begin{align*}
		(B)&= (1-\alpha)\left[\sum_{\lambda\in\Lambda}\sum_{a\in\A}\sum_{\omega\in\Omega}G(\lambda)F_a^\lambda(\omega)\sum_{\bar\lambda\in\Lambda}\sum\limits_{\bar a\in\BAindR(\bar\lambda)}\frac{1}{|\Lambda||\BAindR(\bar\lambda)|}\left(\mu^{(\bar\lambda,\bar a),\lambda}(a)r_{\omega} - \sum_{i\in N}  x_a^{(\bar\lambda,\bar a),\lambda,i}(\omega)\right)\right]\\
		&\specialcell{\ge-(1-\alpha)\sum_{\lambda\in\Lambda}\sum_{a\in\A}\sum_{\omega\in\Omega}\sum_{i\in N}G(\lambda)F_a^\lambda(\omega)\sum_{\bar\lambda\in\Lambda}\sum\limits_{\bar a\in\BAindR(\bar\lambda)}\frac{x_a^{(\bar\lambda,\bar a),\lambda,i}(\omega)}{|\Lambda||\BAindR(\bar\lambda)|}  \hfill (r_\omega\ge 0)}\\
		&\specialcell{\ge -(1-\alpha)\sum_{\lambda\in\Lambda}\sum_{a\in\A}\sum_{\omega\in\Omega}\sum_{i\in N}G(\lambda)F_a^\lambda(\omega) \tau(|I|)\hfill (x_a^{(\bar\lambda,\bar a),\lambda,i}(\omega) \le \tau(|I|))} \\
		&\specialcell{\ge -(1-\alpha) |\A|n\tau(|I|)} \hfill\,\\
		&\specialcell{\ge -\eps/2,\hfill (\textnormal{Definition of }\alpha)}
	\end{align*}
	
	which proves that 
	\[
	(A)+(B)\ge \sum_{\lambda\in\Lambda}\sum_{a\in\A}\sum_{\omega\in\Omega}G(\lambda)F_a^\lambda(\omega)\left(\mu^\lambda(a)r_{\omega} - \sum_{i\in N}  x_a^{\lambda, i}(\omega)\right)-\eps,
	\]
	as desired.
\end{proof}

The result of \cref{le:lemmatransbay} implies that $\opt_{\baylp(M, \BAindR)} \ge \opt_{\baylp(\BAindR)}$. As the next step, we show that $\baylp(\BAindR)$ yields a relaxation of Program \ref{prog:baylprc}.

\begin{lemma}\label{le:lemmarelaxationbay}
	It holds that $\opt_{\baylp(\BAindR)}\ge \boptcrc$.
\end{lemma}
\begin{proof}
	In order to prove the result, it is enough to show that for each $(\mu, \pi)$ feasible to Program \ref{prog:baylprc}, there exists a $(\mu, x, z)$ that is feasible to $\baylp(\BAindR)$ achieving the same  value. To this extent, let $(\mu, \pi)$ be a feasible solution to Program \ref{prog:baylprc} and define $x$ and $z$ as:
	\begin{align*}
		&x_a^{\lambda, i}(\omega)\coloneqq\mu^\lambda(a)\pi_a^{\lambda, i}(\omega)\hspace{6cm}\,\forall i\in N,\forall\lambda\in\Lambda \forall a\in\A, \omega\in \Omega \\
		&\specialcell{z^i(\lambda, \lambda_i^\prime, a_i) \coloneqq \max_{a_i^\prime\in A_i} U_i^{\mu, \pi}\left(\lambda_i\to\lambda_i^\prime, a_i\to a_i^\prime\vert\lambda_{-i}\right) \hfill\forall i\in N,\forall\lambda\in\Lambda,\forall\lambda_i^\prime\in\Lambda_i,\forall a_i\in A_i}.
	\end{align*}
	First, let us focus on the objective function and notice that by definition of $x$ it holds that:
	\begin{align*}
		\sum_{\lambda\in\Lambda}&\sum_{a\in\A}\sum_{\omega\in\Omega} G(\lambda)\left(\mu^\lambda(a) F^\lambda_a(\omega)r_\omega- \sum_{i\in N} \mu^\lambda(a)F^\lambda_a(\omega)\pi_a^{\lambda, i}(\omega)\right) \\
		&= \sum_{\lambda\in\Lambda}\sum_{a\in\A}\sum_{\omega\in\Omega} G(\lambda)\left(\mu^\lambda(a) F^\lambda_a(\omega)r_\omega- \sum_{i\in N} F^\lambda_a(\omega)x_a^{\lambda, i}(\omega)\right),
	\end{align*}
	which proves the equivalence in the objective function of Program \ref{progr:bay-optrc} in $(\mu,\pi)$ and Program \ref{prog:baylprc} in $(\mu, x, z)$.
	
	Now, let us prove the feasibility of $(\mu, x, z)$ for Program \ref{progr:bay-optrc}.  For all $i\in N$, for all $\lambda\in\Lambda$, and for all $\lambda_i^\prime\in\Lambda_i$, it holds that
	\begin{align*}
		\sum_{a\in\A}&x_a^{\lambda, i}(\omega)F_a^\lambda(\omega) - c_i^{\lambda_i}(a_i) \hspace{10cm}\,\\
		&\specialcell{= \sum_{a_i\in A_i}\left(\sum_{a_{-i}\in A_{-i}} F^\lambda_{(a_i, a_{-i})}(\omega) \mu^{\lambda}(a_i, a_{-i}) \pi^{\lambda, i}_{(a_i, a_{-i})}(\omega) - c_i^{\lambda_i}(a_i)\right) \hfill (\textnormal{Definition of } x)} \\
		&\specialcell{= \sum_{a_i\in A_i} U_i^{\mu, \pi}(\lambda_i, a_i\vert\lambda_{-i})\hfill(\textnormal{Definition of } U_i^{\mu, \pi})} \\
		&\specialcell{\ge \sum_{a_i\in A_i}\max_{a_i^\prime\in A_i} U_i^{\mu, \pi}(\lambda_i\to\lambda_i^\prime, a_i\to a_i^\prime\vert\lambda_{-i})\hfill(\textnormal{Feasibility of } (\mu, \pi))} \\
		&\specialcell{=\sum_{a_i\in A_i} z^i(\lambda, \lambda_i^\prime, a_i),\hfill(\textnormal{Definition of $z$})}
	\end{align*}
	which proves that $(\mu, x, z)$ satisfies Constraint \eqref{eq:baylprc_ic} of Program \ref{prog:baylprc}. Moreover, notice that for all $i\in N$, for all $\lambda\in\Lambda$, for all $\lambda_i^\prime\in\Lambda_i$, and for all $a_i,a_i^\prime\in A_i$ it holds that
	\begin{align*}
		z^i(\lambda, \lambda_i^\prime, a_i) &= \max_{a_i^\star\in A_i} U_i^{\mu, \pi}(\lambda_i\to\lambda_i^\prime, a_i\to a_i^\star\vert\lambda_{-i}) \hspace{6.7cm} (\textnormal{Definition of } z)\\
		&\ge U_i^{\mu, \pi}(\lambda_i\to\lambda_i^\prime, a_i\to a_i^\prime\vert\lambda_{-i}) \\
		&\specialcell{= \hspace{-2mm}\sum_{a_{-i}\in\A_{-i}}\hspace{-2mm} \mu^{(\lambda_i^\prime, \lambda_{-i})}(a_i, a_{-i})\left(\sum_{\omega\in\Omega} \pi^{(\lambda_i^\prime, \lambda_{-i}), i}_{(a_i, a_{-i})}(\omega)F^{(\lambda_i, \lambda_{-i})}_{(a_i^\prime, a_{-i})}(\omega) - c_i^{\lambda_i}(a_i^\prime)\right) \hfill(\textnormal{Definition of } U_i^{\mu, \pi})} \\
		&\specialcell{= \sum_{a_{-i}\in\A_{-i}} \sum_{\omega\in\Omega} x^{(\lambda_i^\prime, \lambda_{-i}), i}(\omega) F^{(\lambda_i, \lambda_{-i})}_{(a_i^\prime, a_{-i})}(\omega) - c_i^{\lambda_i}(a_i^\prime), \hfill (\textnormal{Definition of } x)}
	\end{align*}
	which proves that also Constraint~\eqref{eq:baylprc_max} of Program~\ref{prog:baylprc} is satisfied by $(\mu, x, z)$.
	
	Additionally, notice that Constraints~\eqref{eq:baylprc_bound} and \eqref{eq:baylprc_defmu} of $\baylp(\BAindR)$ are trivially satisfied by $\mu$ and $x$. Similarly, by definition of $\BAindR$, we have that for every $\lambda\in\Lambda$ and for every $a\in\BAindR(\lambda)$ it holds $\mu^\lambda(a) = 0$. Thus, $x_a^{\lambda, i}(\omega) = 0$ for all $i\in N$ and $\omega\in\Omega$. This proves that also Constraints~\eqref{eq:baylprc_restr1} and \eqref{eq:baylprc_restr2} are satisfied by $(\mu, x, z)$. 
	Thus, we can conclude that $(\mu, x, z)$ is feasible to $\baylp(\BAindR)$. This concludes the proof.
\end{proof}

Now, we can combine the results of \cref{le:lemmatransbay} and \cref{le:lemmarelaxationbay} to obtain our main result of this section. 

\thbayoptcrc*
\begin{proof}
	Let $(\mu, x, z)$ be an optimal solution to $\baylp(M, \A)$, where $M=M(|I|,\varepsilon, K)$ is defined as per \Cref{le:lemmarelaxationbay} and $K=\tau({|I|})$. Note that $(\mu, x, z)$ can be find in time $\poly(|I|, \log(M))=\poly(|I|, \varepsilon^{-1})$ solving the linear program $\lprc(M, \A)$. Indeed, it is well known that there exists a solution $( \mu,  x, z)$ such that $\lVert  x\rVert_\infty\le \tau(|I|)$, where $\tau:\mathbb{N}\to\mathbb{R}^+$ is the same function of \Cref{le:lemmatransbay} such that $\tau (x)\in 2^{\poly(x)}$, and this solution can be found in polynomial time \citep{bertsimas1997introduction}.

	By using \Cref{le:lemmarelaxationbay} and \Cref{le:lemmatransbay}, we can obtain the following inequalities:
	\begin{align}
		\opt_{\baylp(M, \A)}&\ge \opt_{\baylp(M, \BAindR)}\tag{$\BAindR(\lambda)\subseteq\A\,\,\forall\lambda\in\Lambda$}\\
		&\ge \opt_{\baylp(\BAindR)}-\varepsilon\tag{\Cref{le:lemmatrans}}\\
		&\ge \boptcrc-\varepsilon\tag{\Cref{le:lemmarelax}}.
	\end{align}
	
	We used \Cref{le:lemmatransbay} with $K=\tau(|I|)$ as we used it for the optimal solution of $\lprc(\BAindR)$, which is guaranteed to be bounded by $K=\tau(|I|)$.
	
	Then,  define $(\mu, \pi)$ from $(\mu, x, z)$, where:
	\[
	\pi^{\lambda, i}_a(\omega) \coloneq \begin{cases}
		\frac{x^{\lambda, i}_a(\omega)}{\mu^\lambda(a)} & \textnormal{if } \mu^\lambda(a) > 0 \\
		0 & \textnormal{otherwise}
	\end{cases}\quad\forall i\in N,\forall\lambda\in\Lambda,\forall a\in\A,\forall\omega\in\Omega.
	\]
	It follows that:
	\begin{align*}
		\opt_{\baylp(M, \A)} &\coloneq \sum_{\lambda\in\Lambda}\sum_{a\in\A}\sum_{\omega\in\Omega}G(\lambda)F^{\lambda}_a(\omega)\left(\mu^\lambda(a)r_\omega - \sum_{i\in N} x_a^{\lambda, i}(\omega)\right) \\ 
		&= \sum_{\lambda\in\Lambda}\sum_{a\in\A}\sum_{\omega\in\Omega}G(\lambda)F^{\lambda}_a(\omega)\mu^\lambda(a)\left(r_\omega - \sum_{i\in N} \pi_a^{\lambda, i}(\omega)\right).
	\end{align*}
	This, together with the previous inequalities, shows that:
	\[
	\sum_{\lambda\in\Lambda}\sum_{a\in\A}\sum_{\omega\in\Omega}G(\lambda) \mu^\lambda(a) F_a^\lambda(\omega) \left( r_\omega - \sum_{i\in N} \pi^{\lambda, i}_a(\omega)\right)\ge \boptcrc-\varepsilon.
	\]
	
	Finally, note that $(\mu,\pi)$ is feasible for Program \ref{progr:bay-optrc}, thanks to the fact that when $\mu^\lambda(a)=0$ then $x_a^{\lambda, i}(\omega)=0$ for all $\lambda\in\Lambda$, $i\in N$ and $\omega\in\Omega$. This guarantees that $x^{\lambda, i}_a(\omega)=\mu^\lambda(a)\pi^{\lambda, i}_a(\omega) $ for each $\lambda\in\Lambda$, $i \in N$, $a \in \A$, and $\omega \in \Omega$. 
\end{proof}

%% file: src/app_approx.tex
\section{Proofs Omitted from Section \ref{sec:approx}}

\propbaysw*

\begin{proof}
	Consider the instance of the Bayesian principal-multi-agent problem with $n=2$ agents, in which each agent $i$ has two possible types $\Lambda_i=\{\lambdaone, \lambdatwo\}$, with $\lambdaone=1$, $\lambdatwo=3/\alpha$.  Distributions $G_i$ are defined for all $i\in N$ as follows:
	\[
	G_i(\lambdaone) = \frac{\alpha^2/9}{1 - \alpha/3 + \alpha^2/9}\quad\quad\textnormal{and}\quad\quad G_i(\lambdatwo) = \frac{1 - \alpha/3}{1 - \alpha/3 + \alpha^2/9}.
	\]
	The agents' types are drawn independently, hence the distribution $G$ is defined such that $G(\lambda) = G_1(\lambda_1)G_2(\lambda_2)$ for all $\lambda\in\Lambda$.
	For each agent, the action set is composed by two distinct actions $A_i=\{\aone,\atwo\}$, with costs $c_i^{\lambda_i}(\aone) = \lambda_i\alpha/3$ and $c_i^{\lambda_i}(\atwo)=0$, for each $\lambda_i\in \Lambda_i$. The set of outcomes is $\Omega=\{\omega_1, \omega_2\}$, with $r_{\omega_1} = 1$ and $r_{\omega_2}=0$. Finally, the outcome probabilities are independent from the agents' types and are defined as follows: 
	\begin{table}[H]
		\centering
		\begin{tabular}{|c | c c |}
			\hline
			$F_a(\omega)$ & $\omega_1$ & $\omega_2$  \\
			\hline
			$\aone\aone$ & 1 & 0\\
			$\aone\atwo$ & $\alpha/3$ & $1-\alpha/3$  \\
			$\atwo\aone$ & $\alpha/3$ & $1-\alpha/3$  \\
			$\atwo\atwo$ & 0 & 1 \\
			\hline
		\end{tabular}
	\end{table}
	We prove the proposition in two steps: first we show that $\vsw^{1/\alpha} > 0$, and then we show that $\boptdc = 0$. 
	
	\paragraph{Virtual social welfare.} By the definition of virtual social welfare, we have that $\vsw^{1/\alpha} = \mathbb{E}_{\lambda\sim G} \vsw^{1/\alpha}(\lambda)$, where $\vsw^{1/\alpha}(\lambda) = \max_{a\in\A}\left\{\sum_{\omega\in\Omega}F_a(\omega)r_\omega - \sum_{i\in N}c_i^{\lambda_i}(a_i)/\alpha\right\}\geq 0$. Then, it is enough to notice that $G(\lambdaone\lambdaone) > 0$ and
	\begin{align*}
		\vsw^{1/\alpha}(\lambdaone,\lambdaone) &= \max\left\{\frac{1}{3}, \frac{\alpha - 1}{3}, 0\right\} = \frac{1}{3} > 0 
	\end{align*}
	which implies that $\vsw^{1/\alpha} > 0$. 
	
	\paragraph{Optimal deterministic contract.} Now, we move to analyze the value of $\boptdc$. For each $\lambda\in\Lambda$, let $(\mathfrak{a}, p)$ be an optimal deterministic contract (\emph{i.e.,} an optimal solution to Program \ref{progr:bay-optdc}) and, for each $\lambda\in\Lambda$, let $\boptdc(\lambda)$ be the expected utility for the principal when they commit to $(\mathfrak{a}, p)$ and the profile of types is $\lambda$, \emph{i.e.,}
	\[
	\boptdc(\lambda) = \sum_{\omega\in\Omega} F_{a^\lambda}(\omega)\left(r_\omega - \sum_{i\in N}p^{\lambda, i}(\omega)\right).
	\]
	Then, we can write $\boptdc = \mathbb{E}_{\lambda\sim G} \boptdc(\lambda)$. By individual rationality, it must hold: 
	\[
	\boptdc(\lambda) \leq \sum_{\omega\in\Omega}F_{a^\lambda}(\omega)r_\omega - \sum_{i\in N} c_i^{\lambda_i}(a_i^\lambda) \leq \sw(\lambda),
	\]
	where $\sw(\lambda)$ is the social welfare of the  non-Bayesian principal-multi-agent problem instance $(N, \Omega, (A_i)_{i\in N}, (c_i^{\lambda_i})_{i\in N}, F, r)$ defined by the profile of types $\lambda$. Let us  consider first the instances defined by types $(\lambdaone,\lambdatwo)$, $(\lambdatwo, \lambdaone)$ and $(\lambdatwo,\lambdatwo)$. Notice that
	\begin{align}
		\sw(\lambdaone, \lambdatwo) &= \sw(\lambdatwo, \lambdaone) = \max\left\{-\frac{\alpha}{3}, \frac{\alpha}{3} - 1, 0\right\} = 0 \label{eq:boundeqsw12_propgap}\\
		\sw(\lambdatwo, \lambdatwo) &= \max\left\{-1, \frac{\alpha}{3}-1, 0\right\} = 0,\label{eq:boundeqsw22_propgap}
	\end{align}
	which implies that $\boptdc(\lambdatwo,\lambdaone) \le 0$,  $\boptdc(\lambdaone, \lambdatwo)\le 0$, and $\boptdc(\lambdatwo,\lambdatwo) \le 0$.
	Now, focus on the instance defined by types $(\lambdaone, \lambdaone)$. By individual rationality, we have that
	\begin{align*}
		\boptdc(\lambdaone, \lambdaone) &= \sum_{\omega\in\Omega} F_{a^{(\lambdaone,\lambdaone)}}(\omega)\left(r_\omega - \sum_{i\in N} p^{(\lambdaone, \lambdaone), i}(\omega)\right)\\
		&\leq \sum_{\omega\in\Omega} F_{a^{(\lambdaone,\lambdaone)}}(\omega)r_{\omega} - \sum_{i\in N} c_i^{\lambdaone}(a_i) \\
		&= 
		\begin{cases}
			1 - \frac{2}{3}\alpha & \textnormal{if } a^{(\lambdaone,\lambdaone)}=(\aone,\aone) \\
			0 & \text{otherwise.}
		\end{cases}
	\end{align*}

	Thus, if $a^{(\lambdaone,\lambdaone)}\neq(\aone, \aone)$, we would have that
	\[
	\boptdc = \mathbb{E}_{\lambda\sim G}\boptdc(\lambda) \le 0,
	\]
	which gives the desired result. 
	
	Assume, instead, $a^{(\lambdaone,\lambdaone)}=(\aone, \aone)$. By Constraint \eqref{eq:bayoptdc_ic} of Program \ref{progr:bay-optdc}, we have that, for all $i\in N$, it holds
	\begin{align*}
		\sum_{\omega\in\Omega} F_{(\aone, \aone)}(\omega) p^{(\lambdaone,\lambdaone), i}(\omega) - c_i^{\lambdaone}(\aone) &= p^{(\lambdaone,\lambdaone), i}(\omega_1) - \frac{\alpha}{3} \\
		&\geq \sum_{\omega\in\Omega} F_{(\atwo, \aone)}(\omega) p^{(\lambdaone,\lambdaone), i}(\omega) - c_i^{\lambdaone}(\atwo) \\
		&= \frac{\alpha}{3} p^{(\lambdaone,\lambdaone), i}(\omega_1) + (1-\frac{\alpha}{3}) p^{(\lambdaone,\lambdaone), i}(\omega_2).
	\end{align*}
	Rearranging and using $p^{(\lambdaone,\lambdaone), i}(\omega_2)\ge 0$, we get 
	\begin{equation}\label{eq:lowerboundpay1_bay}
		p^{(\lambdaone,\lambdaone), i}(\omega_1) \ge \frac{\alpha/3}{1-\alpha/3}\quad\forall i\in N.
	\end{equation}
	Moreover, Constraint~\eqref{eq:bayoptdc_ic} guarantees that, for all $i\in N$, it holds:
	\begin{align}
		\sum_{\omega\in\Omega}& F_{a^{(\lambdatwo,\lambdaone)}}(\omega)p^{(\lambdatwo,\lambdaone), 1}(\omega) - c_1^{\lambdatwo}(a_i^{\lambdatwo})\ge \sum_{\omega\in\Omega} F_{(\atwo, \aone)}(\omega)p^{(\lambdaone,\lambdaone), 1}(\omega) - c_1^{\lambdatwo}(\atwo)\notag\\
		&= \frac{\alpha}{3} p^{(\lambdaone,\lambdaone), 1}(\omega_1) + \left(1 - \frac{\alpha}{3}\right) p^{(\lambdaone,\lambdaone), 1}(\omega_2) \notag \\
		&\ge \frac{\alpha}{3} p^{(\lambdaone,\lambdaone), 1}(\omega_1) \notag\\
		&\ge \frac{\alpha^2/9}{1-\alpha/3}, \label{eq:boundsw1_bayprop}
	\end{align}
	where we used $p^{(\lambdaone,\lambdaone), i}(\omega_2)\ge 0$ and \cref{eq:lowerboundpay1_bay}. 
	In a similar way, we can obtain the following inequality:
	\[
	\sum_{\omega\in\Omega} F_{a^{(\lambdaone,\lambdatwo)}}(\omega)p^{(\lambdaone,\lambdatwo), 2}(\omega) - c_2^{\lambdatwo}(a_i^{\lambdatwo}) \ge \frac{\alpha^2/9}{1-\alpha/3}.
	\]	 
	Then, to prove the proposition, it is enough to recover suitable bounds for the values of $\boptdc(\lambda)$ for each $\lambda\in\Lambda$. In particular, notice that
	\begin{align*}
		\boptdc(\lambdatwo, \lambdaone) &= \sum_{\omega\in\Omega}F_{a^{(\lambdatwo, \lambdaone)}}(\omega)\left(r_\omega - \sum_{i\in N}p^{(\lambdatwo,\lambdaone), i}(\omega) \right)\hspace{6cm}\, \\
		&\specialcell{\le \sum_{\omega\in\Omega}F_{a^{(\lambdatwo, \lambdaone)}}(\omega)r_\omega - \sum_{i\in N}c_i^{\lambda_i}(a_i^{(\lambdatwo,\lambdaone)})- \frac{\alpha^2/9}{1-\alpha/3}\hfill (\textnormal{\cref{eq:boundsw1_bayprop} and IR})}\\
		&\specialcell{\le \sw(\lambdaone,\lambdatwo) - \frac{\alpha^2/9}{1-\alpha/3} \hfill (\textnormal{Definition of }\sw)} \\
		&\specialcell{\le - \frac{\alpha^2/9}{1-\alpha/3}.\hfill(\textnormal{\cref{eq:boundeqsw12_propgap}})}
	\end{align*}
	Following similar steps, we can prove that
	\[
	\boptdc(\lambdaone,\lambdatwo)\le - \frac{\alpha^2/9}{1-\alpha/3}.
	\]
	Finally, using that $\boptdc(\lambdaone,\lambdaone)\le 1$ and that $\boptdc(\lambdatwo,\lambdatwo)\le\sw(\lambdatwo,\lambdatwo)\le 0$, we can conclude that
	\begin{align*}
		&\boptdc = \sum_{\lambda\in\Lambda}G(\lambda) \boptdc(\lambda) \\
		&\hspace{2mm}\le \left(\frac{\alpha^2/9}{1 - \alpha/3 + \alpha^2/9}\right)^2 \boptdc(\lambdaone,\lambdaone) + \frac{\alpha^2/9\left(1 - \alpha/3\right)}{(1 - \alpha/3 + \alpha^2/9)^2} \left(\boptdc(\lambdaone,\lambdatwo) + \boptdc(\lambdatwo,\lambdaone)\right) \\
		&\hspace{2mm}\le \frac{\alpha^4/81}{(1 - \alpha/3 + \alpha^2/9)^2} -  2\frac{\alpha^4/81(1-\alpha/3)}{(1 - \alpha/3 + \alpha^2/9)^2\left(1-\alpha/3\right)} \\
		&\hspace{2mm}= -  \frac{\alpha^4/81(1-\alpha/3)}{(1 - \alpha/3 + \alpha^2/9)^2\left(1-\alpha/3\right)} \\
		&\hspace{2mm}\le 0.
	\end{align*}
	Notice that $\boptdc=0$ follows trivially from the fact that the contract that never pays the agents and always recommend action profile $\atwo\atwo$ guarantees utility $0$ to the principal.
	
\end{proof}

\lemmadsicbay*
\begin{proof}
	Fix $\delta>0$.  First, notice that, since there always exists an action profile $a\in\A$ with $0$ costs and $0$ reward for the principal, it always holds that
	\[
	\sum_{\omega\in\Omega} F^\lambda_{a_\delta^\lambda}(\omega)r_\omega - n(1+\delta)\sum_{i\in N}c_i^{\lambda_i}(a_{\delta,i}^\lambda)\ge 0\quad\quad\forall\lambda\in\Lambda,
	\]
	which rearranged yields:
	\begin{equation}\label{eq:equation1lemmadsic}
		\frac{1}{n(1+\delta)} \sum_{\omega\in\Omega} F^\lambda_{a_\delta^\lambda}(\omega)r_\omega \ge \sum_{i\in N}c_i^{\lambda_i}(a_{\delta,i}^\lambda)\quad\quad\forall\lambda\in\Lambda.
	\end{equation}
	
	Then, let us consider the IR constraint. The following holds for all $i\in N$ and for all $\lambda\in\Lambda$:
	\begin{align*}
		\sum_{\omega\in\Omega}F^\lambda_{a_\delta^\lambda}(\omega)p_\delta^{\lambda,i}(\omega) &= \frac{1}{n(1+\delta)}\sum_{\omega\in\Omega}F^\lambda_{a_\delta^\lambda}(\omega)r_\omega - \sum_{j\in N\setminus\{i\}} c_j^{\lambda_j}(a_{\delta,j}^\lambda)\tag{Definition of $p_\delta$} \\
		&\ge \sum_{i\in N}c_i^{\lambda_i}(a_{\delta,i}^\lambda) - \sum_{j\in N\setminus\{i\}} c_j^{\lambda_j}(a_{\delta,j}^\lambda)\tag{\cref{eq:equation1lemmadsic}} \\
		&= c_i^{\lambda_i}(a_{\delta, i}^\lambda),
	\end{align*}
	which proves that $(\mathfrak{a}_\delta, p_\delta)$ satisfies constraint \eqref{eq:bayoptdc_ir}. 
	
	Now, consider the DSIC constraint (\cref{eq:bayoptdc_ic}). Let $\lambda$ be the true type profile, and consider the case in which agent $i$ misreport their type by communicating type $\lambda_i^\prime$ and then plays action $a_i$. Their expected utility after following this behavior satisfies the following: 
	\begin{align*}
		&\sum_{\omega\in\Omega} F^\lambda_{(a_i, a_{\delta,-i}^{(\lambda_i^\prime, \lambda_{-i})})}(\omega)p_\delta^{(\lambda_i^\prime, \lambda_{-i}), i}(\omega) - c_i^{\lambda_i}(a_i) \\
		&\hspace{2mm}= \frac{1}{n(1+\delta)} \sum_{\omega\in\Omega} F^\lambda_{(a_i, a_{\delta,-i}^{(\lambda_i^\prime, \lambda_{-i})})}(\omega)r_\omega - \sum_{j\in N\setminus \{i\}} c_j^{\lambda_j}(a_{\delta,j}^{(\lambda_i^\prime, \lambda_{-i})}) - c_i^{\lambda_i}(a_i) \tag{Definition of $p_\delta$}\\
		&\hspace{2mm}\le \frac{1}{n(1+\delta)} \sum_{\omega\in\Omega} F^\lambda_{a_\delta^\lambda}(\omega)r_\omega - \sum_{j\in N} c_j^{\lambda_j}(a_{\delta,j}^\lambda) \tag{Definition of $a_\delta^\lambda$} \\
		&\hspace{2mm}=\sum_{\omega\in\Omega} F^\lambda_{a_\delta^\lambda}(\omega) p_\delta^{\lambda,i}(\omega) - c_i^{\lambda_i}(a_{\delta,i}^\lambda)\tag{Definition of $p_\delta$}.
	\end{align*}
	This proves that $(\mathfrak{a}, p)$ satisfies also Constraint \eqref{eq:bayoptdc_ic}. This concludes the proof.
\end{proof}

\lemmaapproxvswbay*

\begin{proof}
	Fix $\delta>0$ and $\lambda\in\Lambda$. First, notice that the following holds:
	\begin{align*}
		\sum_{\omega\in\Omega} F^\lambda_{a_\delta^\lambda}(\omega) r_\omega &\ge \sum_{\omega\in\Omega} F^\lambda_{a_\delta^\lambda}(\omega) r_\omega - n(1+\delta)\sum_{i\in N} c_i^{\lambda_i}(a_{\delta, i}^\lambda) \tag{$c\ge0$}\\
		&= \max_{a\in \A}\left\{\sum_{\omega\in\Omega} F^\lambda_{a}(\omega) r_\omega - n(1+\delta)\sum_{i\in N} c_i^{\lambda_i}(a_i)\right\}\tag{Definition of $a_\delta^\lambda$} \\
		&= \vsw^{n(1+\delta)}(\lambda)\tag{Definition of $\vsw$}.
	\end{align*}
	 Then, we can use the previous inequality to conclude that
	\begin{align*}
		\sum_{\omega\in\Omega} F^\lambda_{a_\delta^\lambda}(\omega)\left(r_\omega - \sum_{i\in N}p_\delta^{\lambda, i}(\omega)\right) &= \frac{\delta}{1+\delta}\sum_{\omega\in\Omega} F^\lambda_{a_\delta^\lambda}(\omega)r_\omega + \sum_{i\in N}\sum_{j\in N\setminus\{i\}} c_j^{\lambda_j}(a_{\delta, j}^\lambda) \tag{Definition of $p_\delta$}\\
		&\ge \frac{\delta}{1+\delta}\sum_{\omega\in\Omega} F^\lambda_{a_\delta^\lambda}(\omega)r_\omega\tag{$c>0$} \\
		&\ge \frac{\delta}{1+\delta}\vsw^{n(1+\delta)}(\lambda).
	\end{align*}
	This concludes the proof.
\end{proof}

\proplowerboundbay*
\begin{proof}
	Fix $\alpha>0$ and $\eps\in (0,1)$ and consider the following instance of the Bayesian principal-multi-agent problem. The set of agents is $N=\range{n}$, where $n\coloneqq \left\lceil (2/\eps)^{2/\alpha}\right\rceil$, while the outcomes are $\Omega=\{\omega_1,\omega_2\}$, where $r_{\omega_1} = 1$ and $r_{\omega_2}=0$. The set of types $\Lambda_i$ for each agent $i$ is composed by two discrete types $\lambdaone=0$ and $\lambdatwo=1$ and the joint distribution $G$ is characterized as follows:
	\[
	G(\lambda) = \begin{cases}
		\frac{n^{-\alpha/2}}{\gamma} &\textnormal{if }\lambda_i=\lambdaone,\forall i\in N \\
		\frac{1}{n\gamma} & \textnormal{if } \exists i\in N\textnormal{ s.t. } \lambda_i=\lambdatwo,\wedge\, \lambda_j = \lambdaone,\forall j\in N\setminus\{i\}\\
		0 &\textnormal{otherwise}
	\end{cases}\quad\quad\forall\lambda\in\Lambda,
	\]
	where, for notational convenience, we let $\gamma \coloneqq n^{-\alpha/2} + 1$.
	
	Each agent $i$ can choose between two actions $A_i=\{\aone,\atwo\}$. The outcome distribution is independent from the agents' types and is specified as follows: 
	\[
	F_a(\omega_1) = \begin{cases}
		1 & \textnormal{if } a_i=\aone,\forall i\in N\\
		0 & \textnormal{otherwise}
	\end{cases}\quad\quad\forall a\in\A.
	\]
	Intuitively, the outcome $\omega_1$ guaranteeing positive reward to the principal is induced (deterministically) only if all agents play action $\aone$. The action costs are the following: 
	\[
	c_i^{\lambda_i}(a_i) = \begin{cases}
		\frac{\lambda_i}{2n^{1-\alpha}}&\textnormal{if } a_i=\aone \\
		0&\text{otherwise}
	\end{cases} \quad\quad\forall i\in N,\forall\lambda_i\in\Lambda_i,\forall a_i\in A_i.
	\]
	We prove this result in two steps: first we derive a lower bound for $\vsw^{n^{1-\alpha}}$, then we show that the principal's utility under an optimal deterministic contract is upper bounded by a suitably defined quantity.	
	
	\paragraph{Lower bound for $\vsw^{n^{1-\alpha}}$.}
	Notice that, letting $\bar a\in\A$ be such that $\bar a_i=\aone$ for all $i\in N$ and starting from the definition of virtual social welfare, we can prove the following: 
	\begin{align*}
		\vsw^{n^{1-\alpha}} &\coloneqq \mathbb{E}_{\lambda\sim G}\left[\max_{a\in\A}\left\{\sum_{\omega\in\Omega}F_a(\omega)r_\omega - n^{1-\alpha}\sum_{i\in N}c_i^{\lambda_i}(a_i)\right\}\right]\hspace{6cm} \\
		&\geq \mathbb{E}_{\lambda\sim G}\left[\sum_{\omega\in\Omega}F_{\bar a}(\omega)r_\omega - n^{1-\alpha}\sum_{i\in N}c_i^{\lambda_i}(\bar a_i)\right] \\
		&\specialcell{=1 - \frac{1}{2\gamma} \hfill (\textnormal{Definition of } G)} \\
		&\specialcell{\ge \frac{1}{2\gamma}.\hfill (\gamma \ge 1)}
	\end{align*}
	
	\paragraph{Upper bound for $\nollboptdc$.}
	Let $(\mathfrak{a}, p)$ be an optimal deterministic contract without limited liability (\emph{i.e.,} such that it achieves expected utility $\nollboptdc$ for the principal). With a slight abuse of notation, for each $i\in N$, we denote as $\lambda^{(i)}\in\Lambda$ the profile of types such that $\lambda^{(i)}_i = \lambdatwo$ and $\lambda_j^{(i)} = \lambdaone$ for each $j\in N\setminus \{i\}$. Then, for each $\lambda\in\Lambda$, let $N^\circ\subseteq N$ be the set of agents incentivized to play action $\aone$ when their type is $\lambdatwo$. Formally, \(N^\circ \coloneqq \{i\in N: a_i^{\lambda^{(i)}} = \aone\}\). Additionally, we denote as $\kappa\coloneqq |N^\circ|$ the number of agents belonging to $N^\circ$. 
	Then, the IC and IR constraints guarantee that for each $i\in N^\circ$
	\begin{align}
		\sum_{\omega\in\Omega}F_{a^{(\lambdaone,..,\lambdaone)}}(\omega)p^{(\lambdaone,..,\lambdaone), i}(\omega) &\ge \sum_{\omega\in\Omega} F_{a^{\lambda^{(i)}}}(\omega) p^{\lambda^{(i)}, i}(\omega) \tag{IC} \\
		&\ge c_i^{\lambdatwo}(a_i^{\lambda^{(i)}})\tag{IR} \\
		&= \frac{1}{2n^{1-\alpha}}.\label{eq:boundpay1_proplb} 
	\end{align}
	
	Moreover, notice that IR also implies that for each $i\in N$ and for each $\lambda\in\Lambda$ it holds
	\begin{equation}\label{eq:boundpay2_proplb} 
		\sum_{\omega\in\Omega}F_{a^{\lambda}}(\omega)p^{\lambda, i}(\omega) \ge c_i^{\lambda_i}(a_i^\lambda) \ge 0.
	\end{equation}	
	
	Finally, for each  $j\in N^\circ$, the following chain of inequalities holds: 
	\begin{align}
		\sum_{\omega\in\Omega} F_{a^{\lambda^{(j)}}}(\omega)\left(r_\omega - \sum_{i\in N} p^{\lambda^{(j)}, i}(\omega)\right) &\le 1 - \sum_{\omega\in\Omega} \sum_{i\in N} F_{a^{\lambda^{(j)}}}(\omega) p^{\lambda^{(j)}, i}(\omega) \notag \\
		&\le 1 - \sum_{\omega\in\Omega}F_{a^{\lambda^{(j)}}}(\omega)p^{\lambda^{(j)}, j}(\omega) \notag \\
		&\le 1 - \frac{1}{2n^{1-\alpha}},\label{eq:boundpay3_proplp}
	\end{align}
	where we used $r_\omega\le 1$ for each $\omega \in \Omega$, \cref{eq:boundpay2_proplb} and \cref{eq:boundpay1_proplb}. If $j\notin N^\circ$, the following holds:
	\begin{align}
		\sum_{\omega\in\Omega} F_{a^{\lambda^{(j)}}}(\omega)\left(r_\omega - \sum_{i\in N} p^{\lambda^{(j)}, i}(\omega)\right) &= - \sum_{\omega\in\Omega} \sum_{i\in N} F_{a^{\lambda^{(j)}}}(\omega) p^{\lambda^{(j)}, i}(\omega) \le 0,\label{eq:boundpay4_proplb}
	\end{align}
	where we used the fact that $F_{a^{\lambda^{(j)}}}(\omega_2) = 1$ when $j\notin N^\circ$ and \cref{eq:boundpay2_proplb}.
	
	Then, by definition of $G$, we can write the following:
	\begin{align*}
		\nollboptdc &= \sum_{\lambda\in\Lambda}\sum_{\omega\in\Omega} G(\lambda)F_{a^\lambda}(\omega)\left(r_\omega - \sum_{i\in N} p^{\lambda, i}(\omega) \right)\hspace{6.2cm}\,\\
		&= \underbrace{G(\lambdaone,..,\lambdaone)\sum_{\omega\in\Omega}F_{a^{(\lambdaone,..,\lambdaone)}}(\omega)\left(r_\omega - \sum_{i\in N} p^{(\lambdaone,...,\lambdaone),i}(\omega)\right)}_{(A)} \\
		&\hspace{6mm} + \underbrace{\sum_{j\in N} G(\lambda^{(j)})\sum_{\omega\in\Omega}\sum_{i\in N}F_{a^{\lambda^{(j)}}}(\omega)\left(r_\omega - \sum_{i\in N}p^{\lambda^{(j)}, i}(\omega)\right)}_{(B)}.
	\end{align*}
	We analyze the two terms separately. For (A), it holds that
	\begin{align*}
		(A) &= G(\lambdaone,..,\lambdaone)\sum_{\omega\in\Omega}F_{a^{(\lambdaone,..,\lambdaone)}}(\omega)\left(r_\omega - \sum_{i\in N} p^{(\lambdaone,...,\lambdaone),i}(\omega)\right) \hspace{5.5cm}\, \\
		&\specialcell{\le G(\lambdaone,..,\lambdaone) \left(1 - \sum_{\omega\in\Omega}F_{a^{(\lambdaone,..,\lambdaone)}}(\omega)\sum_{i\in N}p^{(\lambdaone,...,\lambdaone),i}(\omega)\right)\hfill (r\le 1)} \\
		&= G(\lambdaone,..,\lambdaone) \left[1 - \sum_{\omega\in\Omega}F_{a^{(\lambdaone,..,\lambdaone)}}(\omega)\left(\sum_{i\in N^\circ}p^{(\lambdaone,...,\lambdaone),i}(\omega) +\sum_{i\in N\setminus N^\circ}p^{(\lambdaone,...,\lambdaone),i}(\omega)\right)\right]\\
		&\specialcell{\le G(\lambdaone,..,\lambdaone)\left(1 - \frac{\kappa}{2n^{1-\alpha}}\right) \hfill (\textnormal{Eq. \eqref{eq:boundpay1_proplb} and \eqref{eq:boundpay2_proplb})}} \\
		&= \frac{n^{-\alpha/2}}{\gamma}\left(1 - \frac{\kappa}{2n^{1-\alpha}}\right).
	\end{align*}
	For $(B)$, it holds that 
	\begin{align*}
		(B) &= \sum_{j\in N} G(\lambda^{(j)})\sum_{\omega\in\Omega}\sum_{i\in N}F_{a^{\lambda^{(j)}}}(\omega)\left(r_\omega - \sum_{i\in N}p^{\lambda^{(j)}, i}(\omega)\right) \\
		&= \sum_{j\in N^\circ} G(\lambda^{(j)})\sum_{\omega\in\Omega}\sum_{i\in N}F_{a^{\lambda^{(j)}}}(\omega)\left(r_\omega - \sum_{i\in N}p^{\lambda^{(j)}, i}(\omega)\right) \\
		&\hspace{6mm}+ \sum_{j\in N\setminus N^\circ} G(\lambda^{(j)})\sum_{\omega\in\Omega}\sum_{i\in N}F_{a^{\lambda^{(j)}}}(\omega)\left(r_\omega - \sum_{i\in N}p^{\lambda^{(j)}, i}(\omega)\right) \\
		&\le \sum_{j\in N^\circ} G(\lambda^{(j)})\left(1 - \frac{1}{2n^{1-\alpha}}\right) \tag{Eq. \eqref{eq:boundpay3_proplp} and \eqref{eq:boundpay4_proplb}} \\
		&= \frac{\kappa}{n\gamma} \left(1 - \frac{1}{2n^{1-\alpha}}\right).
	\end{align*}
	Putting all together, we get
	\[
	\nollboptdc = (A) + (B)\le \frac{n^{-\alpha/2}}{\gamma}\left(1 - \frac{\kappa}{2n^{1-\alpha}}\right) + \frac{\kappa}{n\gamma} \left(1 - \frac{1}{2n^{1-\alpha}}\right).
	\]
	Taking the derivative with respect to $\kappa$, we get the following function: 
	\[
	-\frac{1}{2\gamma n^{1-\alpha/2}} + \frac{1}{n\gamma} - \frac{1}{2\gamma n^{2-\alpha}} \le \frac{1}{n\gamma} - \frac{1}{2\gamma n^{1-\alpha/2}} \le 0,
	\]
	where the last inequality follows from $n\ge (2/\eps)^{2/\alpha}\ge 2^{2/\alpha}$ 
	%\fe{tbc (è vero solo con $eps<1$, se vogliamo lasciare lo statement per $\eps>0$ generico come prop. 5.3 allora va ridefinito  $n=\left\lceil \max\{2, 2/\eps\}^{2/\alpha}\right\rceil$)}. 
	Then, since the derivative with respect to $\kappa$ is negative, the upper bound on \nollboptdc is maximized for $k=0$, yielding: 
	\[
	\nollboptdc \le \frac{n^{-\alpha/2}}{\gamma}.
	\]
	
	Thus, we can conclude that 
	\[
	\frac{\nollboptdc}{\vsw^{n^{1-\alpha}}} \le \frac{n^{-\alpha/2}}{\gamma}2\gamma = 2n^{-\alpha/2} \le \eps,
	\]
	as desired.	
\end{proof}

\thsocialwelfarebayesian*

\begin{proof}
	First, notice that for each $\lambda\in\Lambda$ and $\alpha>0$, the following holds:
	\begin{align}
		\vsw^{\alpha}(\lambda) &= \max_{a\in\A}\left\{\sum_{\omega\in\Omega}F_a(\omega) r_\omega - \alpha \sum_{i\in N}c_i^{\lambda_i}(a_i)\right\} \nonumber\\
		&=  \max_{a\in\A}\left\{\sum_{\omega\in\Omega}F_a(\omega) r_\omega - \sum_{i\in N}c_i^{\alpha\lambda_i}(a_i)\right\} \nonumber\\
		&= \sw(\alpha\lambda),\label{eq:thsocialwelfarebay1}
	\end{align}
	where we used the definition of $\vsw$, the fact that, since the types are single-dimensional, $\alpha c_i^{\lambda_i}(a_i) = c_i^{\alpha\lambda_i}(a_i)$, and the definition of $\sw$.
	Moreover, notice that for all $\alpha>0$ the following holds (see \emph{e.g.,} \citep[Theorem~12.6]{jacod2004probability}):
	\begin{align}
		\int_{\Lambda}\sw(\alpha\lambda)g(\lambda)d\lambda = \frac{1}{(\alpha)^n}\int_{\alpha\Lambda} \sw(\lambda)g\left(\frac{\lambda}{\alpha}\right)d\lambda.\label{eq:thsocialwelfarebay2}
	\end{align}

	Then, we can write the following chain of inequalities:
	\begin{align*}
		\mathbb{E}_{\lambda\sim G}&\left[\sum_{\omega\in\Omega} F_{a_\delta^\lambda}(\omega)\left(r_\omega - \sum_{i\in N}p_\delta^{\lambda, i}(\omega)\right)\right] \ge \frac{\delta}{1+\delta}\vsw^{n(1+\delta)} \hspace{3.8cm}(\textnormal{\cref{th:thapproxvswbay}})\\
		&= \frac{\delta}{1+\delta}\int_{\Lambda}\vsw^{n(1+\delta)}(\lambda)g(\lambda) d\lambda \\
		&\specialcell{= \frac{\delta}{1+\delta}\int_{\Lambda}\sw(n(1+\delta)\lambda)g(\lambda)d\lambda \hfill (\textnormal{\cref{eq:thsocialwelfarebay1}})}\\
		&\specialcell{=\frac{\delta}{1+\delta}\frac{1}{[n(1+\delta)]^n}\int_{n(1+\delta)\Lambda}\sw(\lambda)g\left(\frac{\lambda}{n(1+\delta)}\right)d\lambda\hfill(\textnormal{\cref{eq:thsocialwelfarebay2}})} \\
		&\specialcell{\ge \frac{\delta}{1+\delta}\frac{1}{[n(1+\delta)]^n}\int_{n(1+\delta)\Lambda}\sw(\lambda)g(\lambda)d\lambda\hfill(g \textnormal{ non-increasing})}	\\
		&= \frac{\delta}{1+\delta}\frac{1}{[n(1+\delta)]^n}\left(x\int_{n(1+\delta)\Lambda\cap \Lambda}\sw(\lambda)g(\lambda)d\lambda + \int_{n(1+\delta)\Lambda\setminus\Lambda}\sw(\lambda)g(\lambda)d\lambda\right) \\
		&\specialcell{\ge \frac{\delta}{1+\delta}\frac{1}{[n(1+\delta)]^n} \int_{n(1+\delta)\Lambda\cap \Lambda}\sw(\lambda)g(\lambda)d\lambda \hfill (\sw(\lambda) > 0)} \\
		&\specialcell{\ge \eta \frac{\delta}{1+\delta}\frac{1}{[n(1+\delta)]^n} \sw \hfill (\textnormal{Small-tail assumption})}.
	\end{align*}
	This concludes the proof.
\end{proof}